\documentclass[11pt]{article}
\usepackage[margin=1in]{geometry}
\usepackage{amsmath}
\usepackage{graphicx}
\usepackage{enumerate}
\usepackage{natbib}
\usepackage{url} 
\usepackage{mathtools}
\usepackage{empheq}
\usepackage{amsfonts}
\usepackage{amsgen,amsmath,amstext,amsbsy,amsopn,amssymb,amsthm,stmaryrd,bbm,subfigure}
\usepackage{cprotect}
\usepackage{natbib}
\usepackage[dvipsnames]{xcolor}
\usepackage{comment}
\usepackage{enumitem}
\usepackage{multirow,booktabs,makecell}
\usepackage[affil-it]{authblk}
\usepackage[ruled,vlined,longend]{algorithm2e}
\usepackage{hyperref}
\newtheorem{thm}{Theorem}
\newtheorem{lem}{Lemma}
\newtheorem{rmk}{Remark}
\newtheorem{cor}{Corollary}
\newtheorem{prop}{Proposition}
\newtheorem{claim}{Claim}
\newtheorem{assump}{Assumption}
\newtheorem{cond}{Condition}
\newcommand{\bbR}{\mathbb{R}}
\newcommand{\bbE}{\mathbb{E}}
\newcommand{\bbP}{\mathbb{P}}
\newcommand{\be}{\boldsymbol{e}}
\newcommand{\bu}{\boldsymbol{u}}
\newcommand{\bv}{\boldsymbol{v}}
\newcommand{\bw}{\boldsymbol{w}}
\newcommand{\bz}{\boldsymbol{z}}
\newcommand{\bU}{\boldsymbol{U}}
\newcommand{\bV}{\boldsymbol{V}}
\newcommand{\bW}{\boldsymbol{W}}
\newcommand{\bX}{\boldsymbol{X}}

\newcommand{\bE}{\boldsymbol{E}}
\newcommand{\bD}{\boldsymbol{D}}
\newcommand{\cX}{\mathcal{X}}
\newcommand{\cY}{\mathcal{Y}}
\newcommand{\cE}{\mathcal{E}}
\newcommand{\cO}{\mathcal{O}}
\newcommand{\op}{\mathrm{op}}

\newcommand{\RNum}[1]{\uppercase\expandafter{\romannumeral #1\relax}}

\newcommand{\vertiii}[1]{{\left\vert\kern-0.25ex\left\vert\kern-0.25ex\left\vert #1 
		\right\vert\kern-0.25ex\right\vert\kern-0.25ex\right\vert}}

\begin{document}
\title{Joint Semi-Symmetric Tensor PCA for Integrating Multi-modal Populations of Networks}
\author[1]{Jiaming Liu$^\dagger$}
\author[2]{Lili Zheng$^\dagger$\thanks{Corresponding author: lili.zheng@rice.edu}}
\author[3]{Zhengwu Zhang}
\author[1,2,4]{Genevera I. Allen}

\affil[1]{Department of Statistics, Rice University}
\affil[2]{Department of Electrical and Computer Engineering, Rice University}
\affil[3]{Department of Statistics and Operations Research, University of North Carolina at Chapel Hill}
\affil[4]{Neurological Research Institue, Baylor College of Medicine}

\date{}
\maketitle
\begin{abstract}
    Multi-modal populations of networks arise in many scenarios including in large-scale multi-modal neuroimaging studies that capture both functional and structural neuroimaging data for thousands of subjects. A major research question in such studies is how functional and structural brain connectivity are related and how they vary across the population. we develop a novel PCA-type framework for integrating multi-modal undirected networks measured on many subjects. Specifically, we arrange these networks as semi-symmetric tensors, where each tensor slice is a symmetric matrix representing a network from an individual subject. We then propose a novel Joint, Integrative Semi-Symmetric Tensor PCA (JisstPCA) model, associated with an efficient iterative algorithm, for jointly finding low-rank representations of two or more networks across the same population of subjects. We establish one-step statistical convergence of our separate low-rank network factors as well as the shared population factors to the true factors, with finite sample statistical error bounds. Through simulation studies and a real data example for integrating multi-subject functional and structural brain connectivity, we illustrate the advantages of our method for finding joint low-rank structures in multi-modal populations of networks.  
\end{abstract}
\noindent%
\footnotetext{$\dagger$: Equal contribution.}

\section{Introduction}\label{intro}

Large-scale network data arises frequently from a wide range of biomedical and social science applications, such as brain connectomes \citep{yao2015review,rubinov2010complex,bullmore2009complex},
social networks \citep{hoff2002latent}, and gene co-expression networks \citep{li2011integrative}. Many of these applications also have different types of networks measured on the same set of nodes from the same subject, called multi-modal networks. For example in neuroimaging, we may have functional connectivity derived from functional magnetic resonance imaging (fMRI) data and structural connectivity derived from diffusion MRI (dMRI) data measured for the same subject \citep{yao2015review,cole2021surface}. We can have a large population of these multi-modal brain networks; the Human Connectome Project is an example of one such multi-modal population study \citep{van2013wu}. Although there are many developed techniques for analyzing multi-modal or multi-view networks \citep{han2015consistent,d2019latent,gao2022testing} and separately populations of networks \citep{paul2020random,macdonald2022latent}, 
there is limited work that can analyze both aspects simultaneously. Some existing works, e.g. \cite{murden2022interpretive}, vectorize networks into vectors and ignore intrinsic structures in networks. This raises the question: can we jointly analyze and perform dimension reduction for multi-modal populations of networks? Such a joint analysis has many benefits, as it allows researchers to extract shared structures and relationships between different network modalities, reflect commonalities and variations amongst subjects, detect outliers, and identify clustering patterns within the population. In the example of multi-modal population neuroimaging studies, this approach can discover connectivity patterns shared between function and structure across many subjects, paving the way to demystifying how the brain works. Moreover, such analysis can highlight how these connections vary across different populations and how they relate to demographic or genomic traits.  Motivated by these applications, we propose to structure multi-modal populations of networks as tensors and develop a novel dimension reduction approach: Joint-Integrative Semi-Symmetric Tensor PCA (JisstPCA).

\subsection{Joint, Integrative Semi-Symmetric Tensor PCA}

Tensors are natural tools for modeling a collection of networks \citep{wu2019tensor,jing2021community,zhang2020flexible}, 
as one can stack the matrix representation of networks along an extra mode. We specifically take inspiration from recent works \citep{weylandt2022multivariate,zhang2019tensor} on semi-symmetric tensor modeling of populations of networks, where each slice of the tensor is a symmetric positive semi-definite matrix representing a network, e.g. an adjacency or Laplacian matrix for undirected networks. This modeling framework and its associated tensor PCA algorithms \citep{weylandt2022multivariate,zhang2019tensor} can help extract principal network factors across the whole population and achieve simultaneous dimension reduction for both the networks and the population. For our goal of analyzing multi-modal populations of networks, we consider the natural idea of integrating multiple semi-symmetric tensors.

For simplicity, we focus on the case with two modalities of networks, while it is straightforward to extend our framework and algorithms to more general cases. We can arrange the two modalities of networks into tensors $\cX\in \bbR^{p\times p\times N}$ and $\cY\in \bbR^{q\times q\times N}$, where $p,\,q$ are the network dimensions, and $N$ is the sample size. We allow the network sizes to be different ($p\neq q$) to account for possibly varying resolutions across modalities. In order to capture the shared low-rank structures from the two tensors, we consider the following joint, integrative semi-symmetric tensor PCA (JisstPCA) model:
\begin{equation}\label{eq:sst_d}
\mathcal{X} = \sum\limits_{k = 1}^{K} d^{*}_{x,k} \cdot \boldsymbol{V}_{k}^* \boldsymbol{V}_{k}^{*\prime} \circ \boldsymbol{u}_{k}^* + \mathcal{E}_{x},\quad\mathcal{Y} = \sum\limits_{k = 1}^{K} d^{*}_{y,k} \cdot \boldsymbol{W}_{k}^* \boldsymbol{W}_{k}^{*\prime} \circ \boldsymbol{u}_{k}^* + \mathcal{E}_{y}.
\end{equation}
Here, both tensors are decomposed into $K$ network factors plus observational noise: $\cE_x\in \bbR^{p\times p\times N}$, $\cE_y\in\bbR^{q\times q\times N}$ are random zero-mean semi-symmetric noise tensors. Note that we do not assume $\cX$ and $\cY$ to be binary so that our approach generalizes to weighted networks.  The orthogonal matrices $\boldsymbol{V}_{k}^* \in \mathcal{O}_{p, r_{x,k}},\, \boldsymbol{W}^*_{k} \in \mathcal{O}_{q, r_{y,k}}$ are the $k$th principal network factors for the two modalities respectively, and represent the major network patterns.  The unit vector $\bu_k^*\in \mathbb{S}^{N - 1}$ is the joint population factor, indicating the weight of the $k$th network factor for each subject; this is similar to the sample loading in classical PCA and represents population-level patterns affiliated with each network pattern.  One can view $(\bV_k^*\bV_k^{*\top}, \bW_k^*\bW_k^{*\top})$ as a pair of network prototypes for the two modalities, and for each subject $i\in [N]$, the corresponding networks $\cX_{:,:,i}$, $\cY_{:,;,i}$ are noisy realizations of a mixture of the $K$ prototypes with weights $(\bu_{1,i}^*,\dots,\bu_{K,i}^*)$. To ensure identifiability, we assume the factors across different layers are linearly independent\footnote{$\mathrm{rank}\left(\boldsymbol{V}_{i}^*, \boldsymbol{V}_{j}^* \right) = r_{x,i} + r_{x,j}, \mathrm{rank}\left(\boldsymbol{W}^*_{i}, \boldsymbol{W}^*_{j}\right) = r_{y,i} + r_{y,j}$ and $\mathrm{rank}\left(\boldsymbol{u}_{i}^*, \boldsymbol{u}_{j}^*\right) = 2$ for $i \neq j = 1, \cdots, K$. }. Our goal is to {\em find the ground truth factors $\bV_k^*,\,\bW_k^*,\,\bu_k^*$ for $k\in[K]$ from noisy network data $\cX$ and $\cY$}. Our main contributions and the paper organization are summarized as follows. 

\paragraph{Main contributions:} We propose the first framework for simultaneously analyzing multi-modal populations of networks: the joint, integrative semi-symmetric tensor PCA (JisstPCA) model \eqref{eq:sst_d}, with associated algorithms for extracting the population and network factors. Our JisstPCA algorithm, introduced in Section \ref{sec:JisstPCA}, consists of efficient joint power iteration and sequential deflation schemes. We additionally propose and study several important extensions of our basic JisstPCA model \eqref{eq:sst_d}, including integration of networks and vector-valued covariates, and a more general model with multiple eigenvalues associated with each factor. In Section \ref{sec:theory}, we prove the first statistical theory for integrative tensor PCA; under the single-factor case, we establish one-step convergence for both individual and joint factors to a small neighborhood of the true factors with corresponding statistical errors under provable initialization conditions. Our theoretical guarantees improve upon or are comparable to prior works for single tensor PCA \citep{weylandt2022multivariate,zhang2018tensor}. Furthermore, we validate our algorithms via empirical studies in Section \ref{sec:simulation} and demonstrate its superior performance compared to baseline methods. The real data study in Section \ref{real_data} showcases how JisstPCA can derive insightful connections between functional and structural networks in human brain and detect outliers in the population.

\subsection{Related Works}\label{sec:relatedwork}

Here, we discuss prior works that most closely relate to ours; a more detailed literature review is included in the Appendix. There has been an extensive literature on tensor decompositions and tensor PCA \citep{kolda2009tensor}, covering various tensor low-rank structures including the CP decomposition \citep{carroll1970analysis,anandkumar2014guaranteed,han2022tensor}, Tucker decomposition \citep{de2000best,zhang2018tensor,luo2021sharp}, and tensor-train low-rank structures \citep{zhou2022optimal}; many of these recent works have also established statistical consistency results. Our JisstPCA framework \eqref{eq:sst_d}, however, is built upon the prior literature on semi-symmetric tensor PCA \citep{weylandt2022multivariate,zhang2019tensor,winter2020multi} for analyzing populations of networks. 
Amongst these, our basic model \eqref{eq:sst_d} is most similar to, and reduces to the semi-symmetric tensor PCA model in \citep{weylandt2022multivariate} when focusing on one modality of networks. 
Different from our work and \cite{weylandt2022multivariate}, however, \cite{zhang2019tensor,winter2020multi,wang2014canonical} utilize the CP decomposition, a sum of rank-1 tensors, to model populations of networks. Although the CP decomposition is widely studied and there also exists joint factorization methods for integrating multiple CP low-rank tensors \citep{acar2011all, acar2014flexible, wu2018ctf, schenker2020flexible, lu2020exploring, farias2016exploring, fu2015joint}, it is not the most appropriate model for our analysis of multi-modal populations of networks. First, rank-one factors cannot sufficiently capture complex network patterns, hence limiting the expressivity of the model.  Second, standard CP decompositions cannot enforce symmetry in the factors.  Finally, although one can write our model as a special CP decomposition with the same $\bu_k^*$ factors repeated $r_k$ times, existing joint CP decomposition algorithms are inapplicable of capturing this model as they often require an angle /incoherence condition between any pair of the factors in all modes \citep{anandkumar2014guaranteed}.

Our model \eqref{eq:sst_d} is also closely related to the Tucker low-rank model. When constraining all factors $\bV_1,\dots,\bV_K$ and $\bW_1,\dots,\bW_K$ in \eqref{eq:sst_d} to be mutually orthogonal, our model is a special case of the Tucker model with shared factors on the third mode across modalities. However, the mutual orthogonality constraints in Tucker models can be especially restrictive for network factors, again limiting the expressivity of the model. The Tucker core also complicates the one-to-one correspondence between the sample factor $\bu_k$ and the network factors $\bV_k$ and $\bW_k$, which sacrifices interpretability. 

Further, our work builds upon the extensive data integration literature. Most existing data integration methods focus on tabular data that can be arranged as matrices, including the JIVE \citep{lock2013joint}, the iPCA \citep{tang2021integrated}, the multi-block PCA family \citep{abdi2013multiple,westerhuis1998analysis}, and many others. 
Some other works \citep{acar2011all, acar2014flexible, wu2018ctf, schenker2020flexible} consider joint factorizations for tensors and tensor-matrix integration, but are based on the CP decomposition and are less appropriate for joint network analysis as discussed earlier. Finally, to the best of our knowledge, our work also provides the first theoretical guarantees for integrative tensor PCA. 
\section{JisstPCA Algorithms}\label{sec:JisstPCA}
In this section, we introduce our JisstPCA algorithm for finding the tensor factors in \eqref{eq:sst_d} and some extensions of it. We begin with the introduction of some notations.

\subsection{Notation and Preliminary Tensor Algebra}\label{sec:notation}

We first introduce notation that will be used frequently in this paper; we closely follow \citet{kolda2009tensor}'s notation for tensor algebra and refer the reader here for further details.  
We denote tensors as $\mathcal{X}$, matrices as $\mathbf{X}$, vectors as $\mathbf{x}$, and scalars as $x$. Matricization of $\mathcal{X}$ along the $k^{th}$ mode is denoted as $\mathcal{M}_{k}(\mathcal{X})$; multiplication of $\mathcal{X}$ with a matrix along the first mode is denoted as $\times_1$ with $\times_2$ and $\times_3$ defined similarly; $\langle \mathcal{X}, \mathcal{Y} \rangle$ denotes the inner (trace) product; $\circ$ denotes the outer product; for $\cX\in \bbR^{p_1\times \cdots\times p_d}$, $\|\mathcal{X}\|_{F}$ is the tensor Frobenius norm and $\|\cX\|_{\op} = \sup_{u_1\in \mathbb{S}^{p_1-1}\,\dots,u_d\in \mathbb{S}^{p_d-1}}\cX\times_1u_1\times_2\cdots\times_du_d$ is the tensor operator norm. We similarly use $\|\mathbf{X}\|_{\op} = \sup_{u:\|u\|_2=1}\|\mathbf{X}u\|_2$ to denote the matrix operator norm. 
For a tensor $\mathcal{X} \in \mathbb{R}^{p \times p \times N}$, we say it is semi-symmetric if its $N$ slices are all $p \times p$ symmetric matrices; the trace product of $\mathcal{X}$ and a matrix $\boldsymbol{V} \in \mathbb{R}^{p \times r}$ is denoted by $[\mathcal{X}; \boldsymbol{V}]\in \bbR^{N}$, whose $k^{th}$ element is $\langle \mathcal{X}_{:,:,k}, \boldsymbol{V}\boldsymbol{V}^{\prime} \rangle = \text{Tr}(\boldsymbol{V}^{\prime} \mathcal{X}_{:,:,k} \boldsymbol{V})$. Similarly, for a diagonal matrix $\bD\in \bbR^{r\times r}$, we let $[\mathcal{X}; \boldsymbol{V}, \mathbf{D}]\in \bbR^{N}$ denote a vector whose $k^{th}$ element is $\langle \mathcal{X}_{\cdot \cdot k}, \boldsymbol{V}\bD\boldsymbol{V}^{\prime} \rangle$.  Additionally, we let $\sin\theta(\mathbf{u}_*,\mathbf{u}) = \sin(\arccos \mathbf{u}_*'\mathbf{u})$ quantify the distance between two unit vectors and $\sin\Theta(\boldsymbol{V}_{*}, \boldsymbol{V}) = \text{diag}\left\{\sin(\arccos \sigma_1), \ldots, \sin(\arccos \sigma_r)\right\}\in \bbR^r$ quantify the distance between $\boldsymbol{V}_{*}, \boldsymbol{V} \in \mathbb{R}^{p \times r}$ and  $\sigma_1 \geq \ldots \geq \sigma_r \geq 0$ are the singular values of $\boldsymbol{V}_{*}^{\prime}\boldsymbol{V}$. Finally, we let $\mathbb{S}^{N-1} = \{\boldsymbol{u} \in \mathbb{R}^{N}: \|\boldsymbol{u}\|_{2}=1\}$, $\mathcal{O}_{p, r} = \{\boldsymbol{V} \in \mathbb{R}^{p \times r}: \boldsymbol{V}^{\prime}\boldsymbol{V} = \boldsymbol{I}_{r}\}$, and $\mathcal{D}_{r}$ the set of diagonal matrices of size $r \times r$. More detailed notations can be found in the Appendix. 

\subsection{JisstPCA Algorithm}
\label{sec:jisstPCA}
The goal of JisstPCA is to estimate both the population factor $\boldsymbol{u}_{k}^*$ and network factors $\boldsymbol{V}_{k}^*, \boldsymbol{W}_{k}^*$ in \eqref{eq:sst_d} simultaneously. Inspired by the multi-block PCA literature \citep{westerhuis1998analysis,abdi2013multiple} where different data tables are appropriately scaled before a joint PCA, one may consider the following weighted objective: 
\begin{equation}\label{eq:jisstPCAloss}
\begin{aligned}
\underset{\boldsymbol{u}_{k}, \boldsymbol{V}_{k}, \boldsymbol{W}_{k}, d_{x,k}, d_{y,k}}{\arg\min} \; & \lambda \left\|\mathcal{X}-\sum\limits_{k = 1}^{K} d_{x,k} \cdot \boldsymbol{V}_{k}\boldsymbol{V}^{\prime}_{k} \circ \boldsymbol{u}_{k}\right\|_{F}^{2} + (1-\lambda) \left\|\mathcal{Y}-\sum\limits_{k = 1}^{K} d_{y,k} \cdot \boldsymbol{W}_{k}\boldsymbol{W}^{\prime}_{k} \circ \boldsymbol{u}_{k}\right\|_{F}^{2}\\
\mathrm{s.t.} \;& \boldsymbol{V}_{k} \in \mathcal{O}_{p, r_{x,k}}, \; \boldsymbol{W}_{k} \in \mathcal{O}_{q, r_{y,k}}, \; \boldsymbol{u}_{k} \in \mathbb{S}^{N-1}.
\end{aligned}
\end{equation}
Here, the weight parameter $\lambda \in [0, 1]$ is analogous to the scaling parameter in multiblock PCA; and it can reflect our knowledge about which data set is more important or more trustworthy. If we let $\lambda = 0$ or $1$, the problem degenerates to non-integrated semi-symmetric tensor factorization studied in \cite{weylandt2022multivariate}. We discuss more about the selection of $\lambda$ in Section \ref{sec:practice} of the Appendix.

\subsubsection{Single-Factor JisstPCA}
\label{single_isst_pca} 
We start from the single-factor ($K=1$) case and will then extend it to $K>1$. The subscript $k$ will be omitted when we discuss the single-factor model. Inspired by the success of power iteration in the tensor PCA literature \citep{zhang2018tensor,weylandt2022multivariate} and the weighted objective \eqref{eq:jisstPCAloss}, we propose a integrated power iteration algorithm that incorporates the weight parameter $\lambda$ into the updates. In particular, given the prior update of the joint factor $\bu$, we update $\bV$ and $\bW$ as the top eigenvectors of $\cX\times_3\bu$ and $\cY\times_3\bu$; given $\bV$, $\bW$, $\bu$ can be updated by pooling together the trace products $[\cX;\bV^{(t+1)}]$ and $[\cY, \bW^{(t+1)}]$ defined in Section \ref{sec:notation}, with the weight parameter $\lambda$. The detailed procedure is summarized in Algorithm \ref{alg:single_tt}. 
The main difference between Algorithm \ref{alg:single_tt} and the SS-TPCA algorithm in \cite{weylandt2022multivariate} lies in the update of $\bu$ that utilizes weighted power iterates from both tensors. Note that we do not claim to solve the optimization problem \eqref{eq:jisstPCAloss}, but instead, we will show the statistical convergence properties of our factor updates with high probability when the data is sampled from model \eqref{eq:sst_d} in Section \ref{sec:theory}. 
\begin{algorithm}[!htb]
\caption{Single-factor JisstPCA}\label{alg:single_tt}
\begin{itemize}
\item Input: $\mathcal{X}$, $\mathcal{Y}$, $r_x$, $r_y$, $\lambda$, and maximum iteration $t_{\max}$.
    \item Initialization: Let $t = 0$, and $\boldsymbol{u}^{(0)} = \text{Leading singular vector of } \left[\lambda \mathcal{M}_{3}(\mathcal{X}), (1- \lambda) \mathcal{M}_{3}(\mathcal{Y})\right]$.
    \item \textbf{repeat} until $t = t_{\max}$ or convergence:
    \begin{align*}
        & \boldsymbol{V}^{(t+1)} = \text{Leading $r_x$ singular vectors of } \mathcal{X} \times_3 \boldsymbol{u}^{(t)}\\
        & \boldsymbol{W}^{(t+1)} = \text{Leading $r_y$ singular vectors of } \mathcal{Y} \times_3 \boldsymbol{u}^{(t)}\\
        & \boldsymbol{u}^{(t+1)} = \dfrac{\lambda[\mathcal{X}; \boldsymbol{V}^{(t+1)}] + (1-\lambda)[\mathcal{Y}; \boldsymbol{W}^{(t+1)}]}{\left\|\lambda[\mathcal{X}; \boldsymbol{V}^{(t+1)}] + (1-\lambda)[\mathcal{Y}; \boldsymbol{W}^{(t+1)}]\right\|_2}\\
        & t = t+1
    \end{align*}
    \item \textbf{return} $\hat{\boldsymbol{u}} = \boldsymbol{u}^{(t)}$, $\hat{\boldsymbol{V}} = \boldsymbol{V}^{(t)}$, $\hat{\boldsymbol{W}} = \boldsymbol{W}^{(t)}$; $\hat{d}_x = \langle \mathcal{X}, \hat{\boldsymbol{V}} \hat{\boldsymbol{V}}^{\prime} \circ \hat{\boldsymbol{u}} \rangle / r_x$, $\hat{d}_y = \langle \mathcal{Y}, \hat{\boldsymbol{W}} \hat{\boldsymbol{W}}^{\prime} \circ \hat{\boldsymbol{u}} \rangle / r_y$; $\hat{\mathcal{X}} = \hat{d}_x \cdot \hat{\boldsymbol{V}} \hat{\boldsymbol{V}}^{\prime} \circ \hat{\boldsymbol{u}}$, $\hat{\mathcal{Y}} = \hat{d}_y \cdot \hat{\boldsymbol{W}} \hat{\boldsymbol{W}}^{\prime} \circ \hat{\boldsymbol{u}}$.
\end{itemize}
\end{algorithm}
For the initialization $\bu^{(0)}$, \cite{weylandt2022multivariate} suggests using the warm initialization $\bu^{(0)}=\mathbf{1}_{N}/\sqrt{N}$. We will show in Section \ref{sec:theory} that when the networks across the population are not too different from each other, warm initialization may suffice to guarantee convergence; while without such prior knowledge/belief, we suggest using the spectral initialization as in many prior works \citep{zhang2018tensor}:
\vspace{-10pt}
\begin{equation}\label{init}
    \boldsymbol{u}^{(0)} = \text{leading singular vector of } \left[\lambda \mathcal{M}_{3}(\mathcal{X}), (1 - \lambda) \mathcal{M}_{3}(\mathcal{Y})\right].
\end{equation}
\vspace{-10pt}
Here we still use $\lambda$ to weight the two tensors as in Algorithm \ref{alg:single_tt}.

\subsubsection{Multi-Factor JisstPCA}
\label{multi_isst_pca}
We now consider the general and more challenging case with $K > 1$. Inspired by the success of deflation methods in various matrix and tensor PCA approaches \citep{mackey2008deflation, allen2012sparse, weylandt2022multivariate,ge2021understanding}, 
we adopt a scheme that successively applies single-factor JisstPCA followed by deflating the tensor based on previously estimated factors.  Specifically, we apply Algorithm \ref{alg:single_tt} to $\mathcal{X}^{k}$ and $\mathcal{Y}^{k}$, deflate each of these tensors to get $\mathcal{X}^{k+1}$ and $\mathcal{Y}^{k+1}$, and repeat until we have extracted all $K$ factors; note that we set $\mathcal{X}^{1} = \mathcal{X}$ and $\mathcal{Y}^{1} = \mathcal{Y}$.  There are several possible deflation schemes \citep{mackey2008deflation}, but perhaps subtraction deflation is the simplest and imposes the fewest assumptions:
\begin{equation}\label{eq:sub_def}
\begin{aligned}
    \mathcal{X}^{k+1} = \mathcal{X}^{k} - \hat{d}_{x,k} \cdot \hat{\boldsymbol{V}}_{k}\hat{\boldsymbol{V}}_{k}^{\prime} \circ \hat{\boldsymbol{u}}_{k},\quad \mathcal{Y}^{k+1} = \mathcal{Y}^{k} - \hat{d}_{y,k} \cdot \hat{\boldsymbol{W}}_{k}\hat{\boldsymbol{W}}_{k}^{\prime} \circ \hat{\boldsymbol{u}}_{k}.
\end{aligned}
\end{equation} 
In many scenarios, however, one may wish to impose orthogonality constraints on all the tensor factors.  To achieve this, one could employ projection deflation for tensors as in \citep{mackey2008deflation,allen2012sparse}.  We note, however, that mutual orthogonality of the network factors across the layers may be quite restrictive for real data and limit the expressivity and interpretability of the model.  Instead, one might want to impose orthogonality of the population factors, $\mathbf{u}_{k}$, to enable similar interpretation of these as sample PCs.  To achieve this, we propose partial projection deflation:
\begin{equation}\label{eq:proj_def_u}
\begin{aligned}
    & \mathcal{X}^{k+1} = \left(\mathcal{X}^{k} - \hat{d}_{x,k} \cdot \hat{\boldsymbol{V}}_{k}\hat{\boldsymbol{V}}_{k}^{\prime} \circ \hat{\boldsymbol{u}}_{k}\right)\times_3 \left(\boldsymbol{I}_{N}-\hat{\boldsymbol{u}}_{k}\hat{\boldsymbol{u}}_{k}^{\prime} \right)\\
    & \mathcal{Y}^{k+1} = \left(\mathcal{Y}^{k} - \hat{d}_{y,k} \cdot \hat{\boldsymbol{W}}_{k}\hat{\boldsymbol{W}}_{k}^{\prime} \circ \hat{\boldsymbol{u}}_{k}\right) \times_3 \left(\boldsymbol{I}_{N}-\hat{\boldsymbol{u}}_{k}\hat{\boldsymbol{u}}_{k}^{\prime} \right).
\end{aligned}
\end{equation}
More details on this, other possibilities for deflation schemes, and the selection of number of factors $K$ are provided in the Appendix. 
\subsection{Extensions and Practical Considerations}\label{Other}
\paragraph{Generalized JisstPCA} In this section, we consider an important extension of the JisstPCA model \eqref{eq:sst_d} and algorithm to incorporate broader application scenarios. In particular, model \eqref{eq:sst_d} assumes one single eigenvalue for each network factor, which may not be reasonable when one network factor consists of multiple components with different connection strengths. 
Therefore, we extend \eqref{eq:sst_d} to allow {\em multiple eigenvalues} for each network factor: 
\begin{equation}\label{gen_model}
\begin{aligned}
     \mathcal{X} = \sum\limits_{k = 1}^{K} \boldsymbol{V}_{k}^*\boldsymbol{D}^{*}_{x,k} \boldsymbol{V}^{*\prime}_{k} \circ \boldsymbol{u}_{k}^* + \mathcal{E}_{x},\quad \mathcal{Y} = \sum\limits_{k = 1}^{K} \boldsymbol{W}_{k}^*\boldsymbol{D}^{*}_{y,k}\boldsymbol{W}^{*\prime}_{k} \circ \boldsymbol{u}_{k}^* + \mathcal{E}_{y}.
\end{aligned}
\end{equation}
where $\bD^*_{x,k}\in \bbR^{r_{x,k}\times r_{x,k}}$ and $\bD^*_{y,k}\in \bbR^{r_{y,k}\times r_{y,k}}$ are diagonal matrices, replacing the scalar eigenvalues $d^*_{x,k}$, $d^*_{y,k}$ in \eqref{eq:sst_d}.
To estimate factors in this generalozed JisstPCA model, 
we still propose a power iteration algorithm for the single-factor case and apply a successive deflation scheme for multi-factor models. The main change we make to the single-factor JisstPCA algorithm is that we update the diagonal matrix $\bD_x$ and $\bD_y$ within each iteration due to its increased importance, and we use them weight the columns of $\bV$, $\bW$ when updating $\bu$. 
We term this new algorithm the generalized JisstPCA (G-JisstPCA) algorithm, whose detailed procedures are summarized in Section \ref{sec:morealg} of the Appendix. 

\paragraph{Other extensions and hyperparameter tuning} Another important extension of our JisstPCA framework is to jointly analyze network data together with vector-valued covariates, organized as a matrix. This is especially useful for integrating feature vectors in neuroimaging studies, such as genomics, demographic, or behavioral traits for each subject. The detailed matrix-tensor JisstPCA model and its associated algorithm are included in Section \ref{sec:morealg} of the Appendix. We also include a detailed discussion on how to select the hyperparameters in practice in Section \ref{sec:practice} of the Appendix, including the number of factors $K$, the ranks of each factor $r_{x,k}$, $r_{y,k}$, $1\leq k\leq K$, and the integrative scaling parameter $\lambda$. 
\section{Theoretical Guarantees}\label{sec:theory}
In this section, we provide theoretical properties of single-factor JisstPCA (Algorithm \ref{alg:single_tt}) and Generalized JisstPCA (Algorithm \ref{single_diag}) in terms of the estimation errors for population and network factors. We first show a deterministic one-step convergence result that does not assume the noise distribution, under a deterministic SNR condition and an initialization condition. In Section \ref{sec:subGauss}, we focus on the special case of sub-Gaussian noise, showing that under suitable SNR assumptions, spectral initialization is good enough to ensure one-step statistical convergence. 
\subsection{Deterministic Convergence Guarantee} 
To begin with, we state two assumptions on the SNR and initialization, $\bu^{(0)}$.
\begin{assump}[SNR condition]\label{assump:SNR1}
    $d_x^* \geq 5\|\cE_x\|_{\op}$, $d_y^* \geq 5\|\cE_y\|_{\op}$.
\end{assump}
\noindent Assumption \ref{assump:SNR1} enforces that the signal strength $d_x^*$ and $d_y^*$ for tensors $\cX$ and $\cY$ are at least constant times larger than the noise level $\|\cE_x\|_{\op}$ and $\|\cE_y\|_{\op}$. 
\begin{assump}[Initialization condition]\label{assump:init}
    The initialization $\bu^{(0)}$ satisfies 
    $|\sin\theta(\bu^*, \bu^{(0)})|^2\leq 1-8\left(\frac{\|\cE_x\|_{\op}^2}{d_x^{*2}}\vee \frac{\|\cE_y\|_{\op}^2}{d_y^{*2}}\right)$.
\end{assump}
\noindent Assumption \ref{assump:init} is concerned with the initialization error $|\sin\theta(\bu, \bu^{(0)})|\in [0,1]$. We note that when Assumption \ref{assump:SNR1} holds, Assumption \ref{assump:init} can be implied by $|\sin\theta(\bu, \bu^{(0)})|\leq \frac{4}{5}$; when the {\em SNR continues to grow, Assumption \ref{assump:init} becomes weaker and easier to satisfy}. We will discuss this Assumption in more detail after Theorem \ref{thm:stat_converg_main}. 

Before presenting our deterministic one-step convergence guarantee for JisstPCA, we also define a notation representing the integrated noise level: $\|\lambda\cE_x; (1-\lambda)\cE_y\|_{r_x, r_y, \op} = \sup_{\bV\in \cO_{p\times r_x}, \bW\in \cO_{q\times r_y}}\|\lambda[\cE_x;\bV] + (1-\lambda)[\cE_y;\bW]\|_2.$ 
\begin{thm}\label{thm:stat_converg_main}
    Suppose $\mathcal{X}, \mathcal{Y}$ satisfy \eqref{eq:sst_d} with $K = 1$, and Assumptions \ref{assump:SNR1} and \ref{assump:init} hold. 
   Then the output of Algorithm \ref{alg:single_tt} satisfies the following: for $k\geq 1$,
\begin{equation}\label{eq:main_err1}
    |\sin\theta(\boldsymbol{u}^*, \boldsymbol{u}^{(k)})| \leq \frac{4\|\lambda \cE_x; (1-\lambda)\cE_y\|_{r_x,r_y,\op}}{\lambda r_xd_x^* + (1-\lambda)r_yd_y^*},
\end{equation}
\begin{equation}\label{eq:main_err2}
    \left\|\sin\Theta(\boldsymbol{V}^*, \boldsymbol{V}^{(k+1)})\right\|_{\mathrm{op}}\leq \frac{4\|\cE_x\|_{\op}}{d_x^*},\quad \left\|\sin\Theta(\boldsymbol{W}^*, \boldsymbol{W}^{(k+1)})\right\|_{\mathrm{op}}\leq \frac{4\|\cE_y\|_{\op}}{d_y^*}.
\end{equation}
\end{thm}
\noindent For the single-factor model, Theorem \ref{thm:stat_converg_main} suggests that as long as the signal is not masked by the noise (Assumption \ref{assump:SNR1}), and the initialization is reasonably good, then both the population factor $\bu$ and network factors $\bV,\,\bW$ can be well estimated after updating each factor only once. The one-step convergence may be inherited from the nice properties of the power iteration. The estimation error for the joint factor $\bu$ is the ratio between the spectral norm of integrated noise and the integrated signal; the errors for the network factors $\bV$ and $\bW$ are the corresponding SNR for each tensor. This resembles the matrix perturbation bounds (Davis-Kahan's Theorem) and recent tensor perturbation theory under the CP and Tucker models \citep{luo2021sharp,anandkumar2014guaranteed}. \eqref{eq:main_err1} shows that by integrating two tensors, the estimation accuracy of $\bu^*$ achieves a balance between the SNRs of $\cX$ and $\cY$. 

\textbf{Proof sketch}: Recall that Algorithm \ref{alg:single_tt} takes a similar form to a power iteration, where we alternatively update the estimates of $\bV$, $\bW$ and $\bu$ by taking the top singular vectors of a matrix computed from prior updates. In our proof, we apply the Davis-Kahan's Theorem to show that the perturbation bound for each update mainly depends on the SNR, inflated by a factor depending on the error of prior updates. Whenever the initialization and SNR conditions hold, the errors of $\bV^{(1)}$, $\bW^{(1)}$, $\bu^{(1)}$ are all bounded by constants, leading to the sufficiently small statistical error for the next updates $\bu^{(1)}$, $\bV^{(2)}$, $\bW^{(2)}$.
\begin{rmk}[Initialization condition]
    As we will show later in Section \ref{sec:subGauss}, spectral initialization will satisfy Assumption \ref{assump:init} when the noise is sub-Gaussian and under SNR conditions comparable to those in the literature. Furthermore, in the special case where the population variation is not too large ($\bu^*_i$'s are not too different), a warm initialization with $\bu^{(0)} = (\frac{1}{\sqrt{N}},\dots,\frac{1}{\sqrt{N}})'$ may also work without imposing any distributional assumption on the noise. We formalize this intuition in Corollary \ref{cor:warmInit}.
\end{rmk}
\begin{cor}[Warm initialization]\label{cor:warmInit}
    Let $\widehat{\mathrm{Var}}(u^*) = \frac{1}{N}\sum_{i}(\bu^*_i-\frac{1}{N}\sum_j\bu^*_j)^2$, $\widehat{\mathbb{E}}(u^{*}) = \frac{1}{N}\sum_{i}\bu^{*}_i$. Then as long as Assumption \ref{assump:SNR1} holds and
        $\frac{\widehat{\mathrm{Var}}(u^*)}{(\widehat{\mathbb{E}}(u^*))^2}\leq \left(\frac{d_x^{*2}}{8\|\cE_x\|_{\op}^2}\wedge \frac{d_y^{*2}}{8\|\cE_y\|_{\op}^2}\right) - 1$,
    the output of Algorithm \ref{alg:single_tt} with warm initialization $\bu^{(0)} = (\frac{1}{\sqrt{N}},\dots,\frac{1}{\sqrt{N}})'$ satisfies the estimation error bounds \eqref{eq:main_err1} and \eqref{eq:main_err2} for $k\geq 1$.
\end{cor}
\noindent Corollary \ref{cor:warmInit} suggests that as long as the variation amongst sample factors does not dominate its mean, then our JisstPCA algorithm with warm initialization enjoys one-step convergence.
\subsection{Special Case: Sub-Gaussian Noise and Spectral Initialization}\label{sec:subGauss}

In this subsection, we present the convergence property of JisstPCA when the noise tensors are entrywise sub-Gaussian, and when the spectral initialization is applied.
\begin{assump}[Sub-Gaussian noise]\label{assump:subGaussNoise}
    Suppose that for $1\leq k\leq N$, $(\cE_x)_{:,:,k}$ and $(\cE_y)_{:,:,k}$ have independent, zero-mean, sub-Gaussian-$\sigma$ entries subject to symmetry constraints. Let $\sigma^2_{i,j,k}(\cE_x) = \mathrm{Var}((\cE_x)_{i,j,k})$, $\sigma^2_{i,j,k}(\cE_y) = \mathrm{Var}((\cE_y)_{i,j,k})$ be entrywise variances, which satisfy $\sum_{i,j}\sigma^2_{i,j,k}(\cE_x) = \sum_{i,j}\sigma^2_{i,j,k'}(\cE_x)$, and $\sum_{i,j}\sigma^2_{i,j,k}(\cE_y) = \sum_{i,j}\sigma^2_{i,j,k'}(\cE_y)$ for all $1\leq k,\, k' \leq N$.
\end{assump}
\noindent Assumption \ref{assump:subGaussNoise} is weaker than entrywise i.i.d. noise, as we allow $\sigma_{i,j,k}(\cE_x)$ ($\sigma_{i,j,k}(\cE_y)$) to be different across different $i,j\in [p]$ ($i,j\in [q]$): variance can be location-varying within the network.  We also allow difference noise variances between the two tensors $\cE_x$ and $\cE_y$. The only homogeneous requirement is on the third mode (population mode) for a given tensor, meaning that all samples have the same noise level for a given network modality.
\begin{assump}[SNR condition under sub-Gaussian noise]\label{assump:SNR2}
    Let $d_{\lambda} = \sqrt{\lambda r_x d_{x}^{*2} + (1-\lambda)r_y d_{y}^{*2}}$ be the integrated ground truth signal, and suppose that $d_{\lambda} \geq C\sigma\left(N^{1/4} \left(\sqrt{p} + \sqrt{q}\right)+ \sqrt{N}\right)\sqrt{\log N}$ for some constant $C>0$. In addition, the signal of each network tensor satisfies $d_x^*\geq C\sigma(\sqrt{N}+\sqrt{p})$, $d_y^*\geq C\sigma(\sqrt{N}+\sqrt{q})$.
\end{assump}
\noindent The individual SNR conditions ($d_x^*\geq C\sigma(\sqrt{N}+\sqrt{p})$, $d_y^*\geq C\sigma(\sqrt{N}+\sqrt{q})$) in Assumption \ref{assump:SNR2} are equivalent to Assumption \ref{assump:SNR1} under the sub-Gaussian noise. The additional integrated SNR condition on $d_{\lambda}$ ensures a reasonably good spectral initialization (see Proposition \ref{thm:init}). In the special case where $d_x^* = d_y^* =d$, $r_x=r_y=r$, $p=q$, Assumption \ref{assump:SNR2} can be implied by $d/\sigma\geq C\max\{r^{-\frac{1}{2}}(\sqrt{p}N^{1/4}+\sqrt{N})\sqrt{\log N},\sqrt{p}+\sqrt{N}\}$. This is weaker than the SNR condition ($d/\sigma\geq Cr(\sqrt{pN}+\sqrt{N\log N})$) in the prior work on single semi-symmetric tensor PCA \citep{weylandt2022multivariate}. In addition, when $p=N$, prior work on Tucker low-rank tensor PCA \citep{zhang2018tensor} requires $d/\sigma\geq Cp^{3/4}$, comparable to our SNR condition $d/\sigma\geq C\max\{r^{-\frac{1}{2}}p^{3/4}\sqrt{\log p},\sqrt{p}\}$ \footnote{We only have an additional log-factor $\sqrt{\log p}$ since we allow location-varying noise variances.}. 
\begin{prop}\label{thm:init}
    Suppose $\mathcal{X}, \mathcal{Y}$ are generated from \eqref{eq:sst_d} with $K = 1$, and Assumptions \ref{assump:subGaussNoise} and \ref{assump:SNR2} hold. Then, for the spectral initialization for $\boldsymbol{u}^{(0)}$ as defined in \eqref{init}, we have
\begin{equation}\label{thm1}
    \left|\sin\theta(\boldsymbol{u}, \boldsymbol{u}^{(0)}) \right| \leq \frac{C\sigma\left(\sqrt{N} + (N(p^2+q^2))^{\frac{1}{4}}\right)\sqrt{\log N}}{d_{\lambda}}\leq \sqrt{1-8\left(\frac{\|\cE_x\|_{\op}^2}{d_x^{*2}}\vee\frac{\|\cE_y\|_{\op}^2}{d_y^{*2}}\right)},
\end{equation}
with probability at least $1 - C \exp \left(-c N\right)$ for some constant $C, c > 0$.
\end{prop}
\noindent Proposition \ref{thm:init} suggests that under Assumptions \ref{assump:subGaussNoise}-\ref{assump:SNR2}, spectral initialization satisfies Assumption \ref{assump:init} with high probability. We are now in position to state our main statistical error bounds for JisstPCA with sub-Gaussian noise.
\begin{thm}\label{thm:main_subgauss}
    Suppose $\mathcal{X}, \mathcal{Y}$ satisfy \eqref{eq:sst_d} with $K = 1$, and Assumptions \ref{assump:subGaussNoise} and \ref{assump:SNR2} hold. Then Algorithm \ref{alg:single_tt} with spectral initialization \eqref{init} satisfies the following: for $k\geq 1$,
\begin{equation}\label{col_u}
   \left|\sin\theta(\boldsymbol{u}, \boldsymbol{u}^{(k)})\right| \leq \frac{C\sigma (\lambda r_x\sqrt{p+N} + (1-\lambda)r_y\sqrt{q+N})}{\lambda r_xd_x^* + (1-\lambda)r_yd_y^*},
\end{equation}
\begin{equation}\label{col_vw}
    \left\|\sin\Theta(\boldsymbol{V}, \boldsymbol{V}^{(k+1)})\right\|_{\mathrm{op}} \leq  \frac{C\sigma(\sqrt{p}+\sqrt{N})}{d_x^*},\quad \left\|\sin\Theta(\boldsymbol{W}, \boldsymbol{W}^{(k+1)})\right\|_{\mathrm{op}}\leq \frac{C\sigma(\sqrt{q}+\sqrt{N})}{d_y^*},
\end{equation}
with probability at least $1 - C \exp \left( -c N\right)$, where $C, c > 0$ are universal constants.
\end{thm}
\noindent Theorem \ref{thm:main_subgauss} shows that under sub-Gaussian noise and an SNR condition comparable to the prior literature, all factors converge to their statistical errors after being updated at least once. As far as we are aware, this is the first statistical guarantee for integrative tensor analysis. When $r_x=r_y=r$, $d_x^*=d_y^*=d$, $p=q$, our statistical error for $\bu$, $\bV$, $\bW$ all scale as $\frac{\sigma (\sqrt{p}+\sqrt{N})}{d}$. which improves upon the prior result \citep{weylandt2022multivariate} on non-integrated data ($\frac{\sigma r\sqrt{pN}}{d}$ for $\bu$ and $\frac{\sigma r^{3/2}\sqrt{pN}}{d}$ for $\bV$). In addition, our statistical error bounds are satisfied by one-step iterates with provable initialization, also demonstrating extremely fast convergence. Furthermore, our statistical error bound is comparable to prior results on HOOI for Tucker low-rank tensor PCA \cite{zhang2018tensor, luo2021sharp}, where the error of each factor scales as $\dfrac{\sigma\sqrt{p}}{d}$ when each mode's dimension scales as $p$.

\subsection{Extension: Convergence of Generalized JisstPCA}
To accommodate more general scenarios,  we also study the theoretical properties of the Generalized JisstPCA (Algorithm \ref{single_diag}) under the model \eqref{gen_model}. In this setting, we have different signal strengths for estimating the joint factor and individual factors, as reflected by the two separate SNR conditions as follows.
\begin{assump}[SNR for joint factor]\label{assump:dJisst_SNR1}
The integrated signal satisfies:
$\lambda\|\bD_x^*\|_F^2 + (1-\lambda)\|\bD_y^*\|_F^2\geq C\sigma^2(\sqrt{N(p^2+q^2)}+N)\log N$,
and the individual signals satisfy
$\|\bD_x^*\|_F^2\geq C\sigma^2 r_x(p+N),\,\|\bD_y^*\|_F^2\geq C\sigma^2r_y(q+N).$
\end{assump}
\noindent Assumption \ref{assump:dJisst_SNR1} ensures a good estimate for the joint factor $\bu^*$. We note that the signal strength depends on the Frobenious norms of $\bD_x^*$ and $\bD_y^*$, since they are the singular values of the matricization along the third mode. 
\begin{assump}[SNR for individual factors]\label{assump:dJisst_SNR2}
    $\sigma_{r_x}^2(\bD_x^*) \geq C\sigma^2(p+N)$, $\sigma_{r_y}^2(\bD_y^*) \geq C\sigma^2(q+N)$.
\end{assump}
\noindent Assumption \ref{assump:dJisst_SNR2} ensures good estimates for the individual factors $\bV^*$, $\bW^*$. Different from Assumption \ref{assump:dJisst_SNR1}, the estimation of individual factors depends on the $\sigma_{r_x}(\bD_x^*)$ and $\sigma_{r_y}(\bD_y^*)$, since they are the minimum nonzero singular values of the matricizations along first/second modes. Note that we allow entries of $\bD_x^*$ and $\bD_y^*$ to have arbitrary signs, different from the original JisstPCA model \eqref{eq:sst_d} (Assumption \ref{assump:SNR2}). In the special case when $\bD_x^*$ and $\bD_y^*$ are diagonal matrices with $d_x^*$, $d_y^*$, Assumptions \ref{assump:dJisst_SNR1}-\ref{assump:dJisst_SNR2} are both implied by Assumption \ref{assump:SNR2}.
\begin{thm}\label{thm:converg_D}
    Suppose that $\cX$ and $\cY$ are generated from \eqref{gen_model}, and we apply Algorithm \ref{single_diag} on them with the spectral initialization \eqref{init}. Then as long as Assumptions \ref{assump:subGaussNoise} and \ref{assump:dJisst_SNR1} hold, with probability at least $1-CN^{-c}$, for any iteration number $k\geq 1$, we have
\begin{equation*}
\begin{split}
    |\sin\theta(\bu^*,\bu^{(k)})|\leq &\frac{C\sigma\left(\lambda\sqrt{r_x(p+N)}\|\bD^*_x\|_F + (1-\lambda)\sqrt{r_y(q+N)}\|\bD^*_y\|_F\right)}{\lambda\|\bD^*_x\|_F^2 + (1-\lambda)\|\bD^*_y\|_F^2}\\
    \leq&\frac{C\sigma \sqrt{r_x(p+N)}}{\|\bD^*_x\|_F} \vee \frac{C\sigma \sqrt{r_y(q+N)}}{\|\bD^*_y\|_F}.
\end{split}
\end{equation*}
    If Assumption \ref{assump:dJisst_SNR2} also holds, the following holds for any $k\geq 2$ with the same probability:
   \begin{equation*}
   \begin{split}
        \|\sin\Theta(\bV^*,\bV^{(k)})\|_{\op}\leq \frac{C\sigma \sqrt{p+N}}{\sigma_{r_x}(\bD^*_x)},
        \|\sin\Theta(\bW^*,\bW^{(k)})\|_{\op}\leq \frac{C\sigma \sqrt{q+N}}{\sigma_{r_y}(\bD^*_y)}.
   \end{split}
\end{equation*} 
\end{thm}
\noindent Theorem \ref{thm:converg_D} extends the one-step convergence guarantee in Theorem \ref{thm:main_subgauss} to the general model and the Generalized JisstPCA algorithm, where the only change lies in the characterization of the SNR. 
As discussed earlier, the SNR for the joint factor depends on $\|\bD_x^*\|_F$ and $\|\bD_y^*\|_F$, while the SNR for the individual factors depend on $\sigma_{r_x}(\bD_x^*)$ and $\sigma_{r_y}(\bD_y^*)$.
As a comparison, the SNR in prior theoretical results for HOOI and Tucker low-rank tensor PCA \citep{zhang2018tensor} is characterized by the minimum singular value along all modes, which translates to $\sigma_{r_x}(\bD_x^*)$ or $\sigma_{r_y}(\bD_y^*)$ in our setting. Different from prior results, we are able to better separate the estimation guarantees of joint and individual factors since our G-JisstPCA algorithm makes use of a more compact decomposition \eqref{gen_model} than the Tucker decomposition.
\begin{rmk}[Rank misspecification]
    Since no minimum eigenvalue condition on $\bD^*_x$ or $\bD^*_y$ is needed to accurately estimate the joint factor $\bu^*$, using larger ranks $r_x$, $r_y$ in Algorithm \ref{single_diag} does not affect the estimation of $\bu^*$. 
\end{rmk}
\section{Simulation Studies}
\label{sec:simulation}
In this section, we empirically study our JisstPCA algorithms by (i) validating our theoretical guarantees and (ii) comparing the performance of our methods with baseline methods. The code necessary for reproducing our empirical results is available at \href{https://github.com/JmL130169/JisstPCA}{https://github.com/JmL130169/JisstPCA}.
\subsection{Validation of Theoretical Guarantees}
\label{val_thm}
We first validate our main theory (Theorem \ref{thm:main_subgauss} and \ref{thm:converg_D}): one-step convergence and statistical error rates, for the JisstPCA model \eqref{eq:sst_d} and the generalized JisstPCA model \eqref{gen_model}, under the single-factor ($K=1$) case. For both models, we consider tensor dimensions $p=q=\frac{1}{2}N\in\{60,\, 90,\, 120\}$, ranks $r_x = 3$, $r_y = 2$. The ground truth tensor factors $\bV^*$, $\bW^*$, and $\bu^*$ are randomly generated with independent standard Gaussian entries and then orthogonalized. 
Each slice of the noise tensors ($(\cE_x)_{:,;,k}$, $(\cE_y)_{:,:,k}$, $1\leq k\leq N$) is symmetric mean-zero Gaussian with off-diagonal entries of variance $1$ and diagonal entries of variance $2$. 
For JisstPCA model, we let $d_y^* = 1.2d_x^*$, and the values of $d_x^*$ are set such that SNR = $\frac{d_x^*}{\sqrt{p}+\sqrt{N}}$ take values from $1.5$ to $24$. For the G-JisstPCA model, we consider the same set of $d_x^*$, $d_y^*$ but let $\boldsymbol{D}_x^* = d_x^*\mathrm{diag}(1.5, 1, 0.8)$, $\boldsymbol{D}_y^* = d_y^*\mathrm{diag}(1, 0.8)$. Given the noisy observations $\cX,\,\cY$, we run the JisstPCA and Generalized JisstPCA (G-JisstPCA) algorithms with $K=1$, $\lambda=0.5$, oracle ranks, and the spectral initialization \eqref{init}. Figure \ref{fig:validation} reports the spectral norm $\sin\Theta$ errors of the three estimated factors ($|\sin\theta(\bu^{(k)},\bu^*)|$, $\|\sin\Theta(\bV^{(k)},\bV^*)\|_{\op}$, $\|\sin\Theta(\bW^{(k)},\bW^*)\|_{\op}$). The first and third panels plot the estimation errors versus iteration number when the SNR is 1.5, and we can see all three factors converge to the statistical errors after one step update; The second and fourth panels plot the estimation errors of the first update versus its theoretical scaling ($1/\mathrm{SNR}$), validating our theory as well. 

\begin{figure}[!htb]
    \centering
    \includegraphics[width = \textwidth]{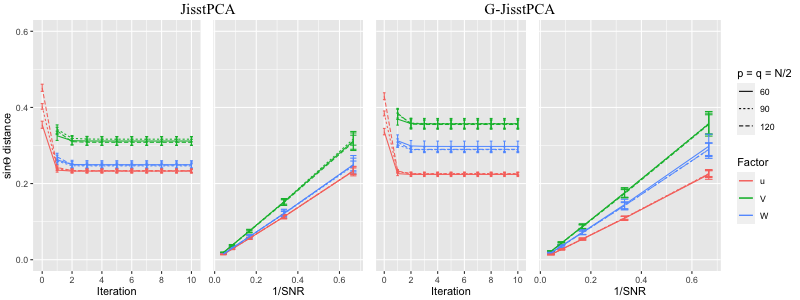}
    \caption{\small Empirical validation of theoretical results for JisstPCA (left) and G-JisstPCA (right). The first and third figures confirm that the $\sin\Theta$ distances of $\bu^{(k)}$, $\bV^{(k)}$, and $\bW^{(k)}$ converge to the statistical error after one or two step iterations, even under a relatively small SNR = 2.25. The second and fourth figures demonstrate linear dependence of the statistical error on 1/SNR, validating our theory. 20 independent replicates are run, and we plot the mean errors with error bars representing $95\%$ approximate normal confidence interval.}
    \label{fig:validation}
\end{figure}
\subsection{Comparative Study}
\label{comp_stu}

We seek to study the comparative advantages of our JisstPCA and G-JisstPCA algorithms across various settings. However, notice that there are no existing methods for integrative semi-symmetric tensor PCA that are natural baselines for comparison. Hence, we develop simple integrative extensions of the classical higher-order SVD (HOSVD) and higher-order orthogonal iteration (HOOI) algorithms as our comparison baselines; we call these ``iHOSVD" and ``iHOOI", respectively. We choose these baselines as the Tucker tensor decomposition model naturally leads to symmetric network factors; existing integrative CP-based approaches \citep{acar2011all} do not offer symmetric factors and are hence incomparable to our approach.  We provide detailed iHOSVD and iHOOI algorithms in the Appendix.  In all our simulations to avoid rank selection for the iHOSVD and iHOOI methods, we employ these approaches with oracle ranks.  Our JisstPCA and G-JisstPCA algorithms, on the other hand, are employed with both oracle ranks and data-driven selection of ranks via our BIC approach described in Section~\ref{sec:practice} and \ref{sec:EmpiricalDetails} of the Appendix.

We first consider data generated from our own models, \eqref{eq:sst_d} and \eqref{gen_model}. We focus on $K=2$ factors, ranks $\boldsymbol{r}_x = \boldsymbol{r}_y = (3,2)'$, $p=150$, $q=50$, and $N=50$. (The Appendix contains additional simulations varying the dimensionality $p,\,q$, and $N$.)  
Noise tensors are generated as previously described.  To simulate different scenarios, the ground truth factors are set as unstructured (randomly generated entries as previously described) or structured (factors of block and star networks as shown in the top panel of Figure~\ref{fig:structuredExample}; more details included in Section~\ref{sec:EmpiricalDetails} of the Appendix); we also leave these factors as non-orthogonal across $k=1,\,2$ or enforce mutual orthogonality along each mode. For non-orthogonal settings, we set $\boldsymbol{d}_x^* = \mathrm{SNR} * (\sqrt{p}+\sqrt{N})(1,\, 0.5)'$, $\boldsymbol{d}_y^* = \mathrm{SNR} * (\sqrt{q}+\sqrt{N})(1,\, 0.5)'$. In the orthogonal setting, we study the effect of singular gap between different factors by letting $\boldsymbol{d}_x^* = \mathrm{SNR} * (\sqrt{p}+\sqrt{N})(1,\, 1)'$, $\boldsymbol{d}_y^* = \mathrm{SNR} * (\sqrt{q}+\sqrt{N})(1,\, 0.9)'$. While under the generalized model \eqref{gen_model}, we let $\bD^*_{x,1} = d_{x,1}^*\mathrm{diag}(2,\,1.5,\,1.2)$, $\bD^*_{x,2} = d_{x,2}^*\mathrm{diag}(2,\,1.6)$, and $\bD^*_{y,1}$, $\bD^*_{y,2}$ similarly defined, where $\boldsymbol{d}_x^*$, $\boldsymbol{d}_y^*$ are set as previously described.

Figure \ref{fig:comp_main} summarizes the $\sin\Theta$ distance errors of the estimated factors by our JisstPCA, G-JisstPCA algorithms (with BIC selected ranks and subtraction deflation), and the baselines iHOSVD and iHOOI (with oracle ranks) under different SNR values. Both iHOSVD and iHOOI enforce orthogonality and hence suffer from biases for non-orthogonal factors. In the orthogonal case, iHOSVD/iHOOI can fail due to the lack of singular gap, while our approach relies on the Frobenious norm gap between  different factors and continues to perform well. (For fair comparison, we also include results with larger singular gaps in Figure \ref{fig:orthogonal2_BIC} in the Appendix, where JisstPCA/G-JisstPCA has similar performance as iHOOI.) For the generalized model, JisstPCA has the same performance as G-JisstPCA, showing some robustness against a slight model misspecification. Of course, when the differences between eigenvalues within the same factor are larger, JisstPCA can fail and be worse than G-JisstPCA (as shown in Figure \ref{fig:generalModel_BIC} in the Appendix). In summary, our methods always outperform the baselines, except for a few settings with small SNR due to incorrect rank selections. When all methods use oracle ranks, our methods always give the lowest errors (as shown in Figures \ref{fig:Main_oracle} in the Appendix). Furthermore, Figure \ref{fig:structuredExample} visualizes the estimated factors of JisstPCA and iHOOI in the structured simulation along with the true factors, showing that iHOOI tends to mix the two graph components because of the rotations possible with the Tucker tensor core as well as the forced orthogonality. Additional simulation details, more empirical results including JisstPCA with projection deflation, and further discussion is available in the Appendix. 

\begin{figure}[!htb]
    \centering
    \includegraphics[width = \textwidth]{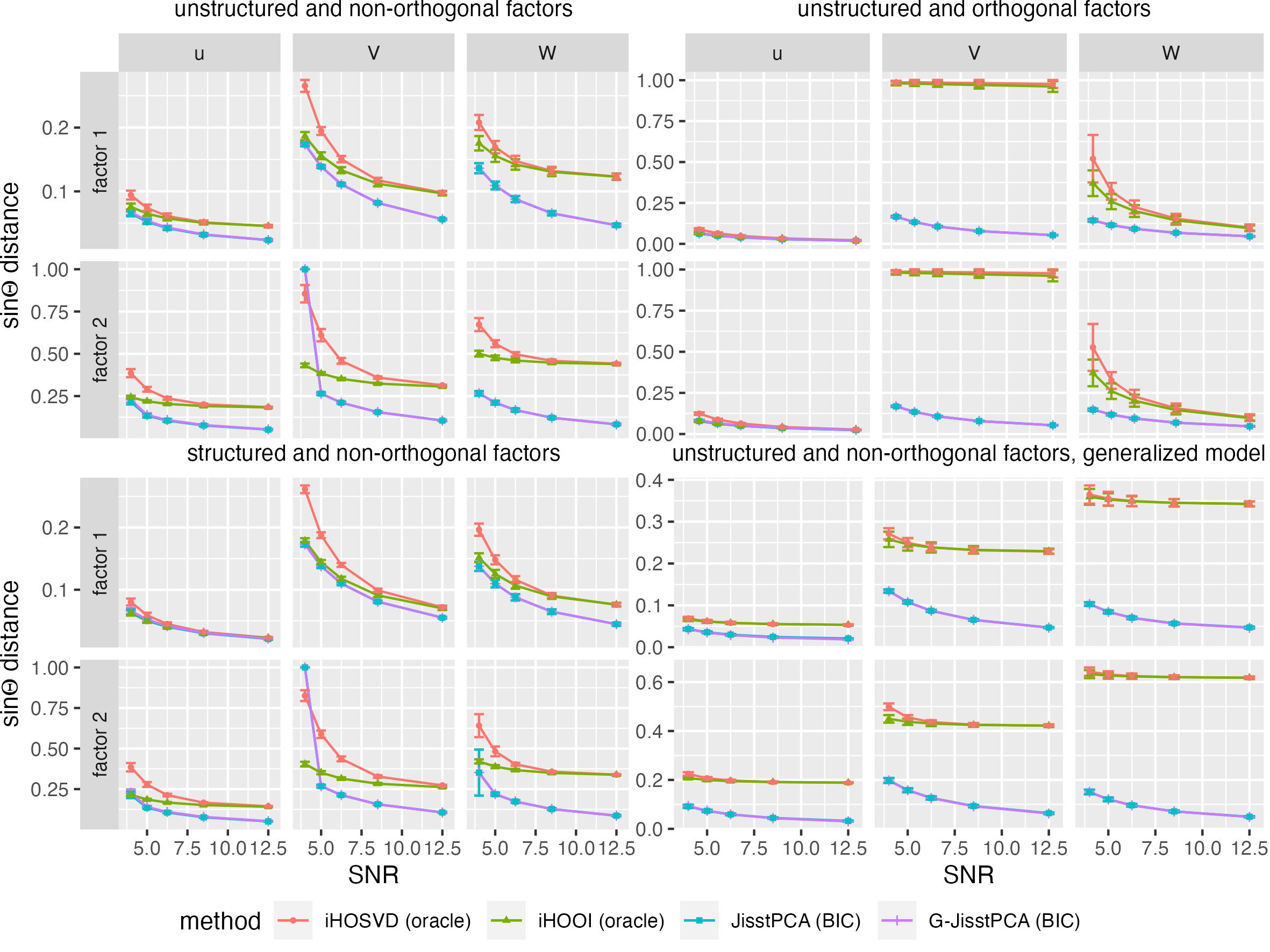}
    \caption{\small Estimation errors ($\sin\Theta$ distances in spectral norm) of all factors using JisstPCA, Generalized JisstPCA, iHOOI, and iHOSVD. The mean errors of $10$ replicates are plotted with error bars. Our methods strongly outperform the baseline methods, even with data-driven rank tuning (BIC).}
    \label{fig:comp_main}
\end{figure}

Finally, we seek to test the appropriateness of both our model as well as our algorithms on multi-modal population network data (binary adjacency tensors); this also serves to test the robustness of our approaches.  Specifically, we divide the $N$ samples randomly into two clusters ($K=2$), with probability $0.75$ and $0.25$, respectively. Each cluster is associated with a pair of stochastic block models (SBMs): $(\mathcal{M}_{k,x},\,\mathcal{M}_{k,y})$ for $k=1,\,2$. For any sample $i$ from cluster $k$, its associated networks $\cX_{:,:,i}$ and $\cY_{:,:,i}$ are adjacency matrices sampled from $\mathcal{M}_{k,x}$ and $\mathcal{M}_{k,y}$, respectively. We set $\mathcal{M}_{1,x}$ and $\mathcal{M}_{1,y}$ as three-block SBMs and $\mathcal{M}_{2,x}$ and $\mathcal{M}_{2,y}$ as two-block SBMs, with the within-block edge probabilities ranging in $(0.5,0.8)$ and out-of-block probability $0.3$.  Thus, these are truly low-rank networks with the rank the number of communities. Given the population of multi-modal adjacency matrices, we test how well our JisstPCA algorithms can extract the population factors and the underlying pairs of network components, which can be further used to detect the population clusters and community structures in each network component by applying k-means.
Note that with populations of adjacency matrices, however, the top singular spaces all share the all one's vector; hence the top singular vectors across different network factors are not linearly independent. Thus, we project the rows and columns of each adjacency matrix onto the complement of the all one's vector before applying all methods; we suggest to follow this procedure in practice for populations of networks represented as adjacency matrices. Noting that the population membership factors are mutually orthogonal while the network factors are not, we apply our JisstPCA algorithms with partial projection deflation on the population mode. We also select the ranks for JisstPCA using BIC-deflation, while applying iHOSVD, and iHOOI with the oracle ranks (the number of stochastic blocks minus one). We then apply k-means on the estimated population factors to cluster the samples, as well as on each network factor to cluster the nodes in the network. Table \ref{tab:network_cluster1} shows the clustering accuracy (Adjusted Rand Index) based on the extracted factors from all methods with different sample sizes, when $p=80$, $q=50$. We also report the $\sin\Theta$ distance errors of each factor; we also perform the same projection for the edge probability tensors as for the adjacency tensors and then extract the ground truth factors. Here, we use relatively smaller dimensions $p,\,q,\,N$ since larger dimensions increase the SNR in SBMs and make the task too easy. Our methods strongly outperform the Tucker model based approaches in all scenarios, hence highlighting the advantanges of our modeling framework and algorithms for analyzing and detecting clusters in real network data. Additional results with the subtraction deflation and different network sizes, more details on the model set-up, and visualizations of the true and estimated network factors can be found in the Appendix. 

    \begin{figure}[!htb]
        \centering
        \includegraphics[width = 15cm]{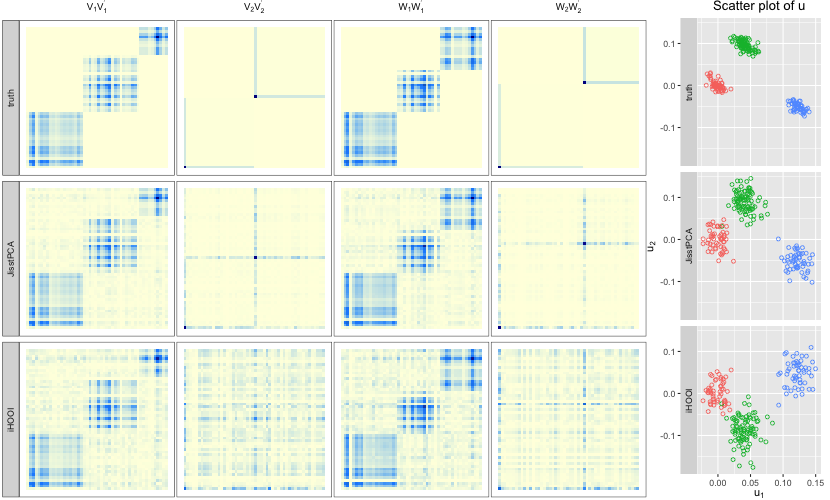}
        \caption{\small Heatmaps and scatterplots for the structured network factors and population factors, respectively.  The truth is shown in the top panel, results of JisstPCA (with BIC rank selection) in the middle panel, and iHOOI (with oracle rank) in the bottom panel.}
        \label{fig:structuredExample}
    \end{figure}

\begin{table}[!htb]
\centering
\caption{\small Population clustering and network community detection for multi-modal populations of networks, based on k-means on the estimated factors from JisstPCA, G-JisstPCA, iHOSVD, and iHOOI. The presented Adjusted Rand Index (ARI) values demonstrate the accuracy of sample clustering and node clustering of two network factors for each modality, when network sizes $p=80$, $q=50$. The $\sin\Theta$ estimation errors of each factor is also presented. The average ARI and $\sin\Theta$ distances of 20 independent repeats are presented, with standard deviation inside the parenthesis. The largest average ARI and lowest estimation error for each setting are marked in bold.}\label{tab:network_cluster1}
\centering
\scalebox{0.65}{
\begin{tabular}{c|c|c|c|c|c|c|c|c}
\hline
\multirow{2}{*}{Clustering}&\multicolumn{4}{|c|}{$N=20$}&\multicolumn{4}{|c}{$N=40$}\\
\cline{2-9}
\multirow{2}{*}{ARI}&{\bf JisstPCA} &{\bf G-JisstPCA}&{\bf iHOSVD}&{\bf iHOOI}&{\bf JisstPCA} &{\bf G-JisstPCA}&{\bf iHOSVD}&{\bf iHOOI}\\
&{\bf (BIC)} &{\bf (BIC)}&{\bf (oracle)}&{\bf (oracle)}&{\bf (BIC)} &{\bf (BIC)}&{\bf (oracle)}&{\bf (oracle)}\\
\hline
Sample&0.947(0.238) & \textbf{1}(0) & \textbf{1}(0) & 0.954(0.205) & \textbf{1}(0) & \textbf{1}(0) & \textbf{1}(0) & \textbf{1}(0)\\
Network 1 of $\cX$&0.971(0.108) & \textbf{1}(0) & 0.912(0.195) & 0.739(0.251) & \textbf{1}(0) & \textbf{1}(0) & 0.805(0.241) & 0.663(0.224)\\
Network 2 of $\cX$&0.995(0.022) & \textbf{0.997}(0.011) & 0.146(0.057) & 0.139(0.043) & \textbf{1}(0) & \textbf{1}(0) & 0.156(0.015) & 0.153(0)\\
Network 1 of $\cY$&0.974(0.114) & \textbf{1}(0) & 0.99(0.027) & 0.94(0.155) & \textbf{1}(0) & \textbf{1}(0) & 0.973(0.108) & \textbf{1}(0)\\
Network 2 of $\cY$&\textbf{0.87}(0.31) & 0.867(0.309) & 0.014(0.046) & 0.121(0.16) & \textbf{1}(0) & \textbf{1}(0) & 0.116(0.101) & 0.246(0.153)\\
\hline
\hline
$\sin\theta(\hat{\bu}_1,\bu_1^*)$&0.087(0.03) & 0.089(0.031) & 0.138(0.051) & \textbf{0.08}(0.026) & 0.092(0.016) & 0.093(0.017) & 0.142(0.028) & \textbf{0.083}(0.016)\\
$\sin\theta(\hat{\bu}_2,\bu_2^*)$&0.167(0.077) & \textbf{0.163}(0.066) & 0.214(0.04) & 0.188(0.162) & 0.152(0.014) & 0.152(0.014) & 0.197(0.017) & \textbf{0.145}(0.015)\\
$\|\sin\Theta(\hat{\bV}_1,\bV_1^*)\|_{\op}$&\textbf{0.084}(0.007) & \textbf{0.084}(0.007) & 0.163(0.06) & 0.163(0.059) & \textbf{0.062}(0.006) & \textbf{0.062}(0.007) & 0.154(0.046) & 0.156(0.047)\\
$\|\sin\Theta(\hat{\bV}_2,\bV_2^*)\|_{\op}$&\textbf{0.21}(0.074) & 0.211(0.072) & 0.817(0.065) & 0.799(0.07) & \textbf{0.154}(0.024) & 0.155(0.023) & 0.776(0.022) & 0.768(0.028)\\
$\|\sin\Theta(\hat{\bW}_1,\bW_1^*)\|_{\op}$&\textbf{0.158}(0.018) & \textbf{0.158}(0.018) & 0.253(0.056) & 0.196(0.058) & \textbf{0.118}(0.014) & \textbf{0.118}(0.014) & 0.202(0.026) & 0.173(0.04)\\
$\|\sin\Theta(\hat{\bW}_2,\bW_2^*)\|_{\op}$&\textbf{0.416}(0.204) & 0.417(0.203) & 0.969(0.044) & 0.903(0.098) & \textbf{0.272}(0.045) & 0.273(0.045) & 0.901(0.062) & 0.771(0.079)\\
\bottomrule
\end{tabular}}
\end{table}

\section{Case Study: Multi-Modal Population Brain Connectivity}
\label{real_data}

We apply our proposed Generalized JisstPCA method to understand multi-modal and multi-subject brain connectivity patterns estimated from neuroimaging data.  We analyze data from the Human Connectome Project (HCP), which can be easily accessed through the ConnectomeDB website and contains various traits, structural MRI (sMRI), functional MRI (fMRI), and diffusion MRI (dMRI) data  for 1058 subjects \citep{glasser2013minimal}.  Our objective is to understand major joint patterns in functional connectivity (FC) and structural connectivity (SC) as well as to see how these patterns vary and are related to other traits across the population.  
We process the data by applying the population-based connectome (PSC) extraction pipeline \citep{zhang2018mapping} to construct the SC. Using the Desikan-Killiany atlas, we identify 68 cortical regions of interest (ROIs) and 19 subcortical ROIs, totaling 87 ROI nodes in our networks. The number of fiber curves between a pair of regions, measured at the logarithmic scale, is used to quantify the connection strength, or edge in our networks. Thus, the final dimension of our population SC tensor is \( 87 \times 87 \times 1058 \), (ROIs by ROIs by subjects). We use the same atlas to extract FC, the computation of which is straightforward since the HCP provides preprocessed fMRI \citep{glasser2013minimal}. Specifically, we compute the mean fMRI time series for each ROI and calculate the Pearson correlation between different ROIs to generate a full FC matrix for each subject; the final population FC tensor has the same dimension as the SC tensor. Figure \ref{fig:mean_sc_fc} in the Appendix shows the mean SC and FC weighted adjacency matrices, where the first 19 rows are the subcortical ROIs and the next 68 are the cortical ROIs; Supplement II is an Excel spreadsheet showing the ROI names. In addition, we also seek to understand how brain connectivity patterns relate to 45 cognitive traits which measure aspects such as fluid intelligence, delay discounting, and language/vocabulary comprehension.

Due to the complexity of human brain networks, we employ Generalized JisstPCA to jointly analyze the SC and FC tensors; we also compare our approach to PCA methods in Section~\ref{supp:realdata} of the Appendix. All hyperparameters are chosen as discussed in Section~\ref{sec:practice} of the Appendix, with the number of factors, $K$ chosen via the proportion of variance explained.  This yields $K=2$ factors which explain \( 85.5\% \) of the variance in the FC data and \( 77.9\% \) in the SC data; increasing $K$ further does not significantly increase the variance explained.  Further, the BIC deflation strategy selects SC and FC network loadings of rank 5 for both sets of factors.  The results from G-JisstPCA are visualized in Figure~\ref{fig:hcp_result_set1} (a), which shows a scatter plot of the estimated population components, $\boldsymbol{\hat{u}}_1$ and $\boldsymbol{\hat{u}}_2$. Figure~\ref{fig:hcp_result_loading_circle} shows circle plots of the top connections from the estimated SC and FC network loadings, $\boldsymbol{\hat{V}}_1 \boldsymbol{\hat{D}}_{SC,1} \boldsymbol{\hat{V}}_1^{\prime}$, $\boldsymbol{\hat{V}}_2 \boldsymbol{\hat{D}}_{SC,2}\boldsymbol{\hat{V}}_2^{\prime}$ and $\boldsymbol{\hat{W}}_1 \boldsymbol{\hat{D}}_{FC,1} \boldsymbol{\hat{W}}_1^{\prime}$,  $\boldsymbol{\hat{W}}_2 \boldsymbol{\hat{D}}_{FC,2}\boldsymbol{\hat{W}}_2^{\prime}$; estimated weighted adjacency matrices are shown in Figure~\ref{fig:hcp_result_set2} of the Appendix.

\begin{figure}[!htb]
    \centering
    \includegraphics[width=\textwidth]{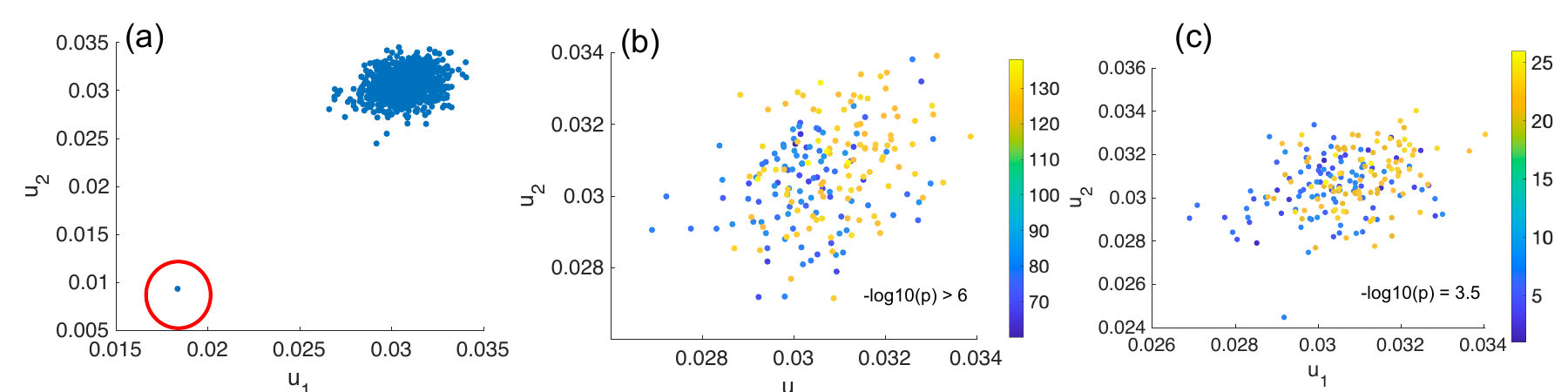}
    \caption{\small Results of Generalized JisstPCA applied to the HCP connectome data analyzing multi-subject Functional Connectivity (FC) and Structural Connectivity (SC). Panel (a) shows the scatter plot of the first two estimated population components, \( \boldsymbol{\hat{u}}_{1} \) and \( \boldsymbol{\hat{u}}_{2} \), exhibiting a major outlier.  Panels (b) and (c) show scatter plots of 200 subjects colored according to their measures on the English Reading and Penn Line Orientation tests, respectively, showing a clear statistically significant association.  }
    \label{fig:hcp_result_set1}
\end{figure}

From Figure~\ref{fig:hcp_result_set1} (a), we identify a clear outlier (marked in the red circle). Upon checking the intermediate outputs of PSC, we find that this subject has a much sparser SC due to misalignment between sMRI and dMRI. This demonstrates a potential application of G-JisstPCA - outlier brain network identification, an important problem in brain network analysis \citep{dey2022outlier}. Next, we explore the relationship between the joint factors obtained by JisstPCA and cognitive traits. In (b) and (c) of Figure~\ref{fig:hcp_result_set1}, we plot \( \boldsymbol{u}^{1} \) and \( \boldsymbol{u}^{2} \) for 200 subjects and color them according to their measures on the English Reading and Penn Line Orientation tests, respectively. For panel (b), the 200 subjects are selected based on their English Reading scores; the first 100 subjects have the highest scores and the second 100 have the lowest scores. The p-value testing the distribution difference between these two groups is displayed in the lower right corner. Similarly, 200 subjects are selected based on the Penn Line Orientation test for panel (c). The small p-values indicate that both SC and FC are significantly associated with the two traits under consideration. 

In Figure~\ref{fig:gjisstpca_correlation} of the Appendix, we correlate \( \boldsymbol{u}^{*}_1 \) and \( \boldsymbol{u}^{*}_2 \) with the 45 cognitive traits. From this result, we observe that 1) \( \boldsymbol{u}^{*}_1 \) correlates better with behavioral traits than does \( \boldsymbol{u}^{*}_2 \), and 2) most behavioral traits show a decent amount of correlation with the joint factors. We also examine how well we can predict the traits using both \( \boldsymbol{u}^{*}_1 \) and \( \boldsymbol{u}^{*}_2 \), and compared the prediction with principal components analysis (PCA). Figure \ref{fig:hcp_result_set2} of the Appendix shows the results, where we can see that the joint factors obtained by G-JisstPCA give much better prediction results, indicating that the joint components from SC and FC are more closely related to cognitive behavior traits.

\begin{figure}[!htb]
    \centering
    \includegraphics[width=0.8\textwidth]{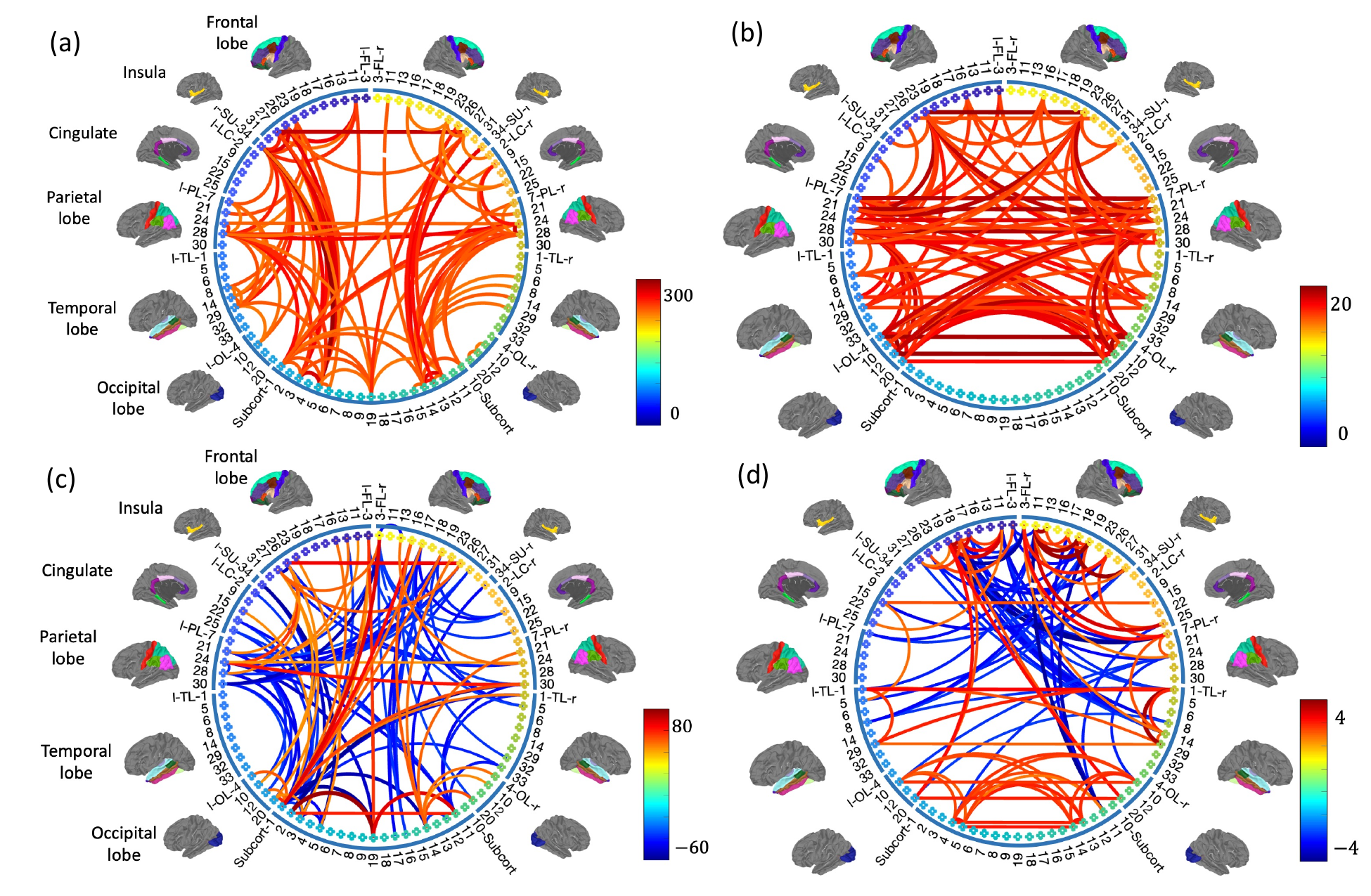}
    \caption{\small
    Visualizations of the estimated brain network loadings for Generalized JisstPCA applied to the HCP data.  Panels (a) and (b) show the first loadings while panels (c) and (d) show the second loads for Structural Connectivity (SC), panels (a) and (c), and Functional Connectivity (FC), panels (b) and (d).  These reveal many expected inter- and intra-hemisphere connections as well as major connectome variations that have associations with cognitive traits.}
\label{fig:hcp_result_loading_circle}
\end{figure}

\section{Discussion}
\label{discussion}
In this paper, we establish the first dimension reduction framework for the joint analysis of multi-modal populations of networks. Specifically, we proposed a novel joint integrative semi-symmetric tensor PCA (JisstPCA) model and associated algorithms to extract both the shared population factors across different modalities as well as low-rank network factors for each modality. We prove the convergence and statistical error bounds of our algorithms under the single-factor model, which improves or is comparable to prior results for single tensor PCA. Finally, a series of simulation studies validate the efficacy of our JisstPCA algorithm and its extensions; it also reveals intriguing structures in the human brain when applied to a real neuroimaging data example.

The joint network tensor PCA is a new problem with many potential fruitful future directions, a few of which we list as follows. First, it is challenging but interesting to extend our current theoretical results to the more general multi-factor model. Second, it is of interest to address common real challenges like missing data and heteroskedastic noise, by leveraging the recent advances in PCA. 
Third, when applied to neuroimaging data, our methods may be extended to incorporate brain connectome data subject to different parcellation. Lastly, there is great potential for applying our methods to integrate genomics data with brain connectomes, possibly revealing genetic effects on brain development. In summary, our work pioneers the analysis of multi-modal populations of networks that paves the way for many future advances.
\section*{Acknowledgments}
JL, LZ, and GIA acknowledge support from NSF NeuroNex-1707400, NIH 1R01GM140468, and NSF DMS-2210837. ZZ acknowledges support from NIH award R25DA058940.
\clearpage
\appendix
\section{Extensions, Additional Details, and Results}
\label{app1}
\subsection{Detailed Literature Review}
Both tensor algebra \citep{kolda2009tensor} and topics about tensors, or multiway arrays, have been studied comprehensively during the recent years. Many efficient tensor PCA algorithms under different tensor low-rank structures have been proposed in the literature with strong theoretical guarantees \citep{han2022tensor,zhang2018tensor,luo2021sharp,zhou2022optimal}, some even exploiting sparsity \citep{allen2012sparse,zhang2019optimal}. Beyond tensor PCA or tensor SVD, there is also rich literature on tensor completion \citep{yuan2016tensor,cai2019nonconvex,xia2021statistically} and tensor regression \citep{raskutti2019convex,hao2020sparse,zhang2020islet}. 

In terms of semi-symmetric tensor PCA, aside from \cite{weylandt2022multivariate}, some other prior works \citep{zhang2019tensor,winter2020multi,wang2014canonical} also studied this topic for analyzing brain connectomes and magnetic resonance spectroscopy data. In particular, \cite{zhang2019tensor} first proposes to study brain connectomes using a semi-symmetric tensor PCA approach, which is further extended by \cite{winter2020multi} to jointly analyze multi-scale graphs from different brain parcellations. \cite{wang2014canonical} considers a semi-symmetric and semi-nonnegative decomposition approach to perform Independent Component Analysis for magnetic resonance spectroscopy data. A comparison between existing modeling approaches and ours is included in Section 1.2.
\cite{jing2021community} considers a semi-symmetric Tucker low-rank model of multilayer networks, with a focus on community detection. We note that the phrase ``semi-symmetric tensor" is also used by \cite{deng2023correlation} to denote fourth-order tensors with two pairs of symmetric modes, different from the third-order tensors we are considering.

There also exist an extensive literature on data integration, which aims to find joint patterns across multiple sources of data. Most existing data integration methods focus on tabular data that can be arranged into matrices, including the JIVE \citep{lock2013joint} that decompose multiple data sets into sum of joint and individual principal components, the iPCA \citep{tang2021integrated} built upon the matrix-variate normal model, the multi-block PCA family \citep{abdi2013multiple,westerhuis1998analysis} that applies regular PCA on concatenated data sets after normalization, and many others. Extended from matrix integration problem, \cite{acar2011all, acar2014flexible, wu2018ctf, schenker2020flexible} consider joint factorization for tensors and tensor-matrix integration, based on the CP decomposition structure, as we mentioned in Section 1.2. In particular, \cite{acar2011all} solves coupled matrix and tensor factorization (CMTF) problem based on gradient methods, and \cite{acar2014flexible} extends this to a more general version that incorporate additional linear or nonlinear constraints in CMTF problem; also see some other extensions in \cite{wu2018ctf,schenker2020flexible}. As for tensor-tensor type of integration, \cite{genicot2016coupled} proposes a joint tensor factorization method RCTF that can extract shared and unshared factors as well as robust components between integrated tensors. And \cite{farias2016exploring} uses Bayesian framework to define flexible coupling models and uncovers joint factors in terms of joint MAP estimators. In addition, some other algorithms, such as CIF-OPT and HOPM, have also been proposed to deal with joint tensor factorization problem or tensor canonical correlation analysis in different scenarios \citep{lu2020exploring, chen2021tensor}.

Lastly, there also exist many prior works devoted to collectively or integratively analyzing multiple networks and extracting interpretable knowledge from them. In particular, one line of literature focuses on the analysis of multiplex networks \citep{mucha2010community,paul2020spectral,macdonald2022latent} where one has the access to a collection of networks associated with the same set of nodes, such as networks measured across a population \citep{wang2019common,paul2020random,pavlovic2020multi} or different time points \citep{mucha2010community,kim2018review}. Many existing works approach this problem by assuming a latent space model with a joint component and individual components that capture the heterogeneity across networks; under each specific model, estimation methods and statistical guarantees for identifying the latent structures are provided \citep{paul2020spectral,macdonald2022latent,zhang2020flexible}. On the other hand, another line of work is concerned with the integrative analysis of multi-modal networks, such as the functional and structural brain connectivity networks based off the fMRI and sMRI data \citep{sui2012review,yao2015review,cole2021surface}. Different from these prior works, we aim to propose a novel statistical method to jointly analyze populations of multimodal networks, extracting meaningful insights both across the population and linking the functional and structural connectivity of the brain.

\subsection{Relationship between the JisstPCA Model and the Tucker and CP Models}
Now we give a detailed correspondence between our model and the Tucker/CP low-rank models. For simplicity, we will focus on the single tensor case, while it is straightforward to extend the model connections from a single tensor to joint factorization of multiple tensors. 

Recall our multi-factor semi-symmetric tensor decomposition $\mathcal{X} = \sum_{k = 1}^{K} d^*_{x,k} \cdot \boldsymbol{V}_{k}^*\boldsymbol{V}^{*\prime}_{k} \circ \boldsymbol{u}_{k}^* \in \mathbb{R}^{p \times p \times N}+\cE_x$. When the factors are mutually orthogonal: $\bV_i^*\perp \bV_j^*$, $\bu_i^*\perp \bu_j^*$ for $1\leq i\neq j\leq K$, we can also write $\cX$ as the following low-rank Tucker decomposition plus noise: $\mathcal{X} = \mathcal{S}^* \times_{1} \boldsymbol{U}_{1}^* \times_{2} \boldsymbol{U}_{2}^* \times_{3} \boldsymbol{U}_{3}^* + \cE_x$. Here, \begin{align*}
    \boldsymbol{U}_{1}^* = \boldsymbol{U}_{2}^* = \left[\boldsymbol{V}_{1}^*, \cdots, \boldsymbol{V}_{K}^*\right] \in \mathbb{R}^{p \times r},\quad \boldsymbol{U}_{3}^* = \left[\boldsymbol{u}_{1}^*, \cdots, \boldsymbol{u}_{k}^*\right] \in 
    \mathbb{R}^{N \times K},
\end{align*}
where $r = \sum r_{k}$ is the sum of ranks of each $\boldsymbol{V}_{k}^*$, $k = 1, \cdots, K$. The core tensor $\mathcal{S}\mathbb{R}^{r \times r \times K}$ satisfies
\begin{align*}
    \mathcal{S}^*_{ijk} = d_{x,k}^* \cdot \mathbbm{1}\left\{\sum_{l = 1}^{k-1} r_{l} + 1 \leq i = j \leq \sum_{l = 1}^{k} r_{l} \right\}.
\end{align*}
That is, the $k$ slice of the core tensor $\mathcal{S}^*$ is a diagonal matrix with only $r_k$ non-zero diagonal entries. 

In addition, we can also write $\cX$ under the CP decomposition model $\mathcal{X} = \sum\limits_{i = 1}^{r} \lambda_{i} \cdot \boldsymbol{a}_{i} \circ \boldsymbol{b}_{i} \circ \boldsymbol{c}_{i} + \cE_x$ with $r=\sum_{k=1}^Kr_k$. If we further denote $\boldsymbol{\lambda} = (\lambda_{1}, \cdots, \lambda_{r})$, $\boldsymbol{A} = [\boldsymbol{a}_{1}, \cdots, \boldsymbol{a}_{r}], \boldsymbol{B} = [\boldsymbol{b}_{1}, \cdots, \boldsymbol{b}_{r}]$ and $\boldsymbol{C} = [\boldsymbol{c}_{1}, \cdots, \boldsymbol{c}_{r}]$, they would satisfy the following:
\begin{align*}
    &\boldsymbol{\lambda} = (\underbrace{d_{x,1}, \cdots, d_{x,1}}_{r_1}, \underbrace{d_{x,2}, \cdots, d_{x,2}}_{r_2}, \cdots, \underbrace{d_{x,k}, \cdots, d_{x,k}}_{r_k}) \in \mathbb{R}^{r},\\
    &\boldsymbol{A} = \boldsymbol{B} = \left[\boldsymbol{V}_{1}, \cdots, \boldsymbol{V}_{K}\right] \in \mathbb{R}^{p \times r},\\
    &\boldsymbol{U} = [\underbrace{u_1, \cdots, u_1}_{r_1}, \underbrace{u_2, \cdots, u_2}_{r_2}, \cdots, \underbrace{u_k, \cdots, u_k}_{r_k}] \in \mathbb{R}^{N \times r}.
\end{align*}
In summary, both the CP and Tucker low-rank models are closely related to our model; however, as discussed in Section 1.2, they both add additional, undesirable constraints (orthogonality or incoherence) to the factors and hence fall short for our network modeling purposes.

\subsection{Additional Notations}
Here, we provide additional details of some notations briefly introduced in Section 2.1. Suppose $\mathcal{X} \in \mathbb{R}^{p_1 \times p_2 \times p_3}$ is a general third-order tensor, for a matrix $\boldsymbol{U} \in \mathbb{R}^{p_1 \times r_1}$, the (marginal) multiplication $\times_1$ of tensor and matrix is $\mathcal{X} \times_1 \boldsymbol{U}\in \bbR^{r_1\times p_2\times p_3}$ satisfying:
\begin{equation*}
     \left(\mathcal{X} \times_1 \boldsymbol{U}\right)_{i,j,k} = \sum\limits_{l = 1}^{p_1} \mathcal{X}_{ljk}\boldsymbol{U}_{il}.
\end{equation*}
 And $\times_2$, $\times_3$ can be defined similarly. The matricization of $\mathcal{X}$ by the $k$th-mode $\mathcal{M}_{k}(\mathcal{X})$ is a $p_k \times p_{-k}$ matrix, where $p_{-k} = \prod\limits_{i \neq k} p_{i}$. Elementwisely,  $\mathcal{M}_{k}(\mathcal{X})$ can be written as 
 \begin{equation*}
     \left[\mathcal{M}_{1}(\mathcal{X})\right]_{i, (k-1)p_2+j} = \mathcal{X}_{ijk},\, \left[\mathcal{M}_{2}(\mathcal{X})\right]_{j, (i-1)p_3+k} = \mathcal{X}_{ijk},\, \left[\mathcal{M}_{3}(\mathcal{X})\right]_{k, (j-1)p_1+i} = \mathcal{X}_{ijk},
 \end{equation*} 
 where $1 \leq i \leq p_1,\, 1 \leq j \leq p_2,\, 1 \leq k \leq p_3$. For two tensors of the same order and dimension, $\mathcal{X}, \mathcal{Y} \in \mathbb{R}^{p_1 \times p_2 \times p_3}$, the inner product of tensors is $\langle \mathcal{X}, \mathcal{Y} \rangle = \sum\limits_{i = 1}^{p_1}\sum\limits_{j = 1}^{p_2}\sum\limits_{k = 1}^{p_3} \mathcal{X}_{ijk} \mathcal{Y}_{ijk}.$ And tensor Frobenius norm is the square root of the inner product with itself, i.e. $\|\mathcal{X}\|_{F} = \sqrt{\langle \mathcal{X}, \mathcal{X} \rangle} = \sqrt{\sum\limits_{i = 1}^{p_1}\sum\limits_{j = 1}^{p_2}\sum\limits_{k = 1}^{p_3} \mathcal{X}_{ijk}^{2}}.$
We follow the definition of sub-Gaussian random variables in \citep{wainwright2019high} and say that $X$ is sub-Gaussian-$\sigma$ if and only if $\bbE[e^{\lambda (X-\bbE X)}]\leq e^{\lambda^2\sigma^2/2}$ for all $\lambda\in \bbR$.

\subsection{Additional Algorithms \& Details}\label{sec:morealg}
In this section, we provide the detailed additional algorithms mentioned in the main paper, including the multi-factor JisstPCA (with and without BIC-based rank selection), different deflation schemes one can apply, matrix-tensor JisstPCA, selection strategies for the number of factors $K$ and weight parameter $\lambda\in (0,1)$, and our comparison baselines iHOOI and iHOSVD (integrated versions of HOOI and HOSVD).
\subsubsection{Multi-factor PCA}
\begin{algorithm}
\caption{Multi-factor JisstPCA with Subtraction Deflation and Prespecified Ranks}\label{alg:multi-Jisst}
\begin{itemize}
\item Input: $\mathcal{X}, \mathcal{Y}$, number of factors $K$, $\boldsymbol{r}_x,\,\boldsymbol{r}_y\in \bbR^{K}$, and maximum iteartion $t_{\max}$ 
\item Initialization: Let $k = 1$, $\mathcal{X}^{1} = \mathcal{X}$, and $\mathcal{Y}^{1} = \mathcal{Y}$.
\item While $k \leq K$:
\begin{itemize}
    \item Let $\lambda = \frac{\|\cX^{(k)}\|_F}{\|\cX^{(k)}\|_F+\|\cY^{(k)}\|_F}$.
    \item Apply Single-Factor JisstPCA (Algorithm \ref{alg:single_tt}) on $\cX^{(k)}$, $\cY^{(k)}$, with ranks $r_{x,k},\, r_{y,k}$, $\lambda$, and maximum iteration $t_{\max}$ to obtain $\hat{d}_{x,k}, \hat{d}_{y,k}, \hat{\boldsymbol{V}}_{k}, \hat{\boldsymbol{W}}_{k}, \hat{\boldsymbol{u}}_{k}$. 
    \item Apply subtract deflation to obtain $\mathcal{X}^{k+1}, \mathcal{Y}^{k+1}$ as
    \begin{align*}
        & \mathcal{X}^{k+1} = \mathcal{X}^{k} - \hat{d}^x_{k} \cdot \hat{\boldsymbol{V}}_{k} \hat{\boldsymbol{V}}_{k}^{\prime} \circ \hat{\boldsymbol{u}}_{k}\\
        & \mathcal{Y}^{k+1} = \mathcal{Y}^{k} - \hat{d}^y_{k} \cdot \hat{\boldsymbol{W}}_{k} \hat{\boldsymbol{W}}_{k}^{\prime} \circ \hat{\boldsymbol{u}}_{k}.
    \end{align*}
    \item $k = k + 1$.
\end{itemize}
\item \textbf{return} $\{\hat{\bu}_k,\,\hat{\bV}_k,\,\hat{\bW}_k, \hat{d}_{x,k},\,\hat{d}_{y,k}\}_{k=1}^K$.
\end{itemize}
\end{algorithm}

The subtraction deflation \eqref{eq:sub_def} is computationally efficient and imposes no extra constraint on different factors, such as orthogonality. However, in some cases, orthogonality on certain modes might be desirable due to prior belief or for interpretation purposes. To accommodate for this, one might consider the projection deflation \citep{mackey2008deflation} where the tensor data is projected onto the orthogonal complement of the space spanned by previously estimated factors. In particular, when orthogonality is required for both the joint factor ($\bu_i\perp\bu_j$ for $i\neq j$) and individual factors ($\bV_i\perp\bV_j$, $\bW_i\perp\bW_j$, for $i\neq j$), we let 
\begin{equation}\label{eq:proj_def}
\begin{aligned}
    & \mathcal{X}^{k+1} = \mathcal{X}^{k} \times_{1} \left(\boldsymbol{I}_{p}-\hat{\boldsymbol{V}}_{k}\hat{\boldsymbol{V}}_{k}^{\prime} \right) \times_2 \left(\boldsymbol{I}_{p}-\hat{\boldsymbol{V}}_{k}\hat{\boldsymbol{V}}_{k}^{\prime} \right) \times_3 \left(\boldsymbol{I}_{N}-\hat{\boldsymbol{u}}_{k}\hat{\boldsymbol{u}}_{k}^{\prime} \right)\\
    & \mathcal{Y}^{k+1} = \mathcal{Y}^{k} \times_{1} \left(\boldsymbol{I}_{q}-\hat{\boldsymbol{W}}_{k}\hat{\boldsymbol{W}}_{k}^{\prime} \right) \times_2 \left(\boldsymbol{I}_{q}-\hat{\boldsymbol{W}}_{k}\hat{\boldsymbol{W}}_{k}^{\prime} \right) \times_3 \left(\boldsymbol{I}_{N}-\hat{\boldsymbol{u}}_{k}\hat{\boldsymbol{u}}_{k}^{\prime} \right).
\end{aligned}
\end{equation}
Furthermore, one can also enforce pairwise orthogonality only for the joint factors or for the individual factors through a ``partial projection deflation" scheme. 
\begin{equation}\label{eq:proj_def_u_supp}
\begin{aligned}
    & \mathcal{X}^{k+1} = \left(\mathcal{X}^{k} - \hat{d}_{x,k} \cdot \hat{\boldsymbol{V}}_{k}\hat{\boldsymbol{V}}_{k}^{\prime} \circ \hat{\boldsymbol{u}}_{k}\right)\times_3 \left(\boldsymbol{I}_{N}-\hat{\boldsymbol{u}}_{k}\hat{\boldsymbol{u}}_{k}^{\prime} \right)\\
    & \mathcal{Y}^{k+1} = \left(\mathcal{Y}^{k} - \hat{d}_{y,k} \cdot \hat{\boldsymbol{W}}_{k}\hat{\boldsymbol{W}}_{k}^{\prime} \circ \hat{\boldsymbol{u}}_{k}\right) \times_3 \left(\boldsymbol{I}_{N}-\hat{\boldsymbol{u}}_{k}\hat{\boldsymbol{u}}_{k}^{\prime} \right);
\end{aligned}
\end{equation}
Or, when only the individual factors need to be mutually orthogonal, we can similarly let 
\begin{equation}\label{eq:proj_def_vw}
\begin{aligned}
    & \mathcal{X}^{k+1} = \left(\mathcal{X}^{k} - \hat{d}_{x,k} \cdot \hat{\boldsymbol{V}}_{k}\hat{\boldsymbol{V}}_{k}^{\prime} \circ \hat{\boldsymbol{u}}_{k}\right) \times_{1} \left(\boldsymbol{I}_{p}-\hat{\boldsymbol{V}}_{k}\hat{\boldsymbol{V}}_{k}^{\prime} \right) \times_2 \left(\boldsymbol{I}_{p}-\hat{\boldsymbol{V}}_{k}\hat{\boldsymbol{V}}_{k}^{\prime} \right)\\
    & \mathcal{Y}^{k+1} = \left(\mathcal{Y}^{k} - \hat{d}_{y,k} \cdot \hat{\boldsymbol{W}}_{k}\hat{\boldsymbol{W}}_{k}^{\prime} \circ \hat{\boldsymbol{u}}_{k}\right) \times_{1} \left(\boldsymbol{I}_{q}-\hat{\boldsymbol{W}}_{k}\hat{\boldsymbol{W}}_{k}^{\prime} \right) \times_2 \left(\boldsymbol{I}_{q}-\hat{\boldsymbol{W}}_{k}\hat{\boldsymbol{W}}_{k}^{\prime} \right).
\end{aligned}
\end{equation}
\begin{rmk}
    Applying multi-factor JisstPCA with projection deflation \eqref{eq:proj_def} yield mutually orthogonal factors: $\hat{\bu}_i\perp \hat{\bu}_j$ for $i\neq j$; when using the partial projection deflation \eqref{eq:proj_def_u} or \eqref{eq:proj_def_vw}, the factors with projection would be mutually orthogonal. 
\end{rmk}
\subsubsection{Selection of Hyperparameters}\label{sec:practice}
Our model includes a number of hyperparameters that must be specified or tuned in practice: the number of factors $K$, the ranks of each factor $r_{x,k}$, $r_{y,k}$, $1\leq k\leq K$, and the integrative scaling parameter $\lambda$. 
First, there are many widely employed approaches to estimate the number of PCA factors, $K$ \citep{wold1978cross,jolliffe2016principal,dobriban2019deterministic,donoho2023screenot}; these methods are also applicable for JisstPCA. Noting that $\bbE\mathcal{M}_{3}(\mathcal{X}) = \mathcal{M}_{3}(\mathcal{X}^*)$ and $\mathrm{rank}(\mathcal{M}_{3}(\mathcal{X}^*)) = K$, one possible way to estimate $K$ is to apply usual PCA-type methods on $\mathcal{M}_{3}(\mathcal{X})$. 
Another possible approach to estimate $K$ is based on the cumulative proportion of variance explained by the factors. Note that the proportion of variance explained by each single factor is given by $d_{x,k}^2 / \| \mathcal{X} \|_{F}^{2}$ for $\mathcal{X}$ and $d_{y,k}^{2} / \| \mathcal{Y} \|_{F}^{2}$ for $\mathcal{Y}$.  But as the components are not orthogonal across factors (unless projection deflation is employed), we cannot simply add the variance explained by each factor to get the cumulative variance explained \citep{allen2012regularized}.  Instead, we must calculate this by projecting out the effect of all components from the $k$ factors: 
\begin{rmk}\label{prop_var_exp}
Let $\mathbf{P}_{k}^{(\mathbf{U})} = \mathbf{U}_{k} (\mathbf{U}_{k}^{\prime} \mathbf{U}_{k} )^{-1} \mathbf{U}_{k}^{\prime}$, with $\mathbf{U}_{k} = [ \mathbf{u}_1, \ldots, \mathbf{u}_{k} ]$, and let $\mathbf{P}_{k}^{(\mathbf{V})} = \tilde{\mathbf{V}}_{k} (\tilde{\mathbf{V}}_{k}^{\prime} \tilde{\mathbf{V}}_{k} )^{-1} \tilde{\mathbf{V}}_{k}^{\prime}$, with $\tilde{\mathbf{V}}_{k} = [\mathbf{V}_{1}, \ldots, \mathbf{V}_{k}]$; define $\mathbf{P}_{k}^{(\mathbf{W})}$ and $\tilde{\mathbf{W}}_{k}$ analogously.  Then, the cumulative proportion of variance explained by the first $k$ JisstPCA factors in $\mathcal{X}$ and $\mathcal{Y}$ is given by $\frac{\| \mathcal{X} \times_{1} \mathbf{P}_{k}^{(\mathbf{V})} \times_2 \mathbf{P}_{k}^{(\mathbf{V})} \times_3  \mathbf{P}_{k}^{(\mathbf{U})} \|_{F}^{2}}{\| \mathcal{X}\|_{F}^{2}}$ and $\frac{\| \mathcal{Y} \times_{1} \mathbf{P}_{k}^{(\mathbf{W})} \times_2 \mathbf{P}_{k}^{(\mathbf{W})} \times_3  \mathbf{P}_{k}^{(\mathbf{U})} \|_{F}^{2}}{\| \mathcal{Y}\|_{F}^{2}}$ respectively.  
\end{rmk}

Second, $\lambda$ is important when $\mathcal{X}$ and $\mathcal{Y}$ have different scales. To study the effect of $\lambda$, we generate data from two-factor JisstPCA models with different SNR levels and apply JisstPCA with a range of $\lambda$. The detailed set-up is the same as the unstructured simulation in Section \ref{sec:simulation} in the main paper, where the SNR of both tensors $\cX$ and $\cY$ are the same. The factor estimation errors are summarized in Figure \ref{fig:lambda_selection}, suggesting that when both modalities have similar SNR levels, $\lambda = \dfrac{\|\mathcal{X}^*\|_{F}}{\|\mathcal{X}^*\|_{F} + \|\mathcal{Y}^*\|_{F}}$ is likely the optimal choice (red line). Without the knowledge of the true tensor norms, we suggest using $\lambda = \dfrac{\|\mathcal{X}\|_{F}}{\|\mathcal{X}\|_{F} + \|\mathcal{Y}\|_{F}}$ as a surrogate (blue line); this is also our default selection of $\lambda$ throughout the empirical studies in this paper. We can also see from Figure \ref{fig:lambda_selection} that JisstPCA is not sensitive to the choice of $\lambda$ as long as we do not use extreme value. when both modalities have similar noise levels, we suggest using $\lambda = \dfrac{\|\mathcal{X}\|_{F}}{\|\mathcal{X}\|_{F} + \|\mathcal{Y}\|_{F}}$; we employ this scheme in all our empirical studies.  

\begin{figure}
    \centering
    \includegraphics[width = \textwidth]{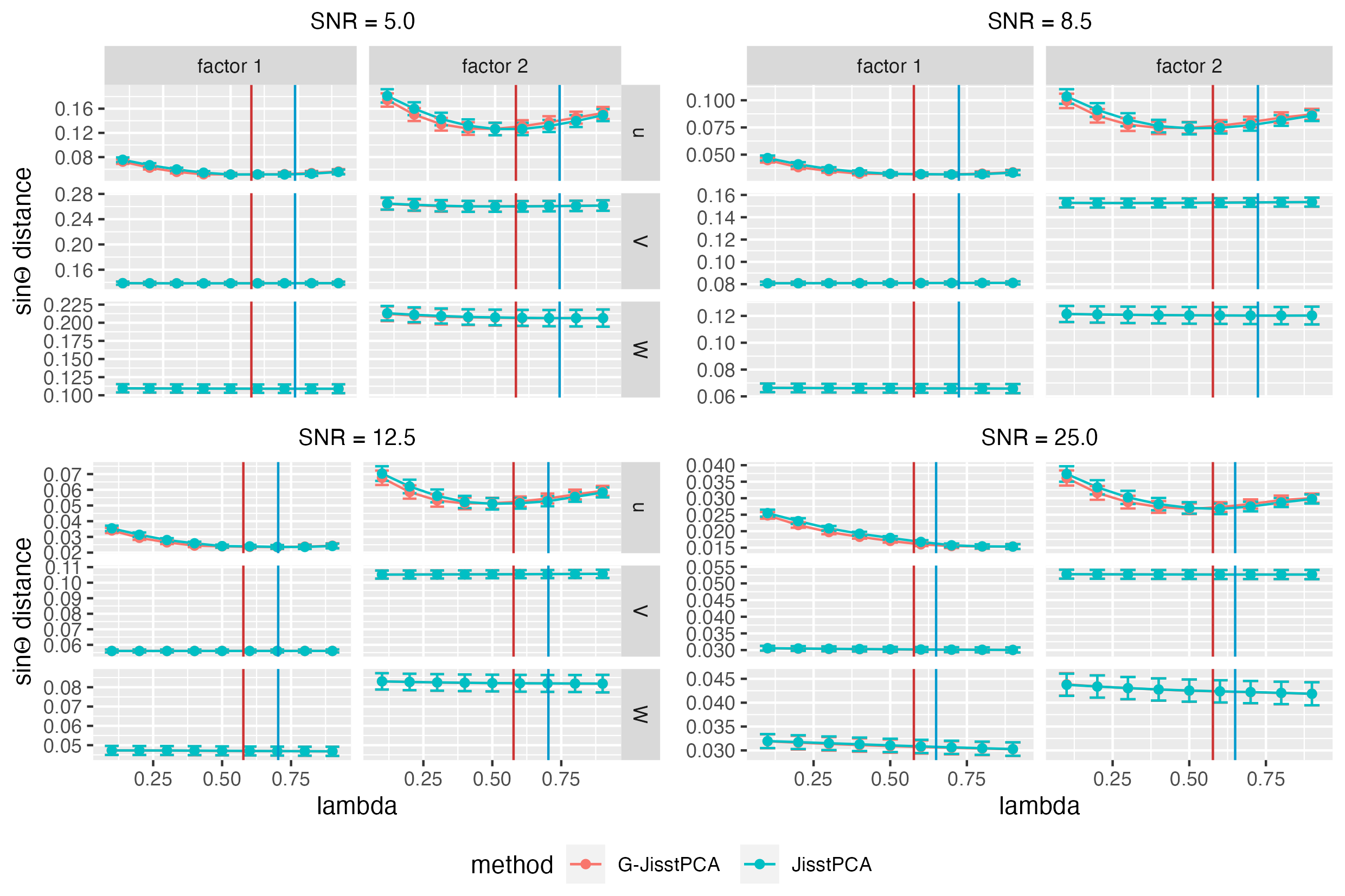}
    \caption{The effect of $\lambda$ values on the estimation errors ($\sin\Theta$ distances in spectral norm) of all factors using JisstPCA and G-JisstPCA. The red vertical line is the conjectured oracle $\lambda = \frac{\|\cX^*\|_F}{\|\cX^*\|_F+\|\cY^*\|_F}$ when the noise level is the same, and the blue vertical line is its estimate: $\frac{\|\cX\|_F}{\|\cX\|_F+\|\cY\|_F}$. The estimated $\lambda$ is closer to the oracle $\lambda$ when SNR increases; in addition, the performance of (G-)JisstPCA seems robust when changing $\lambda$.}
    \label{fig:lambda_selection}
\end{figure}

Finally, selecting the ranks of each factor is more complicated and this choice has more influence on the results. In the multi-factor ($K>1$) case, we need to select $2K$ total ranks. Many prior works approach this problem by finding the minimizer of the Bayesian information criterion (BIC) \citep{allen2012sparse, zhou2022optimal, hu2022generalized}. However, directly minimizing the BIC for all $2K$ ranks simultaneously is a combinatoral problem and computationally infeasible for large $K$. Instead, and inspired by \citep{allen2012sparse}, we propose a ``single-factor BIC + deflation" scheme, called ``BIC deflation".  This greedy approach successively applies single-factor BIC to select the rank for each factor, and then single-factor JisstPCA with the selected rank and deflation to get residual data for extracting the next factor.

In particular, given data $\cX,\,\cY$ and potential single-factor ranks $r_{x},\,r_{y}$, we first apply single-factor JisstPCA to obtain tensor estimates $\hat{\mathcal{X}}(r_x),\,\hat{\mathcal{Y}}(r_y)$, and then compute the BIC for this rank combination as 
\begin{equation}\label{eq:bic}
\begin{aligned}
    \mathrm{BIC}(r_{x}, r_{y}) 
    &=p^2N\log{\|\mathcal{X}-\hat{\mathcal{X}}(r_{x})\|_F^2+ q^2N\log\|\mathcal{Y}-\hat{\mathcal{Y}}}(r_{y})\|_F^2\\
    &\quad+(pr_{x}+qr_{y})\log((p^{2}+q^{2})N)+C(p,q,N).
\end{aligned}
\end{equation}
For a grid of potential rank combinations, we select $(r_x,r_y)$ that minimizes $\mathrm{BIC}(r_{x}, r_{y})$ for the current factor. The detailed ``BIC deflation" scheme is summarized in Algorithm \ref{bic_def}.
\begin{algorithm}
\caption{Multi-factor JisstPCA with Subtraction Deflation and BIC-selected Ranks}\label{bic_def}
\begin{itemize}
\item Input: $\mathcal{X}, \mathcal{Y}$, and number of factors $K$, maximum ranks $r_{x,\max}\leq p$, $r_{y,\max}\leq q$, and maximum iteration $t_{\max}$.
\item Initialization: Let $k = 1$, $\mathcal{X}^{1} = \mathcal{X}$, and $\mathcal{Y}^{1} = \mathcal{Y}$.
\item While $k \leq K$:
\begin{itemize}
    \item Let $\lambda = \frac{\|\cX^{(k)}\|_F}{\|\cX^{(k)}\|_F+\|\cY^{(k)}\|_F}$.
    \item Select rank for factor $k$ via BIC: \\For $i = 1, \cdots, r_{x,\max}$ and $j = 1, \cdots, r_{y,\max}$:
    \begin{itemize}
        \item Apply single-factor JisstPCA (Algorithm \ref{alg:single_tt}) with ranks $(i,j)$, $\lambda$, and maximum iteration $t_{\max}$.
        \item Calculate BIC($i, j$) in equation \eqref{eq:bic} for $\mathcal{X}^{k}, \mathcal{Y}^{k}$.
    \end{itemize}
    Let $(\hat{r}_{x,k}, \hat{r}_{y,k}) = \underset{i \in \{1, \cdots, r_{x,\max}\}, j \in \{1, \cdots, r_{y,\max}\}}{\arg\min} \text{BIC}(i, j)$, and save the output of single-factor JisstPCA with ranks $(\hat{r}_{x,k}, \hat{r}_{y,k})$ as $\hat{d}_{x,k}, \hat{d}_{y,k}, \hat{\boldsymbol{V}}_{k}, \hat{\boldsymbol{W}}_{k}, \hat{\boldsymbol{u}}_{k}$. 
    \item Apply subtract deflation to obtain $\mathcal{X}^{k+1}, \mathcal{Y}^{k+1}$ as
    \begin{align*}
        & \mathcal{X}^{k+1} = \mathcal{X}^{k} - \hat{d}^x_{k} \cdot \hat{\boldsymbol{V}}_{k} \hat{\boldsymbol{V}}_{k}^{\prime} \circ \hat{\boldsymbol{u}}_{k}\\
        & \mathcal{Y}^{k+1} = \mathcal{Y}^{k} - \hat{d}^y_{k} \cdot \hat{\boldsymbol{W}}_{k} \hat{\boldsymbol{W}}_{k}^{\prime} \circ \hat{\boldsymbol{u}}_{k}.
    \end{align*}
    \item $k = k + 1$.
\end{itemize}
\item \textbf{return} $\hat{\boldsymbol{r}}^{x} = (\hat{r}_{x,1}, \cdots, \hat{r}_{x,K})$, $\hat{\boldsymbol{r}}_{y} = (\hat{r}_{y,1}, \cdots, \hat{r}_{y,K})$, $\{\hat{\bu}_k,\,\hat{\bV}_k,\,\hat{\bW}_k, \hat{d}_{x,k},\,\hat{d}_{y,k}\}_{k=1}^K$.
\end{itemize}
\end{algorithm}

\subsubsection{Extensions of JisstPCA}
Here, we provide more details on the matrix-tensor JisstPCA model, the generalized JisstPCA model, as well as their associated algorithms. We focus on the single-factor algorithms, while the multi-factor extensions of them can be similarly defined as our deflation-procedure in Algorithm \ref{alg:multi-Jisst}.

\paragraph{Matrix-Tensor JisstPCA}
One important extension of our JisstPCA framework is to jointly analyze vector-valued covariates, organized as a matrix, together with network data organized as tensors. In neuroimaging studies, for example, we may additionally observe feature vectors such as genomics, demographic, or behavioral traits for each subject. In this case, we would like to find joint population factors shared between the network data and vector-valued covariates, as well as individual low-rank factors for networks and for features separately. 
Specifically, suppose we observe a semi-symmetric tensor $\mathcal{X} \in \mathbb{R}^{p \times p \times N}$ representing a population of networks, and a matrix $\boldsymbol{Y} \in \mathbb{R}^{q \times N}$ representing features measured for the same population. We consider the following joint matrix-tensor PCA model: 
\begin{equation}\label{mt_xy}
\begin{aligned}
    \mathcal{X} = \sum\limits_{k = 1}^{K} d_{x,k} \cdot \boldsymbol{V}_{k}^*\boldsymbol{V}_{k}^{*\prime} \circ \boldsymbol{u}_{k}^* + \mathcal{E}_{x},\quad\boldsymbol{Y} = \sum\limits_{k = 1}^{K} d_{y,k} \cdot \boldsymbol{w}_{k}^* \boldsymbol{u}_{k}^{*'} + \boldsymbol{E}_{y},
\end{aligned}
\end{equation}
where $\cE_x\in \bbR^{p\times p\times N}$, $\bE_y\in \bbR^{p\times N}$ are observational noise, $\boldsymbol{u}_{k}$ is the population factor, and $\boldsymbol{w}_{k}^*\in \mathbb{S}^{q - 1}$ is the $k$th feature factor associated with the network factor $\bV_k^*$. 
By contrasting $\boldsymbol{w}_{k}^*$ with $\bV_k^*$, one can establish a connection between the features and the network, which in the neuroimaging example corresponds to how certain gene expression or human behavior is associated with the brain connectome. As in the JisstPCA model, we do not impose orthogonality constraints on $\boldsymbol{w}_{k}^*$'s. Although orthogonality is required for single matrix decomposition to ensure identifiability, it is not necessary for \eqref{mt_xy}, since matrix $\boldsymbol{Y}$ shares factors $\bu_k^*$ with tensor $\cX$. 

As discussed earlier in Section \ref{sec:relatedwork}, although there are many prior methods that jointly factorize matrices and tensors \citep{acar2011all, acar2014flexible, fu2015joint, wu2018ctf, schenker2020flexible}, they are often based on the CP decomposition, which could be less appropriate for our study of tensor networks. Under model \eqref{mt_xy}, in order to estimate factors $\bV_k^*$, $\bw_k^*$, and $\bu_k^*$ from noisy observations $\cX$, $\cY$, we consider a similar strategy to the JisstPCA algorithm, which consists of successive deflation and single-factor PCA based on joint power iteration. The detailed single-factor matrix-tensor JisstPCA algorithm is summarized in Algorithm \ref{single_mt}.

\begin{algorithm}
\caption{Single-factor Matrix-Tensor JisstPCA}\label{single_mt}
\begin{itemize}
    \item Input: $\mathcal{X}$, $\boldsymbol{Y}$, $r_x$, $\lambda$, and maximum iteration $t_{\max}$.
    \item Initialization: Let $t = 0$, and $\boldsymbol{u}^{(0)} = \text{Leading left singular vector of } \left[\lambda \mathcal{M}_{3}( \mathcal{X}), (1 - \lambda) \boldsymbol{Y}^{\prime} \right]$.
    \item \textbf{repeat} until $t = t_{\max}$ or convergence:
    \begin{align*}
        & \boldsymbol{V}^{(t+1)} = \text{Leading $r_x$ left singular vectors of } \mathcal{X} \times_3 \boldsymbol{u}^{(t)}\\
        & \boldsymbol{w}^{(t+1)} = \dfrac{\boldsymbol{Y}\boldsymbol{u}^{(t)}}{\|\boldsymbol{Y}\boldsymbol{u}^{(t)}\|_{2}}\\
        & \boldsymbol{u}^{(t+1)} = \dfrac{\lambda[\mathcal{X}; \boldsymbol{V}^{(t+1)}] + (1-\lambda) \boldsymbol{Y}^{\prime}\boldsymbol{w}^{(t+1)}}{\left\|\lambda[\mathcal{X}; \boldsymbol{V}^{(t+1)}] + (1-\lambda) \boldsymbol{Y}^{\prime}\boldsymbol{w}^{(t+1)}\right\|_{2}}\\
        & t = t+1
    \end{align*}
    \item \textbf{return} $\hat{\boldsymbol{u}} = \boldsymbol{u}^{(t)}$, $\hat{\boldsymbol{V}} = \boldsymbol{V}^{(t)}$, $\hat{\boldsymbol{w}} = \boldsymbol{w}^{(t)}$; $\hat{d}_x = \langle \mathcal{X}, \hat{\boldsymbol{V}}\hat{\boldsymbol{V}}^{\prime} \circ \hat{\boldsymbol{u}} \rangle / r_{x}$, $\hat{d}_y = \hat{\boldsymbol{w}}^{\prime} \boldsymbol{Y} \hat{\boldsymbol{u}}$; $\hat{\mathcal{X}} = \hat{d}_x \cdot \hat{\boldsymbol{V}} \hat{\boldsymbol{V}}^{\prime} \circ \hat{\boldsymbol{u}}$, $\hat{\boldsymbol{Y}} = \hat{d}_y \cdot \hat{\boldsymbol{w}} \circ \hat{\boldsymbol{u}}$.
\end{itemize}
\end{algorithm}
\paragraph{Generalized JisstPCA} In addition, we consider a generalized JisstPCA model to accommodate for potential different eigenvalues within each network factor. Specifically, we extend \eqref{eq:sst_d} to the following: 
\begin{equation}\label{gen_model_supp}
\begin{aligned}
     \mathcal{X} = \sum\limits_{k = 1}^{K} \boldsymbol{V}_{k}^*\boldsymbol{D}^{*}_{x,k} \boldsymbol{V}^{*\prime}_{k} \circ \boldsymbol{u}_{k}^* + \mathcal{E}_{x},\quad \mathcal{Y} = \sum\limits_{k = 1}^{K} \boldsymbol{W}_{k}^*\boldsymbol{D}^{*}_{y,k}\boldsymbol{W}^{*\prime}_{k} \circ \boldsymbol{u}_{k}^* + \mathcal{E}_{y}.
\end{aligned}
\end{equation}
where $\bD^*_{x,k}\in \bbR^{r_{x,k}\times r_{x,k}}$ and $\bD^*_{y,k}\in \bbR^{r_{y,k}\times r_{y,k}}$ are diagonal matrices, replacing the scalar eigenvalues $d^*_{x,k}$, $d^*_{y,k}$ in \eqref{eq:sst_d}.
To estimate factors in this general JisstPCA model, 
we still propose a power iteration algorithm for the single-factor case and apply a successive deflation scheme for multi-factor models. The main change we make to the single-factor JisstPCA algorithm is that we update the diagonal matrix $\bD_x$ and $\bD_y$ within each iteration due to its increased importance, and we use them weight the columns of $\bV$, $\bW$ when updating $\bu$. 
We term this new algorithm the generalized JisstPCA (G-JisstPCA) algorithm, whose detailed procedures are summarized in Algorithm \ref{single_diag}.
\begin{algorithm}
\caption{Generalized Single-Factor JisstPCA}\label{single_diag}
\begin{itemize}
\item Input: $\mathcal{X}$, $\mathcal{Y}$, $r_x$, $r_y$, $\lambda$, and maximum iteration $t_{\max}$.
\item Initialization: Let $t = 0$, and $\boldsymbol{u}^{(0)} = \text{Leading left singular vector of } \left[\lambda \mathcal{M}_{3}(\mathcal{X}), (1- \lambda) \mathcal{M}_{3}(\mathcal{Y})\right]$.
\item \textbf{repeat} until $t = t_{\max}$ or convergence:
\begin{align*}
    & \boldsymbol{V}^{(t+1)} = \text{Leading $r_x$ left singular vectors of } \mathcal{X} \times_3 \boldsymbol{u}^{(t)}\\
    & \boldsymbol{W}^{(t+1)} = \text{Leading $r_y$ left singular vectors of } \mathcal{Y} \times_3 \boldsymbol{u}^{(t)}\\
    & \boldsymbol{D}^{(t+1)}_x = (\boldsymbol{V}^{(t+1)})^{\prime} (\mathcal{X} \times_{3} \boldsymbol{u}^{(t)}) \boldsymbol{V}^{(t+1)} \\
    & \boldsymbol{D}^{(t+1)}_y = (\boldsymbol{W}^{(t+1)})^{\prime} (\mathcal{Y} \times_{3} \boldsymbol{u}^{(t)}) \boldsymbol{W}^{(t+1)}\\
    & \boldsymbol{u}^{(t+1)} = \dfrac{\lambda[\mathcal{X}; \boldsymbol{V}^{(t+1)}, \boldsymbol{D}^{(t+1)}_x] + (1-\lambda)[\mathcal{Y}; \boldsymbol{W}^{(t+1)}, \boldsymbol{D}^{(t+1)}_y]}{\left\|\lambda[\mathcal{X}; \boldsymbol{V}^{(t+1)}, \boldsymbol{D}^{,(t+1)}_x] + (1-\lambda)[\mathcal{Y}; \boldsymbol{W}^{(t+1)}, \boldsymbol{D}^{(t+1)}_y]\right\|_2}\\
    & t = t+1
\end{align*}
\item \textbf{return} $\hat{\boldsymbol{u}} = \boldsymbol{u}^{(t)}$, $\hat{\boldsymbol{V}} = \boldsymbol{V}^{(t)}$, $\hat{\boldsymbol{W}} = \boldsymbol{W}^{(t)}$; $\hat{\boldsymbol{D}}_{x} = \boldsymbol{D}^{(t)}_x$, $\hat{\boldsymbol{D}}_y = \boldsymbol{D}^{(t)}_y$; $\hat{\mathcal{X}} = \hat{\boldsymbol{V}} \hat{\boldsymbol{D}}_{x} \hat{\boldsymbol{V}}^{\prime} \circ \hat{\boldsymbol{u}}$, $\hat{\mathcal{Y}} = \hat{\boldsymbol{W}} \hat{\boldsymbol{D}}_{y} \hat{\boldsymbol{W}}^{\prime} \circ \hat{\boldsymbol{u}}$.
\end{itemize}
\end{algorithm}
\subsubsection{Comparison Baselines: iHOSVD and iHOOI}
To construct comparison baselines for our methods, we consider two straightforward extensions of HOSVD \citep{de2000multilinear} and HOOI \citep{de2000best}, power method for Tucker model, when we have integrated data. We refer to them as iHOSVD and iHOOI, which simply concatenate the matricized data along the third mode at each iteration. The detailed procedures of iHOSVD and iHOOI are summarized in Algorithms \ref{iHOSVD} and \ref{iHOOI}.

\begin{algorithm}[H]
\caption{Integrated High-Order SVD (iHOSVD)}\label{iHOSVD}
\begin{itemize}
    \item Input: $\mathcal{X} \in \mathbb{R}^{p \times p \times N}$, $\mathcal{Y} \in \mathbb{R}^{q \times q \times N}$, $(r_{x_{1}}, \cdots, r_{x_{K}})$, $(r_{y_{1}}, \cdots, r_{y_{K}})$.
    \item $\hat{\boldsymbol{V}} =$ The leading $\sum\limits_{k = 1}^{K} r^x_k$ singular vectors of $\mathcal{M}_{1}(\mathcal{X})$.
    \item $\hat{\boldsymbol{W}} =$ The leading $\sum\limits_{k = 1}^{K} r^y_k$ singular vectors of $\mathcal{M}_{1}(\mathcal{Y})$.
    \item $\hat{\boldsymbol{U}} =$ The leading $K$ singular vectors of the matrix $\left(\mathcal{M}_{3}(\mathcal{X}), \mathcal{M}_{3}(\mathcal{Y})\right)$.
    \item The estimation of Tucker core tensors are
    \begin{align*}
        & \hat{\mathcal{S}}_{x} = \mathcal{X} \times_{1} \hat{\boldsymbol{V}}^{\prime} \times_{2} \hat{\boldsymbol{V}}^{\prime} \times_{3} \hat{\boldsymbol{U}}^{\prime}\\
        & \hat{\mathcal{S}}_{y} = \mathcal{Y} \times_{1} \hat{\boldsymbol{W}}^{\prime} \times_{2} \hat{\boldsymbol{W}}^{\prime} \times_{3} \hat{\boldsymbol{U}}^{\prime}.
    \end{align*}
    And then the estimation of true parameter tensors are 
    \begin{align*}
        & \hat{\mathcal{X}} = \hat{\mathcal{S}}_{x} \times_{1} \hat{\boldsymbol{V}} \times_{2} \hat{\boldsymbol{V}} \times_{3} \hat{\boldsymbol{U}}\\
        & \hat{\mathcal{Y}} = \hat{\mathcal{S}}_{y} \times_{1} \hat{\boldsymbol{W}} \times_{2} \hat{\boldsymbol{W}} \times_{3} \hat{\boldsymbol{U}}.
    \end{align*}
    \item \textbf{return} $\hat{\boldsymbol{u}}_{i} = \hat{\boldsymbol{U}}_{\cdot, i}$, $\hat{\boldsymbol{V}}_{i} = \hat{\boldsymbol{V}}_{\cdot, 1+\sum_{k = 1}^{i-1} r^x_k: \sum_{k = 1}^{i} r^x_k}$, $\hat{\boldsymbol{W}}_{i} = \hat{\boldsymbol{W}}_{\cdot, 1+\sum_{k = 1}^{i-1} r^y_k: \sum_{k = 1}^{i} r^y_k}$ for $i = 1, \cdots, K$ as the corresponding estimated factors of multi-factor semi-symmetric tensor. And $\hat{\mathcal{X}} $, $\hat{\mathcal{Y}}$.
\end{itemize}
\end{algorithm}

\begin{algorithm}[H]
 \caption{Integrated High-Order Orthogonal Iteration (iHOOI)}\label{iHOOI}
 \begin{itemize}
     \item Input: $\mathcal{X} \in \mathbb{R}^{p \times p \times N}$, $\mathcal{Y} \in \mathbb{R}^{q \times q \times N}$, $(r_{x_{1}}, \cdots, r_{x_{K}})$, $(r_{y_{1}}, \cdots, r_{y_{K}})$, and maximum iteration number $k_{\max}$.
     \item Initialization: Let $k = 0$, and $\boldsymbol{u}^{(0)}, \boldsymbol{V}^{(0)}, \boldsymbol{W}^{(0)}$ are the outputs from iHOSVD (Algorithm \ref{iHOSVD}).
     \item \textbf{repeat} until $k = k_{max}$ or convergence:
     \begin{align*}
         & \boldsymbol{V}^{(k+1)} = \text{The leading $\sum\limits_{k = 1}^{K} r^x_k$ singular vectors of } \mathcal{M}_{1}(\mathcal{X} \times_{2} (\boldsymbol{V}^{(k)})^{\prime} \times_{3} (\boldsymbol{U}^{(k)})^{\prime})\\
         & \boldsymbol{W}^{(k+1)} = \text{The leading $\sum\limits_{k = 1}^{K} r^y_k$ singular vectors of } \mathcal{M}_{1}(\mathcal{Y} \times_{2} (\boldsymbol{W}^{(k)})^{\prime} \times_{3} (\boldsymbol{U}^{(k)})^{\prime})\\
         & \boldsymbol{U}^{(k+1)} = \text{The leading $K$ singular vectors of } \\&\quad\left(\mathcal{M}_{3}(\mathcal{X} \times_{1} (\boldsymbol{V}^{(k+1)})^{\prime} \times_{2} (\boldsymbol{V}^{(k+1)})^{\prime}), \mathcal{M}_{3}(\mathcal{Y} \times_{1} (\boldsymbol{W}^{(k+1)})^{\prime} \times_{2} (\boldsymbol{W}^{(k+1)})^{\prime})\right),
         k = k+1.
     \end{align*}
     \item Denote $\hat{\boldsymbol{U}} = \boldsymbol{U}^{(k)}$, $\hat{\boldsymbol{V}} = \boldsymbol{V}^{(k)}$, $\hat{\boldsymbol{W}} = \boldsymbol{W}^{(k)}$. Then the estimation of Tucker core tensors are
     \begin{align*}
         & \hat{\mathcal{S}}_{x} = \mathcal{X} \times_{1} \hat{\boldsymbol{V}}^{\prime} \times_{2} \hat{\boldsymbol{V}}^{\prime} \times_{3} \hat{\boldsymbol{U}}^{\prime}\\
        & \hat{\mathcal{S}}_{y} = \mathcal{Y} \times_{1} \hat{\boldsymbol{W}}^{\prime} \times_{2} \hat{\boldsymbol{W}}^{\prime} \times_{3} \hat{\boldsymbol{U}}^{\prime}.
     \end{align*}
     And then the estimation of true parameter tensors are 
     \begin{align*}
         & \hat{\mathcal{X}} = \hat{\mathcal{S}}_{x} \times_{1} \hat{\boldsymbol{V}} \times_{2} \hat{\boldsymbol{V}} \times_{3} \hat{\boldsymbol{U}}\\
         & \hat{\mathcal{Y}} = \hat{\mathcal{S}}_{y} \times_{1} \hat{\boldsymbol{W}} \times_{2} \hat{\boldsymbol{W}} \times_{3} \hat{\boldsymbol{U}}.
     \end{align*}
     \item \textbf{return} $\hat{\boldsymbol{u}}_{i} = \hat{\boldsymbol{U}}_{\cdot, i}$, $\hat{\boldsymbol{V}}_{i} = \hat{\boldsymbol{V}}_{\cdot, 1+\sum_{k = 1}^{i-1} r^x_k: \sum_{k = 1}^{i} r^x_k}$, $\hat{\boldsymbol{W}}_{i} = \hat{\boldsymbol{W}}_{\cdot, 1+\sum_{k = 1}^{i-1} r^y_k: \sum_{k = 1}^{i} r^y_k}$ for $i = 1, \cdots, K$ as the
     corresponding estimated factors of multi-factor semi-symmetric tensor. And $\hat{\mathcal{X}}, \hat{\mathcal{Y}}$.
 \end{itemize}
\end{algorithm}

\subsection{Additional Empirical Details and Results}\label{sec:EmpiricalDetails}
We first provide some additional details on the simulation setup of our comparative studies. 
\begin{enumerate}
    \item Setup of the structured simulation presented in the bottom left panel of Figure \ref{fig:comp_main}: $K=2$, $\boldsymbol{r}_x=\boldsymbol{r}_y=(3,2)'$. The ground truth factors are generated as follows.
    \begin{enumerate}
        \item To generate the true population factors $\bu_1^*$ and $\bu_2^*$, we consider a Gaussian mixture model with three components, where the mean of three components $\mu_k$, $k=1,2,3$, are generated from $U[0,1]$. We first randomly assign each sample $i$ into one of the three clusters, say $k$, and then generate $((\bu^*_1)_i,(\bu^*_2)_i)$ from $\mathcal{N}(\mu_k, 0.05)$. We then normalize both $\bu^*_1$ and $\bu^*_2$ to unit vectors.
        \item To generate the true network factors in the first layer, $\bV_1^*\in \bbR^{p\times 3}$ and $\bW_1^*\in \bbR^{q\times 3}$, we consider a three-block graph structure. In particular, each row $i$ of $\bV_1^*$ ($\bW_1^*$) is randomly assigned to one of the three blocks, say $k$, and then we generate $(\bV_1^*)_{i,k}$ from $\mathcal{N}(3,1)$ and let $(\bV_1^*)_{i,\backslash k} = 0$. Both $\bV_1^*$ and $\bW_1^*$ are then orthogonalized individually.
        \item To generate the true network factors in the second layer, $\bV_2^*\in \bbR^{p\times 2}$ and $\bW_2^*\in \bbR^{q\times 2}$, we consider two-star graphs of $p$ and $q$ nodes, respectively. We randomly assign each of the $p$ or $q$ nodes into two components, each being a star graph, and then compute the top two singular vectors of the Laplacian matrices for the two-star graphs. These singular vectors yield are $\bV_2^*$ and $\bW_2^*$, respectively.
        \item The signal strength $\boldsymbol{d}_x^*$, $\boldsymbol{d}_y^*$ are set as described in Section \ref{sec:simulation}, which are the same across structured and unstructured simulations.
    \end{enumerate}
    \item Setup of the structured simulation example in Figure \ref{fig:structuredExample} is slightly different from Figure \ref{fig:comp_main}. In particular, for clearer visualization, we set $p=q=50$ and $N=200$, but the network factors are still generated from three-block graphs, two-star graphs, and the population factors are generated from a Gaussian mixture with three components.
    \item Setup of the generalized models: we consider two settings of the diagonal matrices $\bD_{x,k}^*$, $\bD_{y,k}^*$.
    \begin{enumerate}
        \item For the results presented in Figure \ref{fig:comp_main} and the setting 1 in Figure \ref{fig:generalModel_oracle} and \ref{fig:generalModel_BIC}, we set $\bD_{x,1}^* = \mathrm{SNR} * (\sqrt{p} + \sqrt{N})\mathrm{diag}([2,\,1.5,\,1.2])$, $\bD_{x,2}^* = \mathrm{SNR} * (\sqrt{p} + \sqrt{N})\mathrm{diag}([1,\,0.8])$; $\bD_{y,1}^* = \mathrm{SNR} * (\sqrt{q} + \sqrt{N})\mathrm{diag}([2,\,1.5,\,1.2])$, $\bD_{y,2}^* = \mathrm{SNR} * (\sqrt{q} + \sqrt{N})\mathrm{diag}([1,\,0.8])$. The eigenvalues are moderately different from each other within each factor. 
        \item We also consider a more difficult setting where $\bD_{x,1}^* = \mathrm{SNR} * (\sqrt{p} + \sqrt{N})\mathrm{diag}([3.2,\,2,\,1.2])$, $\bD_{x,2}^* = \mathrm{SNR} * (\sqrt{p} + \sqrt{N})\mathrm{diag}([1,\,0.8])$; $\bD_{y,1}^* = \mathrm{SNR} * (\sqrt{q} + \sqrt{N})\mathrm{diag}([3.2,\,2,\,1.2])$, $\bD_{y,2}^* = \mathrm{SNR} * (\sqrt{q} + \sqrt{N})\mathrm{diag}([1,\,0.8])$.. The results are presented as the setting 2 in Figure \ref{fig:generalModel_oracle} and \ref{fig:generalModel_BIC}.
    \end{enumerate}
    \item Setup of the network simulations: we construct two pairs of stochastic block models (SBMs) for modeling the multi-modal networks, and each sample $i$ falls within each of the two pairs with probability $0.75$ and $0.25$. For the first pair, the two SBMs are of size $p,\,q$ with three blocks of sizes $0.4p$, $0.3p$, $0.3p$ and $0.4q$, $0.4q$, $0.2q$, respectively. For the second pair, the two SBMs have two blocks with sizes $0.5p$, $0.5p$ and $0.6q$, $0.4q$ respectively. The top panel of Figure \ref{fig:networkFactors} shows an example of the block structures of the two pairs of SBMs when $p=q=80$. The within-block connection probabilities range from $0.5$ to $0.8$, and the out-of-block probabilities are set as $0.3$. Formally, the two pairs of SBM probability matrices can be denoted as $(\boldsymbol{P}_{x,1},\, \boldsymbol{P}_{y,1})$, $(\boldsymbol{P}_{x,2},\, \boldsymbol{P}_{y,2})$, where $\boldsymbol{P}_{x,k} = \boldsymbol{\Theta}_{x,k}\boldsymbol{B}_{x,k}\boldsymbol{\Theta}_{x,k}'$, $\boldsymbol{P}_{y,k} = \boldsymbol{\Theta}_{y,k}\boldsymbol{B}_{y,k}\boldsymbol{\Theta}_{y,k}'$. Specifically, $\boldsymbol{\Theta}_{x,1} = \begin{pmatrix}
        1_{0.4p}&0&0\\
        0&1_{0.3p}&0\\
        0&0&1_{0.3p}
    \end{pmatrix}$, $\boldsymbol{B}_{x,1} = \begin{pmatrix}
        0.8&0.3&0.3\\
        0.3&0.8&0.3\\
        0.3&0.3&0.8
    \end{pmatrix}$; $\boldsymbol{\Theta}_{x,2} = \begin{pmatrix}
        1_{0.3p}&0\\
        0&1_{0.2p}\\
        1_{0.2p}&0\\
        0&1_{0.3p}
    \end{pmatrix}$, $\boldsymbol{B}_{x,2} = \begin{pmatrix}
        0.6&0.3\\
        0.3&0.6
    \end{pmatrix}$; $\boldsymbol{\Theta}_{y,1} = \begin{pmatrix}
        1_{0.4q}&0&0\\
        0&1_{0.4q}&0\\
        0&0&1_{0.2q}
    \end{pmatrix}$, $\boldsymbol{B}_{y,1} = \begin{pmatrix}
        0.7&0.3&0.3\\
        0.3&0.7&0.3\\
        0.3&0.3&0.7
    \end{pmatrix}$; $\boldsymbol{\Theta}_{y,2} = \begin{pmatrix}
        1_{0.3p}&0\\
        0&1_{0.4p}\\
        1_{0.3p}&0
    \end{pmatrix}$, $\boldsymbol{B}_{y,2} = \begin{pmatrix}
        0.5&0.3\\
        0.3&0.5
    \end{pmatrix}$. 
\end{enumerate}

Next, we give more implementation details of our methods in the simulation studies. For the hyperparameter selection of our JisstPCA and G-JisstPCA methods, we choose $K$ as the true number of factors, $\lambda = \frac{\|\cX\|_F}{\|\cX\|_F+\|\cY\|_F}$, and we select ranks using the BIC-deflation scheme described in Algorithm \ref{bic_def}, with input rank range from 1 to 5. For the comparative studies with JisstPCA or G-JisstPCA models, subtraction deflation is applied for all the results presented in the main paper, while the results of projection deflation are also included in the Figures \ref{fig:Main_oracle} and \ref{fig:Main_BIC}. Partial projection deflation (projection on the population mode) and subtraction deflation are used for the network simulation. For the clustering experiments in network simulations, we apply the k-means function in Matlab, which uses the squared Euclidean distance metric and the k-means++ algorithm for cluster center initialization. 

Finally, we present and discuss some additional empirical results briefly mentioned in the main paper.
\begin{enumerate}
    \item Figures \ref{fig:Main_oracle} and \ref{fig:Main_BIC} present more detailed results for the settings in Figure \ref{fig:comp_main}. In particular, both subtraction deflation and projection deflation versions of (Generalized) JisstPCA are presented; Subtraction deflation turns out to be the best across different settings. When using BIC tuned ranks (Figure \ref{fig:Main_BIC}), (Generalized) JisstPCA sometimes give larger errors with low SNR. When using oracle ranks as iHOOI and iHOSVD, our methods are always the best.
    \item Figure \ref{fig:non-orthogonal_oracle} and \ref{fig:non-orthogonal_BIC} focus on the non-orthogonal, unstructured factor setting. Baseline methods include iHOSVD and iHOOI with oracle ranks; our methods include JisstPCA and Generalized JisstPCA with oracle ranks and BIC-selected ranks, when the deflation strategy is subtraction or projection.
    The results suggest that
    \begin{itemize}
        \item When the factors are non-orthogonal, our JisstPCA and G-JisstPCA algorithms with subtraction deflation works the best. iHOOI, iHOSVD and JisstPCA with orthogonal deflation makes wrong assumptions so they all have larger errors, but for some reasons, JisstPCA is still slightly better.
        \item When SNR is small, sometimes the BIC selected ranks are inaccurate, leading to $\sin\Theta$ distance 1. But when all methods use oracle ranks, JisstPCA and G-JisstPCA with subtraction deflation work the best.
    \end{itemize}
    \item Figure \ref{fig:orthogonal2_oracle} and \ref{fig:orthogonal2_BIC} focus on a setting where factors are mutually orthogonal, unstructured, and has sufficient singular gap. In particular, we set $\boldsymbol{d}_x^* = \mathrm{SNR} * (\sqrt{p}+\sqrt{N})(1,\, 0.5)'$, $\boldsymbol{d}_y^* = \mathrm{SNR} * (\sqrt{q}+\sqrt{N})(1,\, 0.5)'$ as in the non-orthogonal simulations. The ground truth tensors are Tucker low-rank tensors and iHOOI and iHOSVD are designed to perform well in this setting. We can see that our methods behave comparably to iHOOI and iHOSVD. Also note that our methods with subtraction deflation has the same performance as projection deflation, suggesting it to be a safe scheme to use when it is unclear if the factors should be orthogonal or not. 
    \item Effects of dimensionality: Figure \ref{fig:non-orthogonal_oracle}-\ref{fig:orthogonal2_BIC} study the effect of dimenisonalities on the estimation accuracy of different factors. Four cases of different dimensions are considered: $p=50,\,q=50,\,N=200$ (case 1); $p=150,\,q=150,\,N=50$ (case 2); $p=150,\,q=50,\,N=200$ (case 3); $p=150,\,q=50,\,N=50$ (case 4). The results suggest that 
    \begin{itemize}
        \item when $N$ is larger, $u$ and $V$ tend to be estimated better;
        \item when $p>q$, $V$ tends to be estimated better. 
    \end{itemize}
    \item Figure \ref{fig:generalModel_oracle} and \ref{fig:generalModel_BIC} consider two different generalized model settings with (i) different but similar eigenvalues within each layer, and (ii) more different eigenvalues within each layer, as described earlier in the simulation setup. The factors are not orthogonal across factors. We have the following observations:
    \begin{itemize}
        \item When eigenvalues within each layer are not that different, both JisstPCA and G-JisstPCA with subtraction deflation work very well.
        \item When eigenvalues within each layer are very different, JisstPCA with oracle ranks and subtraction deflation fails miserably on the second factor. 
        \item Interestingly, JisstPCA with projection deflation performs better in this case, probably since it left less residuals from the first factor. G-JisstPCA with subtraction deflation works the best since it assumes the correct model for the generated data.
    \end{itemize} 
    \item Network visualization: We visualize the true factors and the estimated factors using different methods in Figure \ref{fig:networkFactors}. The ground truth population factor is the normalized cluster membership vector. To obtain the ground truth network factors, we take the edge probability matrix of each network component, project both its rows and columns onto the orthogonal complement of the all-one's vector, and then take the top singular vectors. 
\end{enumerate}

\begin{figure}[!t]
    \centering
    \includegraphics[width = \textwidth]{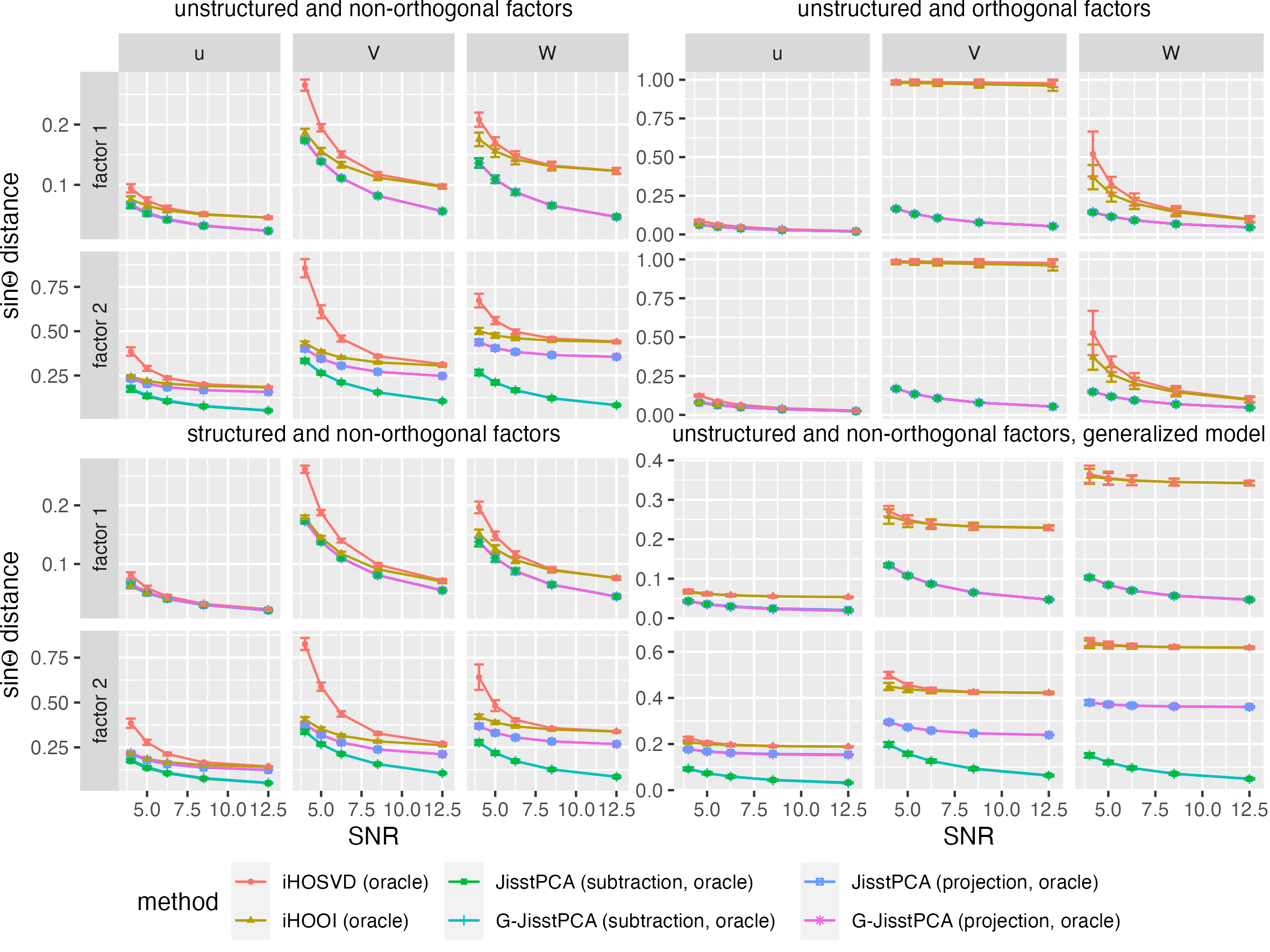}
    \caption{Detailed results of Figure \ref{fig:comp_main}: Estimation errors ($\sin\Theta$ distances in spectral norm) of all factors using iHOSVD, iHOOI, JisstPCA, and Generalized JisstPCA with subtraction and projected deflation, where all methods use the {\bf true ranks}. This figure presents the same four scenarios as those included in the main paper. The mean errors of $10$ replicates are plotted, where the error bars represent 95\% confidence intervals.}
    \label{fig:Main_oracle}
\end{figure}
\begin{figure}[!t]
    \centering
    \includegraphics[width = \textwidth]{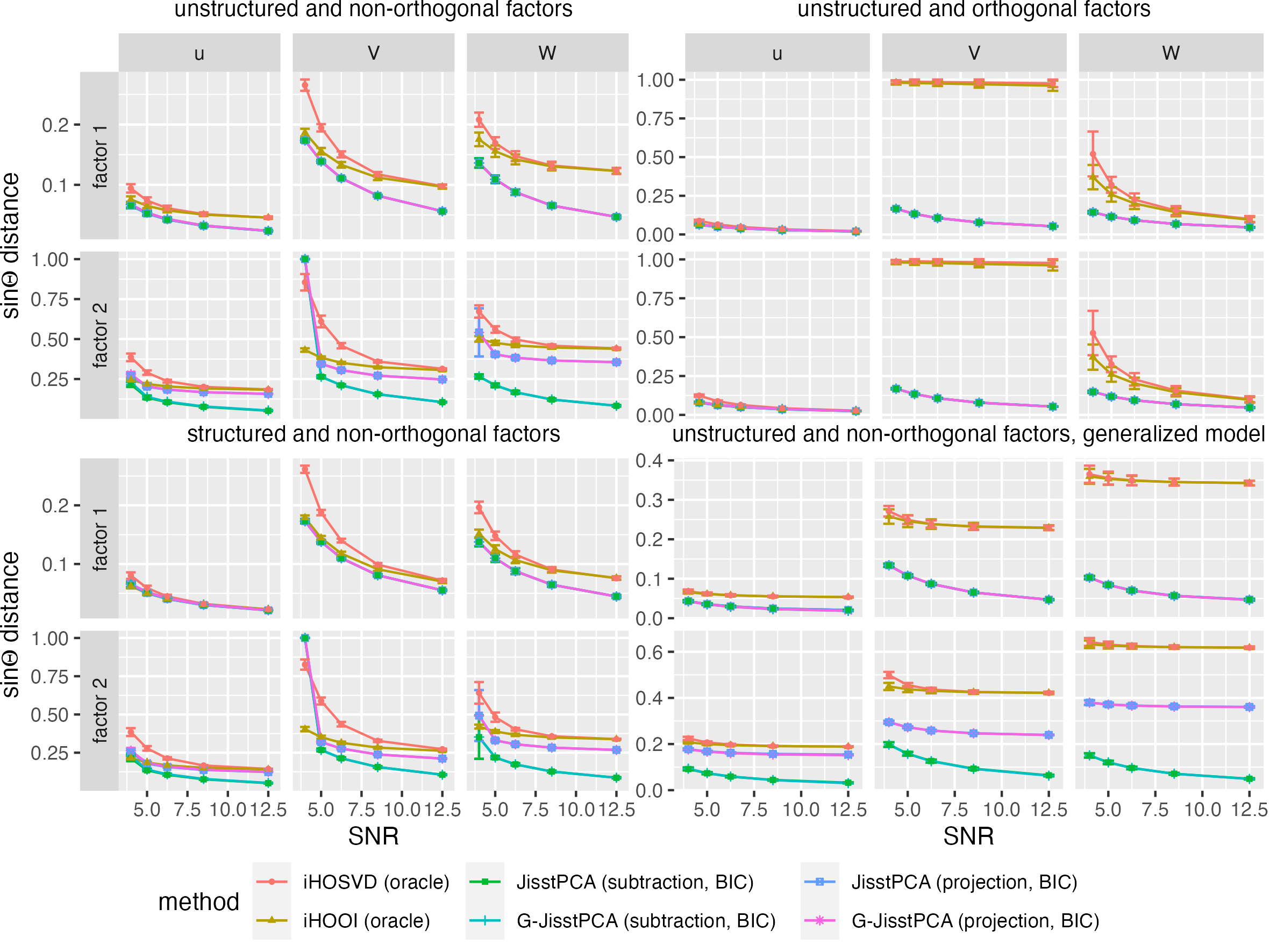}
    \caption{Detailed results of Figure \ref{fig:comp_main}: Estimation errors ($\sin\Theta$ distances in spectral norm) of all factors using iHOSVD, iHOOI, JisstPCA, and Generalized JisstPCA with {\bf both subtraction and projected deflation}, where our methods (JisstPCA and G-JisstPCA) use the BIC selected ranks. The underlying factors are unstructured and non-orthogonal. This figure presents the same four scenarios as those included in the main paper. The mean errors of $10$ replicates are plotted, where the error bars represent 95\% confidence intervals.}
    \label{fig:Main_BIC}
\end{figure}
\begin{figure}[!t]
    \centering
    \includegraphics[width = \textwidth]{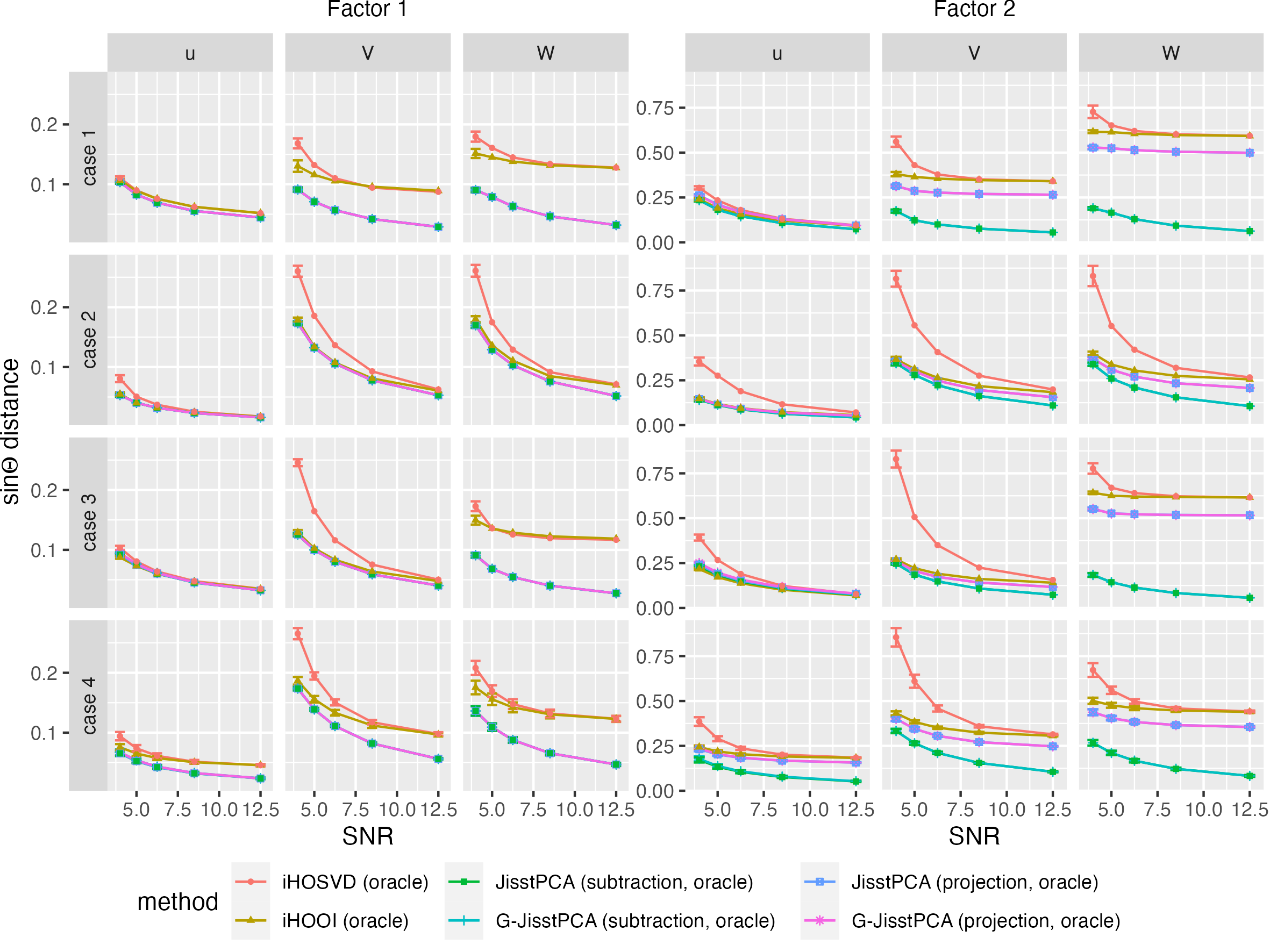}
    \caption{Dimensionality study: Four cases of different dimensions are considered: $p=50,\,q=50,\,N=200$ (case 1); $p=150,\,q=150,\,N=50$ (case 2); $p=150,\,q=50,\,N=200$ (case 3); $p=150,\,q=50,\,N=50$ (case 4). All methods use the true ranks. The underlying factors are unstructured and non-orthogonal. The mean errors of $10$ replicates are plotted, where the error bars represent 95\% confidence intervals.}
    \label{fig:non-orthogonal_oracle}
\end{figure}

\begin{figure}[!t]
    \centering
    \includegraphics[width = \textwidth]{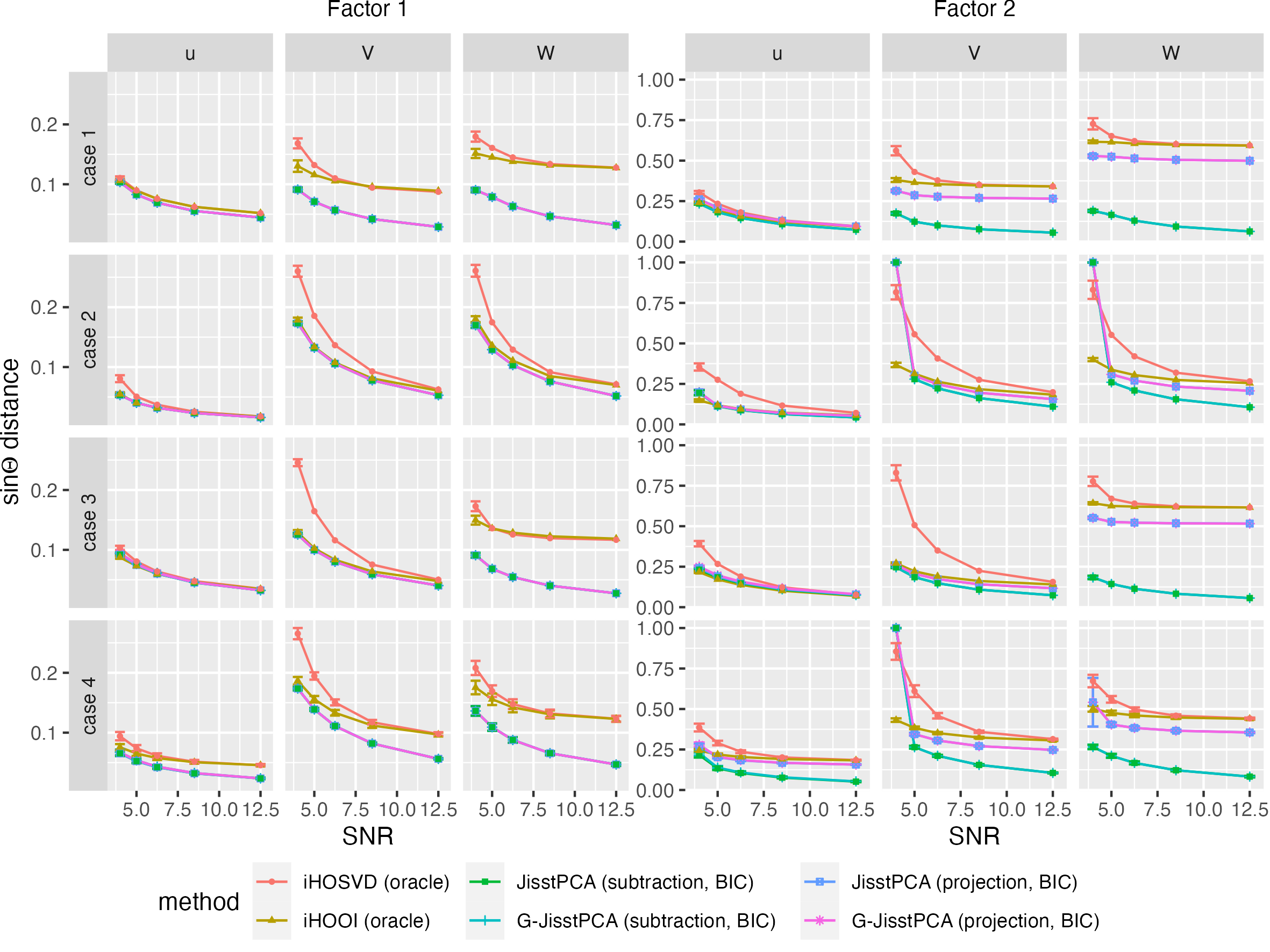}
    \caption{Dimensionality study: Four cases of different dimensions are considered: $p=50,\,q=50,\,N=200$ (case 1); $p=150,\,q=150,\,N=50$ (case 2); $p=150,\,q=50,\,N=200$ (case 3); $p=150,\,q=50,\,N=50$ (case 4). Our methods (JisstPCA and G-JisstPCA) use the BIC selected ranks. The underlying factors are unstructured and non-orthogonal. The mean errors of $10$ replicates are plotted, where the error bars represent 95\% confidence intervals.}
    \label{fig:non-orthogonal_BIC}
\end{figure}

\begin{figure}[!t]
    \centering
    \includegraphics[width = \textwidth]{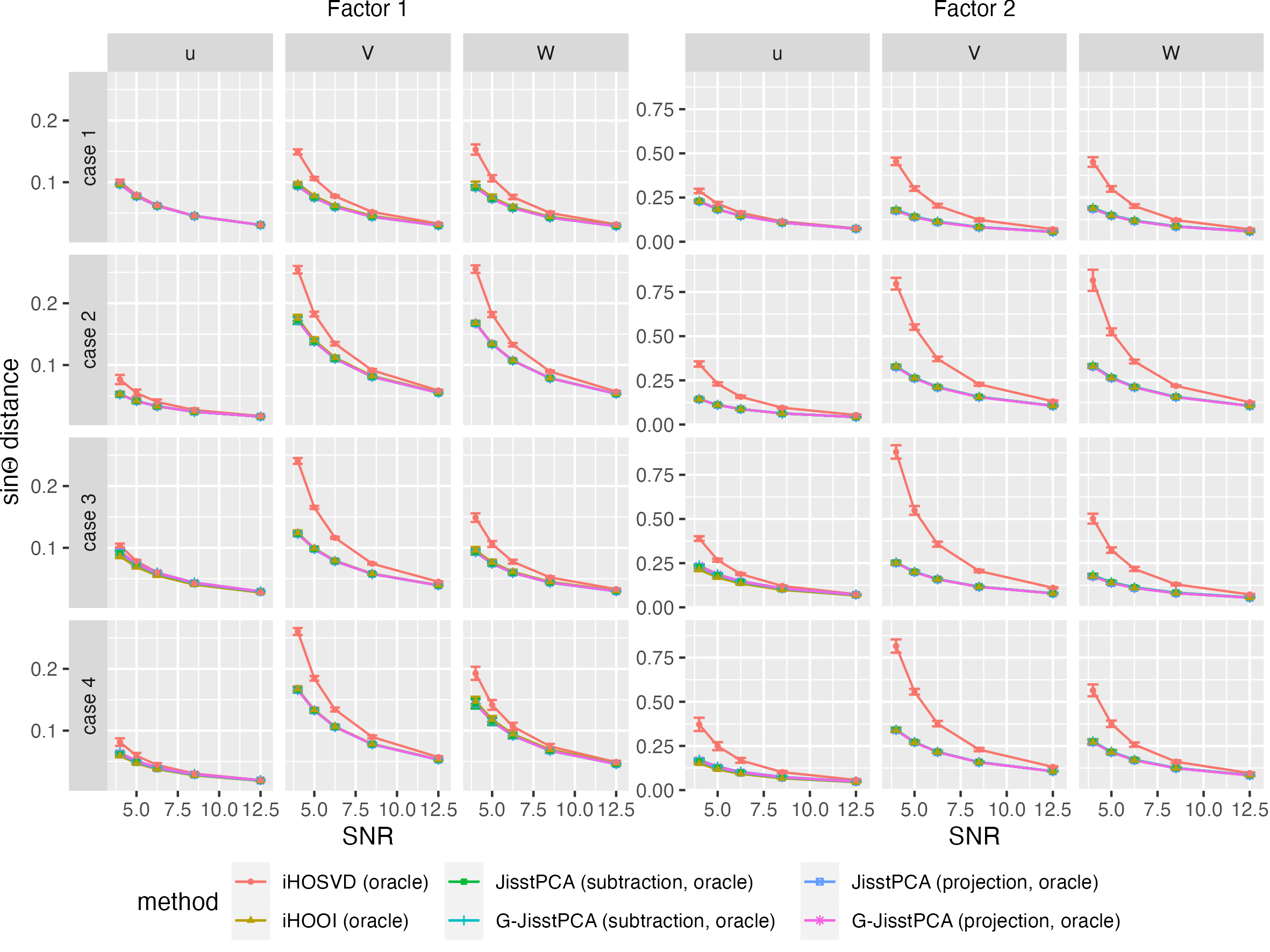}
    \caption{Orthogonal factors with significant singular gaps: we have comparative performance with iHOOI. All methods use the true ranks. Four cases of different dimensions are considered as in Figures \ref{fig:non-orthogonal_oracle}-\ref{fig:non-orthogonal_BIC} are considered. The mean errors of $10$ replicates are plotted, where the error bars represent 95\% confidence intervals.}
    \label{fig:orthogonal2_oracle}
\end{figure}

\begin{figure}[!t]
    \centering
    \includegraphics[width = \textwidth]{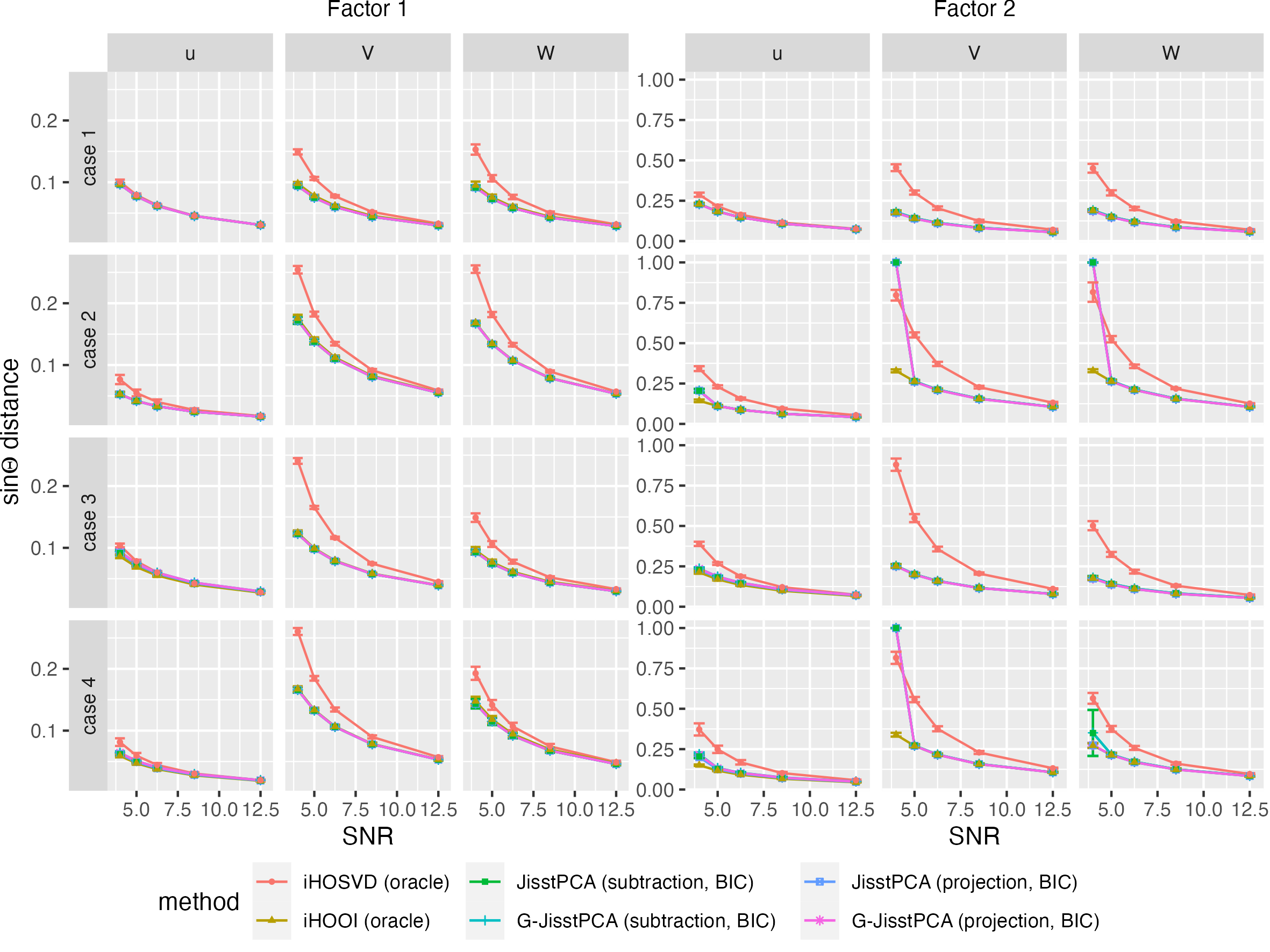}
    \caption{Orthogonal factors with significant singular gaps. Our methods (JisstPCA and G-JisstPCA) use the BIC selected ranks. We have comparative performance with iHOOI except for some small SNR scenario where ranks are selected wrong. Four cases of different dimensions are considered as in Figures \ref{fig:non-orthogonal_oracle}-\ref{fig:non-orthogonal_BIC} are considered. The mean errors of $10$ replicates are plotted, where the error bars represent 95\% confidence intervals.}
    \label{fig:orthogonal2_BIC}
\end{figure}

\begin{figure}[!t]
    \centering
    \includegraphics[width = \textwidth]{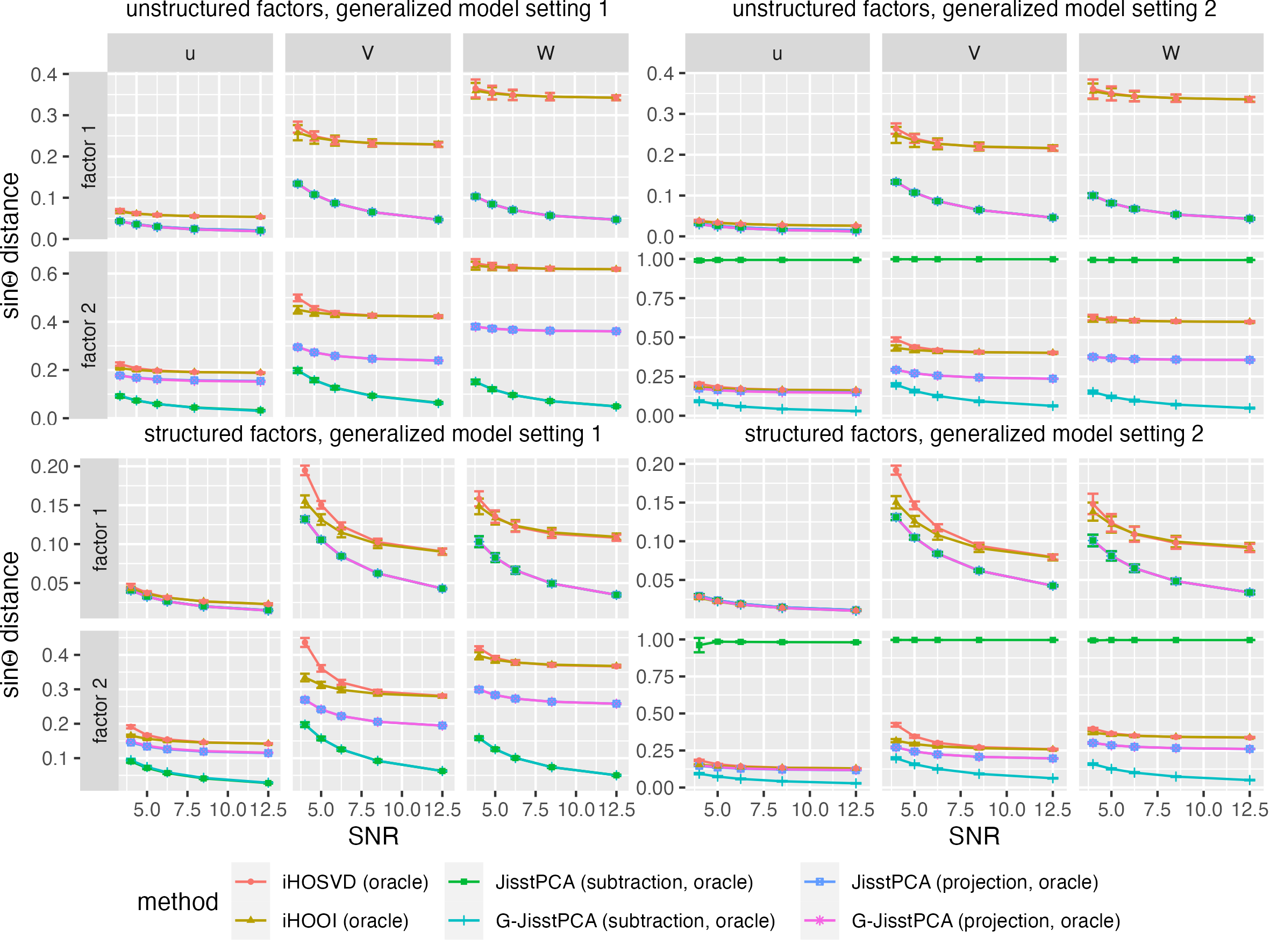}
    \caption{General models with different eigenvalues within each factor: setting 1 has more similar eigenvalues in the first factor while setting 2 is more different. All methods use the true ranks. The mean errors of $10$ replicates are plotted, where the error bars represent 95\% confidence intervals.}
    \label{fig:generalModel_oracle}
\end{figure}
\begin{figure}[!t]
    \centering
    \includegraphics[width = \textwidth]{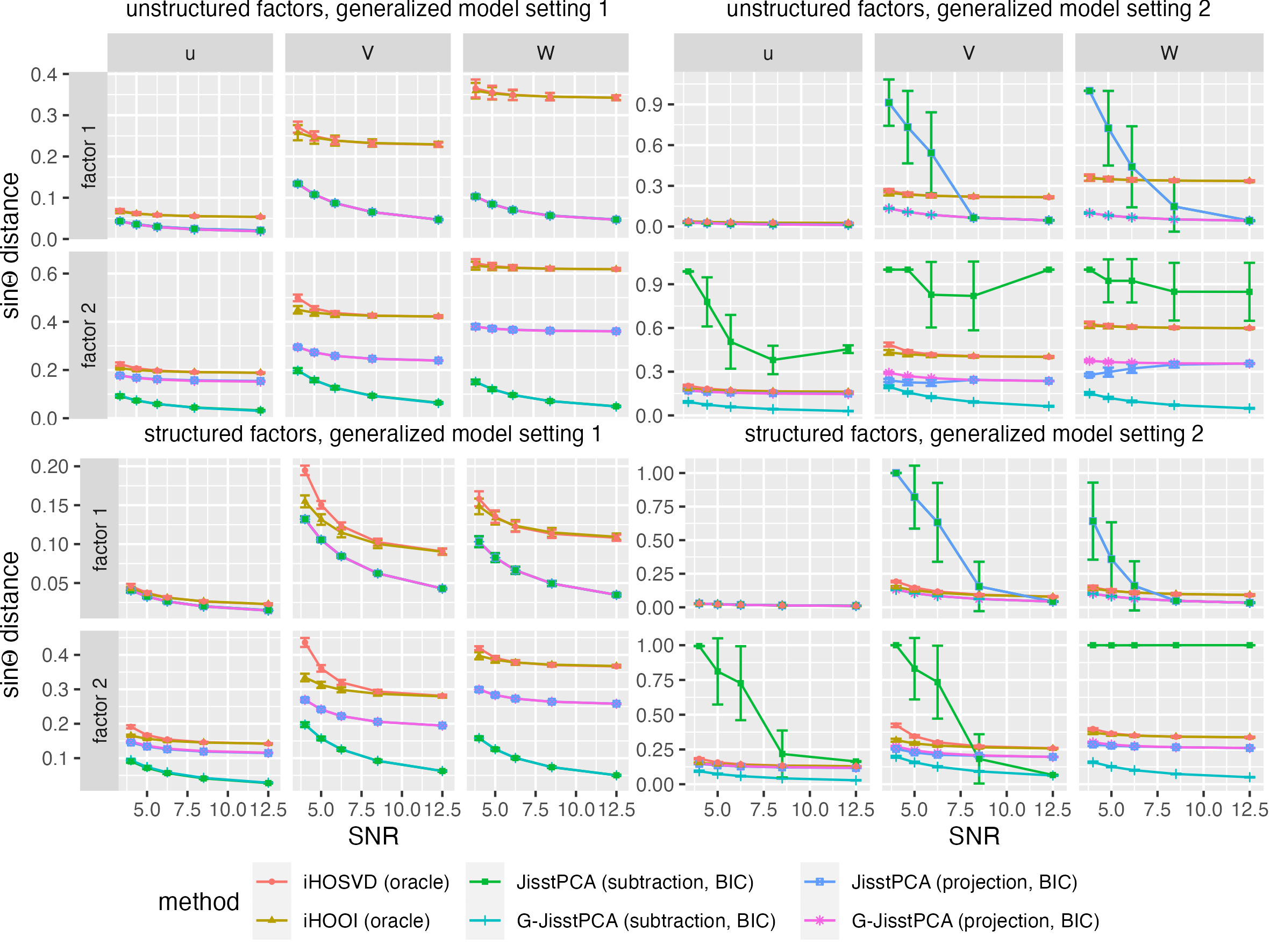}
    \caption{General models with different eigenvalues within each factor: setting 1 has more similar eigenvalues in the first factor while setting 2 is more different. Our methods (JisstPCA and G-JisstPCA) use the BIC selected ranks. The mean errors of $10$ replicates are plotted, where the error bars represent 95\% confidence intervals.}
    \label{fig:generalModel_BIC}
\end{figure}
\begin{figure}
    \centering
    \includegraphics[width = \textwidth]{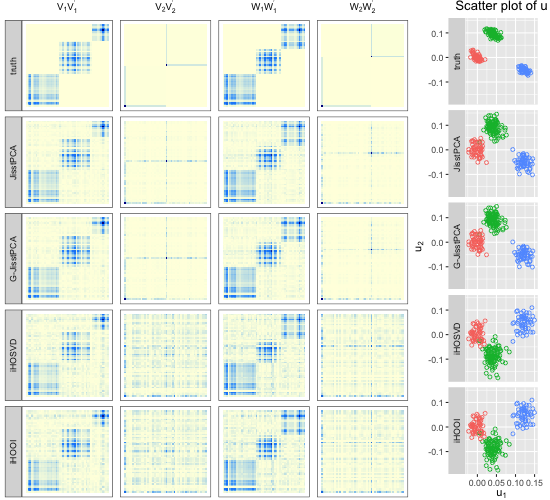}
    \caption{Heatmaps and scatterplots for structured network and population tensor factors reconstruction, by JisstPCA, G-JisstPCA, iHOSVD, iHOOI, together with the true factors.}
    \label{fig:structuredExampleFull}
\end{figure}

\begin{table}[!htb]
\centering
\caption{Population clustering and network community detection for multi-modal populations of networks, based on k-means on the estimated factors from JisstPCA, G-JisstPCA, iHOSVD, and iHOOI. The presented Adjusted Rand Index (ARI) values demonstrate the accuracy of sample clustering and node clustering of two network factors for each modality, when network sizes $p=80$, $q=50$, $N=20$. The $\sin\Theta$ estimation errors of each factor is also presented. The average ARI and $\sin\Theta$ distances of 20 independent repeats are presented, with standard deviation inside the parenthesis. The largest average ARI and lowest estimation error for each setting are marked in bold. Both partial projection deflation and subtraction deflation are considered for JisstPCA and G-JisstPCA.}
\label{tab:network_cluster1Detail1}
\scalebox{0.7}{
\begin{tabular}{c|c|c|c|c|c|c}
\hline
\multirow{2}{*}{Clustering}&{\bf JisstPCA} &{\bf G-JisstPCA}&{\bf JisstPCA} &{\bf G-JisstPCA}&\multirow{2}{*}{\bf iHOSVD}&\multirow{2}{*}{\bf iHOOI}\\
\multirow{2}{*}{ARI}&{\bf partial projection} &{\bf partial projection}&{\bf subtraction}&{\bf subtraction}&&\\
&{\bf (BIC)} &{\bf (BIC)}&{\bf (BIC)} &{\bf (BIC)}&{\bf (oracle)}&{\bf (oracle)}\\
\hline
Sample&0.947(0.238) & {\bf 1}(0) & 0.801(0.41) & 0.842(0.386) & {\bf 1}(0) & 0.954(0.205)\\
Network 1 of $\cX$&0.971(0.108) & {\bf 1}(0) & 0.942(0.162) & {\bf 1}(0) & 0.912(0.195) & 0.739(0.251)\\
Network 2 of $\cX$&0.995(0.022) & {\bf 0.997}(0.011) & 0.699(0.471) & 0.749(0.446) & 0.146(0.057) & 0.139(0.043)\\
Network 1 of $\cY$&0.974(0.114) & {\bf 1}(0) & {\bf 1}(0) & 0.974(0.115) & 0.99(0.027) & 0.94(0.155)\\
Network 2 of $\cY$&{\bf 0.87}(0.31) & 0.867(0.309) & 0.685(0.468) & 0.732(0.442) & 0.014(0.046) & 0.121(0.16)\\
\hline
\hline
$\sin\theta(\hat{\bu}_1,\bu_1^*)$&0.087(0.03) & 0.089(0.031) & 0.087(0.03) & 0.089(0.031) & 0.138(0.051) & 0.08(0.026)\\
$\sin\theta(\hat{\bu}_2,\bu_2^*)$&0.167(0.077) & 0.163(0.066) & 0.431(0.349) & 0.38(0.343) & 0.214(0.04) & 0.188(0.162)\\
$\|\sin\Theta(\hat{\bV}_1,\bV_1^*)\|_{\op}$&0.084(0.007) & 0.084(0.007) & 0.084(0.007) & 0.084(0.007) & 0.163(0.06) & 0.163(0.059)\\
$\|\sin\Theta(\hat{\bV}_2,\bV_2^*)\|_{\op}$&0.21(0.074) & 0.211(0.072) & 0.427(0.384) & 0.428(0.383) & 0.817(0.065) & 0.799(0.07)\\
$\|\sin\Theta(\hat{\bW}_1,\bW_1^*)\|_{\op}$&0.158(0.018) & 0.158(0.018) & 0.158(0.018) & 0.158(0.018) & 0.253(0.056) & 0.196(0.058)\\
$\|\sin\Theta(\hat{\bW}_2,\bW_2^*)\|_{\op}$&0.416(0.204) & 0.417(0.203) & 0.536(0.311) & 0.504(0.293) & 0.969(0.044) & 0.903(0.098)\\
\hline
\end{tabular}}
\end{table}

\begin{table}[!htb]
\centering
\caption{Clustering ARI and factor estimation error for multi-modal populations of networks. The set-up is the same as in Table \ref{tab:network_cluster1Detail1}, but with $N=40$. }
\label{tab:network_cluster1Detail2}
\scalebox{0.7}{
\begin{tabular}{c|c|c|c|c|c|c}
\hline
\multirow{2}{*}{Clustering}&{\bf JisstPCA} &{\bf G-JisstPCA}&{\bf JisstPCA} &{\bf G-JisstPCA}&\multirow{2}{*}{\bf iHOSVD}&\multirow{2}{*}{\bf iHOOI}\\
\multirow{2}{*}{ARI}&{\bf partial projection} &{\bf partial projection}&{\bf subtraction}&{\bf subtraction}&&\\
&{\bf (BIC)} &{\bf (BIC)}&{\bf (BIC)} &{\bf (BIC)}&{\bf (oracle)}&{\bf (oracle)}\\
\hline
Sample&1(0) &{\bf 1}(0) & 0.863(0.334) & 0.851(0.363) & {\bf 1}(0) & {\bf 1}(0)\\
Network 1 of $\cX$&{\bf 1}(0) & {\bf 1}(0) & {\bf 1}(0) & {\bf 1}(0) & 0.805(0.241) & 0.663(0.224)\\
Network 2 of $\cX$&{\bf 1}(0) & {\bf 1}(0) & 0.748(0.448) & 0.85(0.365) & 0.156(0.015) & 0.153(0)\\
Network 1 of $\cY$&{\bf 1}(0) & {\bf 1}(0) & {\bf 1}(0) & {\bf 1}(0) & 0.973(0.108) & {\bf 1}(0)\\
Network 2 of $\cY$&{\bf 1}(0) & {\bf 1}(0) & 0.747(0.45) & 0.847(0.373) & 0.116(0.101) & 0.246(0.153)\\
\hline
\hline
$\sin\theta(\hat{\bu}_1,\bu_1^*)$&0.092(0.016) & 0.093(0.017) & 0.092(0.016) & 0.093(0.017) & 0.142(0.028) & {\bf 0.083}(0.016)\\
$\sin\theta(\hat{\bu}_2,\bu_2^*)$&0.152(0.014) & 0.152(0.014) & 0.387(0.324) & 0.3(0.276) & 0.197(0.017) & {\bf 0.145}(0.015)\\
$\|\sin\Theta(\hat{\bV}_1,\bV_1^*)\|_{\op}$&{\bf 0.062}(0.006) & {\bf 0.062}(0.007) & {\bf 0.062}(0.006) & {\bf 0.062}(0.007) & 0.154(0.046) & 0.156(0.047)\\
$\|\sin\Theta(\hat{\bV}_2,\bV_2^*)\|_{\op}$&{\bf 0.154}(0.024) & 0.155(0.023) & 0.36(0.379) & 0.319(0.349) & 0.776(0.022) & 0.768(0.028)\\
$\|\sin\Theta(\hat{\bW}_1,\bW_1^*)\|_{\op}$&{\bf 0.118}(0.014) & {\bf 0.118}(0.014) & {\bf 0.118}(0.014) & {\bf 0.118}(0.014) & 0.202(0.026) & 0.173(0.04)\\
$\|\sin\Theta(\hat{\bW}_2,\bW_2^*)\|_{\op}$&{\bf 0.272}(0.045) & 0.273(0.045) & 0.448(0.327) & 0.375(0.271) & 0.901(0.062) & 0.771(0.079)\\
\hline
\end{tabular}}
\end{table}

\begin{table}[!htb]
\centering
\caption{Clustering ARI and factor estimation error for multi-modal populations of networks. The set-up is the same as in Table \ref{tab:network_cluster1Detail1}, but with $q=80,\,N=20$.}
\label{tab:network_cluster2Detail1}
\scalebox{0.7}{
\begin{tabular}{c|c|c|c|c|c|c}
\hline
\multirow{2}{*}{Clustering}&{\bf JisstPCA} &{\bf G-JisstPCA}&{\bf JisstPCA} &{\bf G-JisstPCA}&\multirow{2}{*}{\bf iHOSVD}&\multirow{2}{*}{\bf iHOOI}\\
\multirow{2}{*}{ARI}&{\bf partial projection} &{\bf partial projection}&{\bf subtraction}&{\bf subtraction}&&\\
&{\bf (BIC)} &{\bf (BIC)}&{\bf (BIC)} &{\bf (BIC)}&{\bf (oracle)}&{\bf (oracle)}\\
\hline
Sample&{\bf 1}(0) & {\bf 1}(0) & 0.809(0.35) & 0.905(0.294) & {\bf 1}(0) & {\bf 1}(0)\\
Network 1 of $\cX$&0.992(0.026) & {\bf 1}(0) & 0.994(0.019) & {\bf 1}(0) & 0.792(0.248) & 0.686(0.237)\\
Network 2 of $\cX${\bf 0.995}(0.022) & {\bf 0.995}(0.022) & 0.497(0.516) & 0.698(0.473) & 0.142(0.042) & 0.16(0.032) \\
Network 1 of $\cY$&{\bf 0.977}(0.105) & 0.975(0.113) & {\bf 1}(0) & {\bf 1}(0) & 0.946(0.143) & {\bf 1}(0) \\
Network 2 of $\cY$&{\bf 0.948}(0.135) & 0.929(0.227) & 0.499(0.509) & 0.699(0.467) & 0.135(0.093) & 0.298(0.15)\\
\hline
\hline
$\sin\theta(\hat{\bu}_1,\bu_1^*)$&0.081(0.022) & 0.08(0.023) & 0.081(0.022) & 0.08(0.023) & 0.13(0.043) & {\bf 0.072}(0.022)\\
$\sin\theta(\hat{\bu}_2,\bu_2^*)$&0.148(0.058) & 0.145(0.058) & 0.55(0.389) & 0.398(0.365) & 0.206(0.048) & {\bf 0.139}(0.055)\\
$\|\sin\Theta(\hat{\bV}_1,\bV_1^*)\|_{\op}$&{\bf 0.082}(0.011) & {\bf 0.082}(0.011) & {\bf 0.082}(0.011) & {\bf 0.082}(0.011) & 0.158(0.061) & 0.159(0.062)\\
$\|\sin\Theta(\hat{\bV}_2,\bV_2^*)\|_{\op}$&{\bf 0.207}(0.069) & {\bf 0.207}(0.069) & 0.582(0.428) & 0.423(0.387) & 0.815(0.059) & 0.777(0.03)\\
$\|\sin\Theta(\hat{\bW}_1,\bW_1^*)\|_{\op}$&0.135(0.015) & {\bf 0.134}(0.015) & 0.135(0.015) & {\bf 0.134}(0.015) & 0.203(0.055) & 0.192(0.059)\\
$\|\sin\Theta(\hat{\bW}_2,\bW_2^*)\|_{\op}$&{\bf 0.323}(0.115) & 0.339(0.175) & 0.624(0.377) & 0.488(0.342) & 0.87(0.067) & 0.782(0.075)\\
\hline
\end{tabular}}
\end{table}

\begin{table}[!htb]
\centering
\caption{Clustering ARI and factor estimation error for multi-modal populations of networks. The set-up is the same as in Table \ref{tab:network_cluster1Detail1}, but with $q=80,\,N=40$.}
\label{tab:network_cluster2Detail2}
\scalebox{0.7}{
\begin{tabular}{c|c|c|c|c|c|c}
\hline
\multirow{2}{*}{Clustering}&{\bf JisstPCA} &{\bf G-JisstPCA}&{\bf JisstPCA} &{\bf G-JisstPCA}&\multirow{2}{*}{\bf iHOSVD}&\multirow{2}{*}{\bf iHOOI}\\
\multirow{2}{*}{ARI}&{\bf partial projection} &{\bf partial projection}&{\bf subtraction}&{\bf subtraction}&&\\
&{\bf (BIC)} &{\bf (BIC)}&{\bf (BIC)} &{\bf (BIC)}&{\bf (oracle)}&{\bf (oracle)}\\
\hline
Sample & {\bf 1}(0) & {\bf 1}(0) & {\bf 1}(0) & 0.949(0.229) & {\bf 1}(0) & {\bf 1}(0)\\
Network 1 of $\cX$ & {\bf 1}(0) & {\bf 1}(0) & {\bf 1}(0) & {\bf 1}(0) & 0.714(0.262) & 0.689(0.236)\\
Network 2 of $\cX$ & {\bf 1}(0) & {\bf 1}(0) & 0.647(0.494) & 0.848(0.371) & 0.156(0.015) & 0.153(0)\\
Network 1 of $\cY$ & {\bf 1}(0) & {\bf 1}(0) & {\bf 1}(0) & {\bf 1}(0) & 0.975(0.11) & 0.977(0.102)\\
Network 2 of $\cY$ & {\bf 0.997}(0.011) & {\bf 0.997}(0.011) & 0.656(0.48) & 0.849(0.37) & 0.281(0.109) & 0.338(0.079)\\
\hline
\hline
$\sin\theta(\hat{\bu}_1,\bu_1^*)$&0.084(0.016) & 0.084(0.016) & 0.084(0.016) & 0.084(0.016) & 0.131(0.027) & {\bf 0.077}(0.015)\\
$\sin\theta(\hat{\bu}_2,\bu_2^*)$&0.14(0.015) & 0.136(0.015) & 0.457(0.379) & 0.291(0.285) & 0.184(0.013) & {\bf 0.131}(0.018)\\
$\|\sin\Theta(\hat{\bV}_1,\bV_1^*)\|_{\op}$&{\bf 0.061}(0.006) & {\bf 0.061}(0.006) & {\bf 0.061}(0.006) & {\bf 0.061}(0.006) & 0.152(0.044) & 0.153(0.045)\\
$\|\sin\Theta(\hat{\bV}_2,\bV_2^*)\|_{\op}$&{\bf 0.146}(0.014) & {\bf 0.146}(0.014) & 0.441(0.42) & 0.399(0.403) & 0.774(0.019) & 0.764(0.029)\\
$\|\sin\Theta(\hat{\bW}_1,\bW_1^*)\|_{\op}$&{\bf 0.093}(0.006) & {\bf 0.093}(0.006) & {\bf 0.093}(0.006) & {\bf 0.093}(0.006) & 0.161(0.029) & 0.154(0.033)\\
$\|\sin\Theta(\hat{\bW}_2,\bW_2^*)\|_{\op}$&0.{\bf 218}(0.028) &{\bf  0.218}(0.028) & 0.468(0.365) & 0.33(0.289) & 0.787(0.054) & 0.718(0.018)\\
\hline
\end{tabular}}
\end{table}

\begin{figure}[!tb]
    \centering
    \includegraphics[width = 15cm]{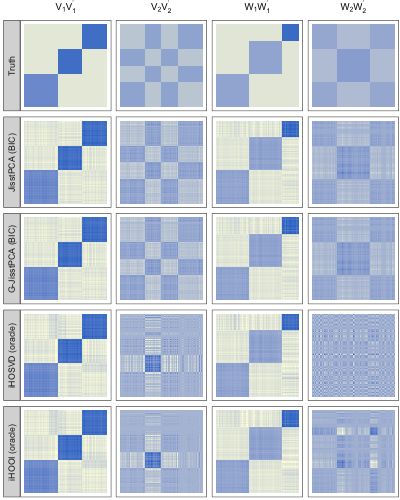}
    \caption{Heatmaps for the network factors reconstruction from network data, by JisstPCA, G-JisstPCA, iHOSVD, and iHOOI, together with the true factors.}
    \label{fig:networkFactors}
\end{figure}

\section{Proof of Main Results}
\label{app2}
\subsection{Additional Notations}
Throughout this paper, we use the following notations repeatedly. The calligraphic letters represent tensors (e.g. $\mathcal{X}$), and the boldface uppercase letters represent matrices (e.g. $\boldsymbol{V}$). Also, we use boldface lowercase letters to denote vectors (e.g. $\boldsymbol{u}$) and lowercase letters to denote the real numbers (e.g. $a$). For $p, q \in \mathbb{R}$, let $p \vee q = \max\{p, q\}$ and $p \wedge q = \min\{p, q\}$. Suppose $\{a_n: n \in \mathbb{N}\}$ and $\{b_n: n \in \mathbb{N}\}$ are two sequences of real numbers, we claim that $a_n = O(b_n)$ if and only if there exists some $N \in \mathbb{N}$ and some positive constant $C > 0$, such that $a_n \leq C b_n$ for all $n \geq N$. And we claim $a_n = o(b_n)$ if for $\forall \epsilon > 0, \exists N \in \mathbb{N}$ such that $a_n \leq \epsilon b_n$ for all $n \geq N$. If there are two positive constants $c>0, C>0$ such that $ca_n \leq b_n \leq Ca_n$ for $\forall n$, we say $a_n \asymp b_n$.

We let $\mathbb{S}^{N-1}$ be the Euclidean unit sphere in $\mathbb{R}^{N}$, i.e. $\mathbb{S}^{N-1} = \{\boldsymbol{u} \in \mathbb{R}^{N}: \|\boldsymbol{u}\|_{2}=1\}$, and $\mathcal{O}_{p, r}$ is the set of $p \times r$ orthogonal matrices, i.e. $\mathcal{O}_{p, r} = \{\boldsymbol{V} \in \mathbb{R}^{p \times r}: \boldsymbol{V}^{\prime}\boldsymbol{V} = \boldsymbol{I}_{r}\}$.
\subsection{Proof of Main Theoretical Results for the JisstPCA}
\begin{proof}[Proof of Theorem \ref{thm:stat_converg_main}]
We use an induction proof for Theorem \ref{thm:stat_converg_main}. Specifically, we will show the following claim holds true, as long as Assumptions \ref{assump:SNR1} and \ref{assump:init} hold.
\begin{claim}\label{claim:induction1}
    For $l\geq 0$, 
    \begin{equation}\label{eq:u_err_induction1}
        |\sin\theta(\bu^*,\bu^{(l)})|^2\leq 1-8\left(\frac{\|\cE_x\|_{\op}^2}{d_x^{*2}}\vee\frac{\|\cE_y\|_{\op}^2}{d_y^{*2}}\right).
    \end{equation}
    For $l\geq 1$, 
    \begin{equation}\label{eq:u_err_induction2}
        |\sin\theta(\bu^*,\bu^{(l)})|\leq \frac{4\|\lambda\cE_x; (1-\lambda)\cE_y\|_{r_x, r_y, \op}}{\lambda r_xd_x^* + (1-\lambda)r_yd_y^*},
    \end{equation}
    \begin{equation}\label{eq:VW_err_induction}
        |\sin\Theta(\bV^*,\bV^{(l+1)})|\leq \frac{2\|\cE_x\|_{\op}}{d_x^*\sqrt{1-\sin^2\theta(\bu^*,\bu^{(l)})}},\quad
        |\sin\Theta(\bW^*,\bW^{(l)})|\leq \frac{2\|\cE_y\|_{\op}}{d_y^*\sqrt{1-\sin^2\theta(\bu^*,\bu^{(l)})}}.
    \end{equation}
\end{claim} 
We first note that \eqref{eq:u_err_induction1} is directly implied by Assumption \ref{assump:init} when $l=0$. In the following, we will show that when \eqref{eq:u_err_induction1} holds for $l=k$, then \eqref{eq:u_err_induction1}-\eqref{eq:VW_err_induction} would all hold for $l = k + 1$. For notational simplicity, throughout the rest of the proof, we denote $|\sin\theta(\bu^*,\bu^{(k)})|$, $\|\sin\Theta(\bV^*,\bV^{(k)})\|_{\op}$, and $\|\sin\Theta(\bW^*,\bW^{(k)})\|_{\op}$ by $\varepsilon_{u,k}$, $\varepsilon_{v,k}$, and $\varepsilon_{w,k}$, respectively.

\paragraph{Analysis of $\bV^{(k+1)}$ and $\bW^{(k+1)}$.} We start with the update of network factors $\bV^{(k+1)}$ and $\bW^{(k+1)}$. Based on the Algorithm \ref{alg:single_tt}, we know in the $(k+1)$th iteration, the update of $\boldsymbol{V}^{(k+1)}$ is leading $r^x$ singular vectors of $\mathcal{X} \times_{3} \boldsymbol{u}^{(k)}$. Since $\mathcal{X} = d_x^* \cdot \boldsymbol{V}^*\boldsymbol{V}^{*\prime} \circ \boldsymbol{u}^* + \mathcal{E}_{x}$, by some basic tensor algebra, we have
    \begin{align*}
        \mathcal{X} \times_{3} \boldsymbol{u}^{(k)} &= (d_x^* \cdot \boldsymbol{V}^* \boldsymbol{V}^{*\prime} \circ \boldsymbol{u}^* + \mathcal{E}_{x}) \times_3 \boldsymbol{u}^{(k)}\\
        &= (d_x^* \cdot \boldsymbol{V}^* \boldsymbol{V}^{*\prime} \circ \boldsymbol{u}^*) \times_3 \boldsymbol{u}^{(k)} + \mathcal{E}_{x} \times_3 \boldsymbol{u}^{(k)}\\
        &= d_x^* \cdot \langle \boldsymbol{u}^*, \boldsymbol{u}^{(k)} \rangle \cdot \boldsymbol{V}^*\boldsymbol{V}^{*\prime} + \mathcal{E}_{x} \times_3 \boldsymbol{u}^{(k)},
    \end{align*}
    where $\langle \bu^*,\bu^{(k)}\rangle = \sqrt{1-\sin^2\theta(\bu^*,\bu^{(k)})}> 2\left(\frac{\|\cE_x\|_{\op}}{d_x^*}\vee \frac{\|\cE_y\|_{\op}}{d_y^*}\right)$, where the last inequality arises from our induction assumption that \eqref{eq:u_err_induction1} holds for $l=k$. By Weyl's inequality, we know that $\lambda_{r+1}(\cX\times_3 \bu^{(k+1)}) \leq \|\cE_x\times_3\bu^{(k)}\|_{\op}\leq \|\cE_x\|_{\op}<\frac{d_x}{2}$, where the third inequality is due to the definition of the tensor operator norm: $\|\cE_x\|_{\op} = \sup_{\bu\in \bbR^N,\bv\in\bbR^p,\bw\in\bbR^p}\cE_x\times_1\bv\times_2\bw\times_3\bu$, and the last inquality is due to the lower bound for $\langle \bu^*,\bu^{(k)}\rangle$ we just derived.
    Therefore, by Davis-Kahan's theorem \citep[see, e.g., Theorem 2.7 in][]{chen2021spectral}, we know that
    \begin{equation}\label{eq:V_err_update}
        \varepsilon_{v,k} =\|\sin\Theta(\bV^*,\bV^{(k+1)})\|_{\op}\leq \frac{\|\cE_x\|_{\op}}{d_x^* \langle \boldsymbol{u}^*, \boldsymbol{u}^{(k)} \rangle - \lambda_{r+1}(\cX\times_3 \bu^{(k+1)})}\leq \frac{2\|\cE_x\|_{\op}}{d_x^* \langle \boldsymbol{u}^*, \boldsymbol{u}^{(k)} \rangle} = \frac{2\|\cE_x\|_{\op}}{d_x^*\sqrt{1-\varepsilon_{u,k}^2}}.
    \end{equation}
 Furthermore, since we have assumed $\varepsilon_{u,k} = |\sin\theta(\bu^*,\bu^{(k)})|\leq \sqrt{1-8\left(\frac{\|\cE_x\|_{\op}^2}{d_x^{*2}}\vee\frac{\|\cE_y\|_{\op}^2}{d_y^{*2}}\right)}$, one can immediately show that
 \begin{equation}\label{eq:V_err_updatebnd}
         \varepsilon_{v,k}^2 \leq \frac{4\|\cE_x\|_{\op}^2}{d_x^{*2}(1-\varepsilon_{u,k}^2)}\leq \frac{1}{2}.
 \end{equation}
 Following the same argument, we can also bound $\|\sin\Theta(\bW,\bW^{(k+1)})\|_{\op}$ as follows:
    \begin{equation}\label{eq:W_err_update}
        \varepsilon_{w,k} =\|\sin\Theta(\bW^*,\bW^{(k+1)})\|_{\op}\leq \frac{2\|\cE_y\|_{\op}}{d_y^* \langle \boldsymbol{u}^*, \boldsymbol{u}^{(k)} \rangle} \leq \frac{2\|\cE_y\|_{\op}}{d_y^*\sqrt{1-\varepsilon_{u,k}^2}}\leq \frac{\sqrt{2}}{2}.
    \end{equation}
    Now we have verified that \eqref{eq:VW_err_induction} holds for $l=k+1$ when \eqref{eq:u_err_induction1} holds for $l=k$.
    
\paragraph{Analysis of $\bu^{(k+1)}$.}  Recall that we update $\boldsymbol{u}$ by $\boldsymbol{u}^{(k+1)} = \mathrm{Norm}\left(\lambda \left[\mathcal{X}; \boldsymbol{V}^{(k+1)}\right] + (1-\lambda) \left[\mathcal{Y}; \boldsymbol{W}^{(k+1)}\right] \right)$, where $\mathrm{Norm}(\cdot)$ is a normalization function that outputs a unit vector, and $[\cX; \bV^{(k+1)}]\in \bbR^{N}$ denotes the trace product defined in Section \ref{sec:notation}. Due to the definition of trace product, one has
    \begin{align*}
        \lambda \left[\mathcal{X}; \boldsymbol{V}^{(k+1)}\right] &= \lambda \left[d_{x}^* \cdot \boldsymbol{V}^*\boldsymbol{V}^{*\prime} \circ \boldsymbol{u}^* + \mathcal{E}_{x}; \boldsymbol{V}^{(k+1)}\right]\\
        & = \lambda \left[ d_{x}^* \cdot \boldsymbol{V}^*\boldsymbol{V}^{*\prime} \circ \boldsymbol{u}^*; \boldsymbol{V}^{(k+1)} \right] + \lambda \left[ \mathcal{E}_{x}; \boldsymbol{V}^{(k+1)} \right]\\
        & = \lambda d_{x}^* \mathrm{Tr}\left( (\boldsymbol{V}^{(k+1)})^{\prime} \boldsymbol{V}^* \boldsymbol{V}^{*\prime} \boldsymbol{V}^{(k+1)} \right) \cdot \boldsymbol{u}^* + \lambda \cdot \sum\limits_{i = 1}^{r} \mathcal{E}_{x} \times_{1} \boldsymbol{v}_{i}^{(k+1)} \times_{2} \boldsymbol{v}_{i}^{(k+1)}\\
        & = \lambda d_{x}^* \left\| \boldsymbol{V}^{*\prime} \boldsymbol{V}^{(k+1)} \right\|_{F}^{2} \cdot \boldsymbol{u}^* + \lambda \cdot \sum\limits_{i = 1}^{r} \mathcal{E}_{x} \times_{1} \boldsymbol{v}_{i}^{(k+1)} \times_{2} \boldsymbol{v}_{i}^{(k+1)}
    \end{align*}
    where in the third equality $\boldsymbol{v}_{i}$ denotes $i$th column of $\boldsymbol{V}$. By the same argument, we also have 
    \begin{align*}
        (1-\lambda) \left[\mathcal{Y}; \boldsymbol{W}^{(k+1)}\right]  = (1 - \lambda) d_{y}^* \left\|\boldsymbol{W}^{*\prime}\boldsymbol{W}^{(k+1)}\right\|_{F}^{2} \cdot \boldsymbol{u} + (1 - \lambda) \cdot \sum\limits_{i = 1}^{r} \mathcal{E}_{y} \times_{1} \boldsymbol{w}_{i}^{(k+1)} \times_{2} \boldsymbol{w}_{i}^{(k+1)},
    \end{align*}
    with $\boldsymbol{w}_{i}$ being the $i$th column of $\boldsymbol{W}$. Let 
    \begin{equation*}
    \begin{split}
        \alpha_{\lambda}^{(k+1)} = &\lambda d_x^*\|\bV^{*\prime}\bV^{(k+1)}\|_F^2  + (1-\lambda)d_y^*\|\bW^{*\prime}\bW^{(k+1)}\|_F^2,\\
        \boldsymbol{e}_{\lambda}^{(k+1)} = &\lambda[\cE_x;\bV^{(k+1)}] + (1-\lambda)[\cE_y;\bW^{(k+1)}]\\
        =&\lambda\sum_{i=1}^r\cE_x\times_1\bv_i^{(k+1)}\times_2\bv_i^{(k+1)}+(1-\lambda)\cE_y\times_1\bw_i^{(k+1)}\times_2\bw_i^{(k+1)}.
    \end{split}
    \end{equation*}
    Then we can also write the distance between $\bu^{(k+1)}$ and $\bu^*$ as follows:
    \begin{equation}\label{eq:u_update_bnd1}
        \begin{split}
            |\sin\theta(\bu^*,\bu^{(k+1)})| = &\sqrt{1 - \langle \bu^*,\bu^{(k+1)}\rangle^2}\\
            = &\sqrt{1 - \frac{\langle \bu^*,\alpha_{\lambda}^{(k+1)}\bu^* + \be_{\lambda}^{(k+1)}\rangle^2}{\|\alpha_{\lambda}^{(k+1)}\bu^* + \be_{\lambda}^{(k+1)}\|^2}}\\
            = &\sqrt{\frac{\|\be_{\lambda}^{(k+1)}\|_2^2 - \langle \be_{\lambda}^{(k+1)}, \bu^*\rangle^2}{\alpha_{\lambda}^{(k+1)2}+\|\be_{\lambda}^{(k+1)}\|_2^2 + 2\alpha_{\lambda}^{(k+1)}\langle \be_{\lambda}^{(k+1)}, \bu^*\rangle}}\\
            \leq &\frac{\|\be_{\lambda}^{(k+1)}\|_2}{\alpha_{\lambda}^{(k+1)} - \|\be_{\lambda}^{(k+1)}\|_2},
        \end{split}
    \end{equation}
where the third line is due to direct calculations, while the last line is due to Cauchy's inequality. 
    Recall that we have shown in \eqref{eq:V_err_updatebnd} and \eqref{eq:W_err_update} that $\|\sin\Theta(\bV^*,\bV^{(k+1)})\|_{\op}\leq \frac{\sqrt{2}}{2}$, $\|\sin\Theta(\bW^*,\bW^{(k+1)})\|_{\op}\leq \frac{\sqrt{2}}{2}$. Hence one can also show that
    \begin{equation*}
        \|\bV^{*\prime}\bV^{(k+1)}\|_F^2 = r_x - \|\sin\Theta(\bV^*, \bV^{(k+1)})\|_F^2 \geq r_x\left(1-\|\sin\Theta(\bV^*, \bV^{(k+1)})\|_{\op}^2\right)\geq \frac{r_x}{2}, 
    \end{equation*}
    suggesting that $\alpha_{\lambda}^{(k+1)}\geq \frac{1}{2}(\lambda r_x d_x^* + (1-\lambda)r_y d_y^*)$. On the other hand, the $\ell_2$ norm of the integrated noise term $\be_{\lambda}^{(k+1)}$ can be bounded as follows:
    \begin{equation*}
        \|\be_{\lambda}^{(k+1)}\|_2\leq \lambda r_x\|\cE_x\|_{\op} + (1-\lambda) r_y\|\cE_y\|_{\op}\leq \frac{1}{2}\alpha_{\lambda}^{(k+1)},
    \end{equation*}
    where we have applied the SNR condition that $d_x^*\geq 5\|\cE_x\|_{\op}$, $d_y^*\geq 5\|\cE_y\|_{\op}$(Assumption \ref{assump:SNR1}).
    Therefore, combining the results above with \eqref{eq:u_update_bnd1} leads us to 
    \begin{equation*}
            |\sin\theta(\bu^*,\bu^{(k+1)})|\leq \frac{2\|\be_{\lambda}^{(k+1)}\|_2}{\alpha_{\lambda}^{(k+1)}}\leq\frac{4\|\lambda \cE_x;(1-\lambda)\cE_y\|_{r_x,r_y,\op}}{\lambda r_xd_x^* + (1-\lambda)r_yd_y^*}.
    \end{equation*}
    Now we only need to show that \eqref{eq:u_err_induction1} holds for $l=k+1$. Since 
    \begin{align*}
        |\sin\theta(\bu^*,\bu^{(k+1)})|\leq &\frac{4\|\lambda \cE_x;(1-\lambda)\cE_y\|_{r_x,r_y,\op}}{\lambda r_xd_x^* + (1-\lambda)r_yd_y^*}\\
        \leq &\frac{4\lambda r_x\|\cE_x\|_{\op} + 4(1-\lambda) r_y\|\cE_y\|_{\op}}{\lambda r_xd_x^* + (1-\lambda)r_yd_y^*}\\
        \leq &\frac{4\|\cE_x\|_{\op}}{d_x^*}\vee\frac{4\|\cE_y\|_{\op}}{d_y^*},
    \end{align*}
   and $d_x^*\geq 5\|\cE_x\|_{\op}$, $d_y^*\geq \|\cE_y\|_{\op}$ by Assumption \ref{assump:SNR1}, we have $|\sin\theta(\bu^*,\bu^{(k+1)})|\leq\frac{4}{5}$. On the other hand, Assumption \ref{assump:SNR1} suggests that the R.H.S. of \eqref{eq:u_err_induction1} is lower bounded by $\frac{17}{25}<\frac{4}{5}$, and hence \eqref{eq:u_err_induction1} holds for $l=k+1$. The proof of Claim \ref{claim:induction1} is now complete. Furthermore, combining the fact that $|\sin\theta(\bu^*,\bu^{(k+1)})|\leq\frac{4}{5}$ and \eqref{eq:VW_err_induction}, we have also validated \eqref{eq:main_err2}. Our proof for Theorem \ref{thm:stat_converg_main} is now complete.
\end{proof}
\begin{proof}[Proof of Theorem \ref{thm:main_subgauss}]
We first note that $\|\cE_x\|_{\op}$ and $\|\cE_y\|_{\op}$ can both be bounded with high probability based on existing results for spectral norms of sub-Gaussian tensors. To deal with the dependency brought by the semi-symmetric constraint, we decompose them as upper and lower triangular components: $\cE_x = \cE_{x,1} + \cE_{x,2}$, $\cE_y = \cE_{y,1} + \cE_{y,2}$, where $\cE_{x,1},\,\cE_{x,2}\in \bbR^{p\times p \times N}$, $\cE_{y,1}$ satisfy $(\cE_{x,1})_{i,j,k} = \begin{cases}
    (\cE_x)_{i,j,k},&i<j,\\
    \frac{1}{2}(\cE_x)_{i,j,k},&i=j,\\
    0,&i>j,\\
\end{cases}$
$(\cE_{x,2})_{i,j,k} = \begin{cases}
    0,&i<j,\\
    \frac{1}{2}(\cE_x)_{i,j,k},&i=j,\\
    (\cE_x)_{i,j,k},&i>j,\\
\end{cases}$; $\cE_{y,1},\,\cE_{y,2}\in \bbR^{q\times q\times N}$ are defined similarly. $\cE_{x,1},\,\cE_{x,2},\,\cE_{y,1},\,\cE_{y,2}$ have independent, zero-mean, sub-Gaussian entries with sub-Gaussian parameter bounded by $\sigma$. Therefore, we can apply Theorem 1 and Lemma 1 in \cite{tomioka2014spectral} on them with $K=3$, $\delta = 2e^{-N}$. Then with probability at least $1-4e^{-N}$, we have
$\|\cE_x\|_{\op}\leq \|\cE_{x,1}\|_{\op} +\|\cE_{x,2}\|_{\op}\leq 16\sigma\sqrt{N+p}$, $\|\cE_y\|_{\op}\leq \|\cE_{y,1}\|_{\op} +\|\cE_{y,2}\|_{\op}\leq 16\sigma\sqrt{N+q}$. Therefore, Assumption \ref{assump:SNR2} implies Assumption \ref{assump:SNR1}. Combining Theorem \ref{thm:stat_converg_main} and Proposition \ref{thm:init}, and also noting that
\begin{equation*}
    \|\lambda\cE_x;(1-\lambda)\cE_y\|_{r_x,r_y,\op}\leq \lambda r_x\|\cE_x\|_{\op} + (1-\lambda) r_y\|\cE_y\|_{\op}\leq 16\sigma (\lambda r_x\sqrt{N+p} + (1-\lambda)r_y\sqrt{N+q}),
\end{equation*}
we have completed the proof of Theorem \ref{thm:main_subgauss}.
\end{proof}
\subsection{Proof of Warm Initialization (Corollary \ref{cor:warmInit}) and Spectral Initialization (Proposition \ref{thm:init})}
\begin{proof}[Proof of Corollary \ref{cor:warmInit}]
    Given Theorem \ref{thm:stat_converg_main}, it suffices to show the warm initialization $\bu^{(0)} = (\frac{1}{\sqrt{N}}, \dots,\frac{1}{\sqrt{N}})$ satisfies Assumption \ref{assump:init}. Now note that $\widehat{\mathrm{Var}}(\bu^*) = \frac{1}{N}\sum_i(\bu_i^* - \frac{1}{N}\sum_j\bu^*_j)^2 = \frac{1}{N}\sum_i(\bu_i^*)^2 - (\widehat{\bbE}(u^*))^2 = \frac{1}{N} - (\widehat{\bbE}(u^*))^2$, where we have utilized the fact that $\bu^*$ is a unit vector. Therefore, $\frac{\widehat{\mathrm{Var}}(u^*)}{(\widehat{\bbE}(u^*))^2}\leq \left(\frac{d_x^{*2}}{8\|\cE_x\|_{\op}^2}\wedge\frac{d_y^{*2}}{8\|\cE_y\|_{\op}^2}\right) - 1$ implies $(\widehat{\bbE}(u^*))^2\geq \frac{8}{N}\left(\frac{\|\cE_x\|_{\op}^2}{d_x^{*2}}\vee\frac{\|\cE_y\|_{\op}^2}{d_y^{*2}}\right)$, and hence
    \begin{equation*}
        |\sin\theta(\bu^*,\bu^{(0)})|^2 = 1-(\bu^{*'}\bu^{(0)})^2 \leq 1-N(\widehat{\bbE}(u^*))^2\leq 1-8\left(\frac{\|\cE_x\|_{\op}^2}{d_x^{*2}}\vee\frac{\|\cE_y\|_{\op}^2}{d_y^{*2}}\right),
    \end{equation*}
    which is Assumption \ref{assump:init}. The proof is now complete.
\end{proof}
\begin{proof}[Proof of Proposition \ref{thm:init}] Recall that $\mathcal{X} = d_{x}^*\cdot \boldsymbol{V}^*\boldsymbol{V}^{*\prime} \circ \boldsymbol{u}^* + \mathcal{E}_{x}$ and $\mathcal{Y} = d_{y}^* \cdot \boldsymbol{W}^*\boldsymbol{W}^{*\prime} \circ \boldsymbol{u}^* + \mathcal{E}_{y}$. By our construction, $\boldsymbol{u}^{(0)}$ is the leading left singular vector of $\left[\lambda\mathcal{M}_{3}(\mathcal{X}), (1-\lambda)\mathcal{M}_{3}(\mathcal{Y}) \right] \in \mathbb{R}^{N \times (p^{2} + q^{2})}$, which can be written as
    \begin{align*}
        \left[\lambda\mathcal{M}_{3}(\mathcal{X}), (1-\lambda)\mathcal{M}_{3}(\mathcal{Y}) \right] &= \boldsymbol{u}^*\left[\lambda d_{x}^* \cdot \mathrm{Vec}\left(\boldsymbol{V}^*\boldsymbol{V}^{*\prime}\right)^\prime, (1-\lambda)d_{y}^*\cdot \mathrm{Vec}\left(\boldsymbol{W}^*\boldsymbol{W}^{*\prime}\right)^\prime \right] + \left[\lambda\mathcal{M}_{3}(\mathcal{E}_{x}), \mathcal{M}_{3}(\mathcal{E}_{y}) \right]\\
        &= d_{\lambda}\cdot \boldsymbol{u}^*\boldsymbol{z}^{\prime} + \bE_{\lambda},
    \end{align*}
    by letting $d_{\lambda} = \sqrt{\lambda^2r_xd_x^{*2} + (1-\lambda)^2r_yd_y^{*2}}$, $\boldsymbol{z} = \mathrm{Norm}(\left[\lambda d_{x}^* \cdot \mathrm{Vec}\left(\boldsymbol{V}\boldsymbol{V}^{\prime}\right)^\prime, (1-\lambda)d_{y}^*\cdot \mathrm{Vec}\left(\boldsymbol{W}\boldsymbol{W}^{\prime}\right)^\prime \right])$, where $\mathrm{Norm}(\cdot)$ is a normalization function that outputs a unit vector, and $\bE_{\lambda} = \left[\lambda\mathcal{M}_{3}(\mathcal{E}_{x}), (1-\lambda)\mathcal{M}_{3}(\mathcal{E}_{y}) \right]$. We can also think of $\boldsymbol{u}^{(0)}$ as the leading left singular vector of $$(d_{\lambda}\cdot \boldsymbol{u}^*\boldsymbol{z}^{\prime} + \bE_{\lambda})(d_{\lambda}\cdot \boldsymbol{u}^*\boldsymbol{z}^{\prime} + \bE_{\lambda})' = d_{\lambda}^2\bu^*\bu^{*\prime} + d_{\lambda}\bu^*(\bE_{\lambda}\bz)' + d_{\lambda}\bE_{\lambda}\bz\bu^{*\prime}+ \bE_{\lambda}\bE_{\lambda}'.$$
    
Recall Assumption \ref{assump:subGaussNoise} on the distributional properties of noise tensors. 
Since the variances of $\cE_x$ and $\cE_y$ are homogeneous across the third mode, we know that $$\bbE\|(\bE_{\lambda})_{k,:}\|_2^2 = \lambda^2\sum_{i,j=1}^p\mathrm{Var}((\cE_x)_{i,j,k}) + (1-\lambda)^2 \sum_{i,j=1}^q\mathrm{Var}((\cE_y)_{i,j,k})$$ take the same value for all $1\leq k\leq N$. We also denote this variance term by $\sigma_{\lambda}^2$, and thus $\bbE \bE_{\lambda}\bE_{\lambda}' = \sigma_{\lambda}^2\boldsymbol{I}_{N\times N}$. Therefore, $\bu^{(0)}$ is also the leading left singular vector of $d_{\lambda}^2\bu^*\bu^{*\prime} + d_{\lambda}\bu^*(\bE_{\lambda}\bz)' + d_{\lambda}\bE_{\lambda}\bz\bu^{*\prime}+ \bE_{\lambda}\bE_{\lambda}' - \bbE(\bE_{\lambda}\bE_{\lambda}')$, with the first term being the rank-1 signal and the rest three terms being random mean-zero perturbations. We would like to invoke the Davis-Kahan's theorem to upper bound $\sin\theta(\bu^*,\bu^{(0)})$, but before that, we first show an upper bound of the spectral norm of the noise term $d_{\lambda}\bu^*(\bE_{\lambda}\bz)' + d_{\lambda}\bE_{\lambda}\bz\bu^{*\prime}+ \bE_{\lambda}\bE_{\lambda}' - \bbE(\bE_{\lambda}\bE_{\lambda}')$. 

To avoid the dependency issue brought by the semi-symmetric constraint, we define a reorganized noise matrix $\widetilde{\bE}_{\lambda}\in \bbR^{N\times [p(p-1) + q(q-1)]}$ that satisfies $(\widetilde{\bE}_{\lambda})_{i,:} = [\lambda \be_{x,i}', (1-\lambda)\be_{y,i}']$ where $\be_{x,i}\in \bbR^{p(p-1)}$ ($\be_{y,i}\in \bbR^{q(q-1)}$) consists of diagonal and upper triangular entries of $(\cE_x)_{:,:,i}$ ($(\cE_y)_{:,:,i}$):
    \begin{align*}
        \be_{x,i} = [\mathrm{diag}((\cE_x)_{j,j,i})_{1\leq j\leq p},\sqrt{2}((\cE_x)_{j,k,i})_{1\leq j<k\leq p}].
    \end{align*}

We also define $\widetilde{\bz}\in \bbR^{p(p-1)+q(q-1)}$ based on $\bz$ similarly:
\begin{equation*}
    \widetilde{\bz} = \mathrm{Norm}([\lambda d_x^*[\mathrm{diag}(\bV\bV'), \sqrt{2}((\bV\bV')_{j,k})_{1\leq j<k\leq p}],(1-\lambda) d_y^*[\mathrm{diag}(\bW\bW'), \sqrt{2}((\bW\bW')_{j,k})_{1\leq j<k\leq q}]).
\end{equation*}
Then we can write $\bE_{\lambda}\bz = \widetilde{\bE}_{\lambda}\widetilde{\bz}$, $\bE_{\lambda}\bE_{\lambda}' = \widetilde{\bE}_{\lambda}\widetilde{\bE}_{\lambda}'$, where $\widetilde{\bE}_{\lambda}$ has independent, mean-zero, sub-Gaussian-$\sqrt{2}\sigma$ entries. 

    Now we first bound $\|d_{\lambda}\widetilde{\bE}_{\lambda}\widetilde{\bz}\bu^{*\prime}\| = \|d_{\lambda}\bu^*(\widetilde{\bE}_{\lambda}\widetilde{\bz})'\| = d_{\lambda}\|\widetilde{\bE}_{\lambda}\widetilde{\bz}\|_2$. Since $\{(\widetilde{\bE}_{\lambda}\widetilde{\bz})_i\}_{i=1}^N$ are independent zero-mean sub-Gaussian random variables with sub-Gaussian parameter $\sqrt{\sum_{i}2\sigma^2\widetilde{\bz}_i^2}=\sqrt{2}\sigma$, and hence $(\widetilde{\bE}_{\lambda}\widetilde{\bz})_i^2-\bbE(\widetilde{\bE}_{\lambda}\widetilde{\bz})_i^2$ are independent sub-exponential-$4\sigma^2$ random variables. Therefore, we can apply the Bernstein-type inequality \citep[see e.g., Proposition 5.6 in][]{vershynin2010introduction} for sub-exponential random variables to obtain the following:
    $$
    \|\widetilde{\bE}_{\lambda}\widetilde{\bz}\|_2^2 \leq \bbE\|\widetilde{\bE}_{\lambda}\widetilde{\bz}\|_2^2+C\sigma^2\sqrt{N\log N}\leq C\sigma^2N
    $$
    with probability at least $1-N^{-c}$, which implies $$
    \|d_{\lambda}\bu^*(\bE_{\lambda}\bz)' + d_{\lambda}\bE_{\lambda}\bz\bu^{*\prime}\|\leq C\sigma \sqrt{N}d_{\lambda}.$$

    For the last noise term $\bE_{\lambda}\bE_{\lambda}' - \bbE(\bE_{\lambda}\bE_{\lambda}') = \widetilde{\bE}_{\lambda}\widetilde{\bE}_{\lambda}' - \bbE(\widetilde{\bE}_{\lambda}\widetilde{\bE}_{\lambda}')$, we would like to apply a technical lemma from \cite{zhou2023deflated}. In particular, since $\widetilde{\bE}_{\lambda}$ has independent zero-mean sub-Gaussian-$\sqrt{2}\sigma$, it satisfies Assumption 3 in \cite{zhou2023deflated} with $\omega_{\max} = \sqrt{2}\sigma$, $B = C\sigma\sqrt{\log N}$ and $\varepsilon = N^{-c}$. We apply Lemma 7 in \cite{zhou2023deflated} with $\bE = \widetilde{\bE}_{\lambda}$, which gives us 
    $$
    \|\widetilde{\bE}_{\lambda}\widetilde{\bE}_{\lambda}' - \mathrm{diag}(\widetilde{\bE}_{\lambda}\widetilde{\bE}_{\lambda}')\| \leq C\sigma^2 (N + \sqrt{N(p^2 + q^2)})\log N,
    $$
    with probability at least $1-N^{-c}$. Furthermore, since $\|\mathrm{diag}(\widetilde{\bE}_{\lambda}\widetilde{\bE}_{\lambda}') - \bbE(\widetilde{\bE}_{\lambda}\widetilde{\bE}_{\lambda}')\|=\max_i\|(\widetilde{\bE}_{\lambda})_{i,:}\|_2^2 - \sigma_{\lambda}^2$ and $\|(\widetilde{\bE}_{\lambda})_{i,:}\|_2^2 - \sigma_{\lambda}^2$ is the sum of $p(p-1) + q(q-1)$ independent zero-mean sub-exponential-$4\sigma^2$ random variables, we can again apply the Bernstein-type inequality \citep[see e.g., Proposition 5.6 in][]{vershynin2010introduction} to derive the following bounds:
    \begin{equation*}
        \begin{split}
            &\bbP(\|\mathrm{diag}(\widetilde{\bE}_{\lambda}\widetilde{\bE}_{\lambda}') - \bbE(\widetilde{\bE}_{\lambda}\widetilde{\bE}_{\lambda}')\|>C\sigma^2\sqrt{p^2+q^2}\sqrt{\log N})\\
            \leq &\sum_{i=1}^N\bbP(|\|(\widetilde{\bE}_{\lambda})_{i,:}\|_2^2 - \sigma_{\lambda}^2|>C\sigma^2\sqrt{p^2+q^2}\sqrt{\log N})\\
            \leq&2N\exp\{-c\log N\}\\
            \leq&N^{-c}.
        \end{split}
    \end{equation*}
    Hence, with probability at least $1-CN^{-c}$, 
    \begin{equation*}
        \begin{split}
            &\|d_{\lambda}\bu^*(\bE_{\lambda}\bz)' + d_{\lambda}\bE_{\lambda}\bz\bu^{*\prime}+ \bE_{\lambda}\bE_{\lambda}' - \bbE(\bE_{\lambda}\bE_{\lambda}')\|\\
            \leq &Cd_{\lambda}\sigma\sqrt{N} + C\sigma^2(N+\sqrt{N(p^2+q^2)})\log N\\
            \leq &\frac{1}{2}d_{\lambda}^2,
        \end{split}
    \end{equation*}
    where the last inequality is due to Assumption \ref{assump:SNR2}.
    
    Now we return to the original problem and invoke the Davis-Kahan's theorem \citep[see, e.g., Theorem 2.7 in][]{chen2021spectral}:
    \begin{equation*}
    \begin{split}
        \sin\theta(\bu^*,\bu^{(0)})\leq &\frac{2\|d_{\lambda}\bu^*(\bE_{\lambda}\bz)' + d_{\lambda}\bE_{\lambda}\bz\bu^{*\prime}+ \bE_{\lambda}\bE_{\lambda}' - \bbE(\bE_{\lambda}\bE_{\lambda}')\|}{d_{\lambda}^2}\\
        \leq &\frac{C\sigma\sqrt{N}}{d_{\lambda}} + \frac{C\sigma^2(N+\sqrt{N(p^2+q^2)})\log N}{d_{\lambda}^2}\\
        \leq&\frac{C\sigma\left(\sqrt{N} + (N(p^2+q^2))^{\frac{1}{4}}\right)\sqrt{\log N}}{d_{\lambda}}
    \end{split}
    \end{equation*}
\end{proof}

\subsection{Proof of Theoretical Guarantees for the Generalized JisstPCA: Theorem \ref{thm:converg_D}}
    In this section, we will first establish bounds for $\bu^{(k)}$, $\bV^{(k)}$, and $\bW^{(k)}$ under an initialization condition and some deterministic conditions for $\cE_x$, $\cE_y$; we then prove that under Assumptions \ref{assump:dJisst_SNR1} and \ref{assump:dJisst_SNR2}, all these conditions hold with high probability. 
    \subsubsection{Deterministic Bounds for the Joint Factor}
    In particular, we will first assume the following two conditions hold and will revisit and show them hold with desired probabilities at the end of the proof.
    \begin{cond}[Initialization condition for generalized JisstPCA]\label{cond:init}
        $$
        |\sin\theta(\bu^*,\bu^{(0)})|\leq \min\left\{\sqrt{1-\frac{32r_x\|\cE_x\|_{\op}^2}{\|\bD^{*}_x\|_F^2}},\sqrt{1-\frac{32r_y\|\cE_y\|_{\op}^2}{\|\bD^{*}_y\|_F^2}}\right\}.
        $$
    \end{cond}
    \begin{cond}[Deterministic SNR condition for generalized JisstPCA]\label{cond:SNR1}
        $$
        \|\bD^{*}_x\|_F\geq 12\sqrt{r_x}\|\cE_x\|_{\op},\quad \|\bD^{*}_x\|_F\geq 12\sqrt{r_y}\|\cE_y\|_{\op}.
        $$
    \end{cond}
    We first note that in the generalized JisstPCA algorithm, the update $\bu^{(1)}$ is computed as $\bu^{(1)} = \mathrm{Norm}(\lambda[\cX; \bV^{(1)}, \bD^{(1)}_x] + (1-\lambda)[\cY; \bW^{(1)}, \bD^{(1)}_y])$, where $\mathrm{Norm}(\cdot)$ is a normalization function outputting a unit vector, and $[\cX; \bV^{(1)}, \bD^{(1)}_x]$, $[\cY; \bW^{(1)}, \bD^{(1)}_y]\in \bbR^{N}$ are the trace products:
    \begin{equation*}
        [\cX; \bV^{(1)}, \bD^{(1)}_x]_i = \langle \cX_{:,:,i}, \bV^{(1)}\bD^{(1)}_x\bV^{(1)\prime}\rangle,\quad [\cY; \bW^{(1)}, \bD^{(1)}_y]_i = \langle \cY_{:,:,i}, \bW^{(1)}\bD^{(1)}_y\bW^{(1)\prime}\rangle.
    \end{equation*}
    Due to the generative model for $\cX$ and $\cY$ in \eqref{gen_model} with $K=1$, we have
    \begin{equation*}
    \begin{split}
        [\cX; \bV^{(1)}, \bD^{(1)}_x] = &\boldsymbol{V}^*\boldsymbol{D}^{x*} \boldsymbol{V}^{*\prime} \circ \boldsymbol{u}^*;  \bV^{(1)}, \bD^{(1)}_x] + [\mathcal{E}_{x}; \bV^{(1)}, \bD^{(1)}_x]\\
        = & \langle \boldsymbol{V}^*\boldsymbol{D}^{x*} \boldsymbol{V}^{*\prime}, \bV^{(1)}\bD^{(1)}_x\bV^{(1)\prime}\rangle \boldsymbol{u}^* + \sum_{i=1}^{r^x}(\bD^{(1)}_x)_{i,i}\mathcal{E}_{x}\times_1\bv^{(1)}_i\times_2\bv^{(1)}_i,
    \end{split}
    \end{equation*}
    and 
     \begin{equation*}
        [\cY; \bW^{(1)}, \bD^{(1)}_y] = \langle \boldsymbol{W}^*\boldsymbol{D}^{y*} \boldsymbol{W}^{*\prime}, \bW^{(1)}\bD^{(1)}_y\bW^{(1)\prime}\rangle \boldsymbol{u}^* + \sum_{i=1}^{r^y}(\bD^{(1)}_y)_{i,i}\mathcal{E}_{y}\times_1\bw^{(1)}_i\times_2\bw^{(1)}_i,
    \end{equation*}
    where $\bv^{(1)}_i$ and $\bw^{(1)}_i$ are the $i$th columns of $\bV^{(1)}\in \bbR^{p\times r^x}$ and $\bW^{(1)}\in \bbR^{p\times r^y}$.
    Let $$\alpha_{\lambda} = \lambda\langle \boldsymbol{V}^*\boldsymbol{D}^{x*} \boldsymbol{V}^{*\prime}, \bV^{(1)}\bD^{(1)}_x\bV^{(1)\prime}\rangle + (1-\lambda)\langle \boldsymbol{W}^*\boldsymbol{D}^{y*} \boldsymbol{W}^{*\prime}, \bW^{(1)}\bD^{(1)}_y\bW^{(1)\prime}\rangle$$ be the pooled signal, and let 
    $$
    \be_{\lambda} = \lambda\sum_{i=1}^{r_x}(\bD^{(1)}_x)_{i,i}\mathcal{E}_{x}\times_1\bv^{(1)}_i\times_2\bv^{(1)}_i + (1-\lambda)\sum_{i=1}^{r^y}(\bD^{(1)}_y)_{i,i}\mathcal{E}_{y}\times_1\bw^{(1)}_i\times_2\bw^{(1)}_i
    $$
    be the pooled noise in $\lambda[\cX; \bV^{(1)}, \bD^{(1)}_x] + (1-\lambda)[\cY; \bW^{(1)}, \bD^{(1)}_y]$. Then similar to the proof of Theorem \ref{thm:stat_converg_main}, we have
    \begin{equation}\label{eq:u_update_dist}
        |\sin\theta(\bu^*,\bu^{(1)})| \leq \frac{\|\be_{\lambda}\|_2}{\alpha_{\lambda} - \|\be_{\lambda}\|_2}.
    \end{equation}
    In the following, we provide a lower bound for $\alpha_{\lambda}$ and an upper bound for $\|\be_{\lambda}\|_2$. To lower bound $\alpha_{\lambda}$, we first notice that $\bV^{(1)}\bD^{(1)}_x\bV^{(1)\prime}$ is the top rank-$r^x$ SVD of $\cX\times \bu^{(0)}$, meaning that it is also its best rank-$r^x$ approximation in Frobenius norm error. In the meantime, we have the decomposition $\cX\times_3 \bu^{(0)} = \langle \bu^*, \bu^{(0)}\rangle \bV^*\bD^{*}_x\bV^{*\prime} + \cE_x\times_3\bu^{(0)}$, which is a rank-$r^x$ matrix plus a perturbation. The following lemma shows that given a perturbed low-rank matrix, one can upper bound the estimation error of the top SVD solution for the true low-rank matrix.
    \begin{lem}\label{lem:perturb_frobenius_err}
        Suppose that $\bX = \bX^* + \bE$ where $\bX^*$ is of rank $r$. Let $\bX_r = \hat{\bU}\hat{\boldsymbol{\Lambda}}\hat{\bV}^{\prime}$ be the top-$r$ SVD of $\bX$, then $\|\bX_r-\bX^*\|_F\leq 2\sqrt{2r}\|\bE\|$.
    \end{lem}
    Applying Lemma \ref{lem:perturb_frobenius_err} with $\bX^* = \langle \bu^*, \bu^{(0)}\rangle \bV^*\bD^{*}_x\bV^{*\prime}$, $\bE = \cE_x\times_3\bu^{(0)}$, we then have
    \begin{equation}\label{eq:VDV_err}
        \|\bV^{(1)}\bD^{(1)}_x\bV^{(1)\prime} - \langle \bu^*, \bu^{(0)}\rangle \bV^*\bD^{*}_x\bV^{*\prime}\|_F\leq 2\sqrt{2r^x}\|\cE_x\times_3\bu^{(0)}\|\leq 2\sqrt{2r^x}\|\cE_x\|_{\op}.
    \end{equation}
    The Frobenious norm error bound \eqref{eq:VDV_err} also suggests
    \begin{equation*}
        \begin{split}
            \langle \bV^{(1)}\bD^{(1)}_x\bV^{(1)\prime}, \bV^*\bD^{*}_x\bV^{*\prime}\rangle \geq &\langle \bu^*, \bu^{(0)}\rangle\|\bV^*\bD^{*}_x\bV^{*\prime}\|_F^2 - 2\sqrt{2r^x}\|\cE_x\|_{\op}\|\bV^*\bD^{*}_x\bV^{*\prime}\|_F\\
            \geq &\frac{1}{2}\langle \bu^*, \bu^{(0)}\rangle\|\bD^{*}_x\|_F^2,
        \end{split}
    \end{equation*}
    where the last line is due to Condition \ref{cond:init}, which implies  \begin{equation}\label{eq:update_SNR_generalized}
        2\sqrt{2r^x}\|\cE_x\|_{\op}\leq \frac{1}{2}\sqrt{1-\sin^2\theta(\bu,\bu^{(0)})}\|\bD^{*}_x\|_F=\frac{1}{2}\langle\bu^*,\bu^{(0)}\rangle \|\bV^*\bD^{*}_x\bV^{*\prime}\|_F.
    \end{equation} Similarly, one can show that $\langle \bW^{(1)}\bD^{(1)}_y\bW^{(1)\prime}, \bW^*\bD^{*}_y\bW^{*\prime}\rangle \geq \frac{1}{2}\langle \bu^*, \bu^{(0)}\rangle\|\bD^{*}_y\|_F^2$, leading to a lower bound for $\alpha_{\lambda}$:
    $$
    \alpha_{\lambda}\geq \frac{1}{2}\langle \bu^*, \bu^{(0)}\rangle\big(\lambda\|\bD^{*}_x\|_F^2+(1-\lambda)\|\bD^{*}_y\|_F^2\big).
    $$
    On the other hand, to upper bound $\|\be_{\lambda}\|_2$, we note that
    \begin{equation*}
    \begin{split}
        \left\|\sum_{i=1}^{r_x}(\bD^{(1)}_x)_{i,i}\mathcal{E}_{x}\times_1\bv^{(1)}_i\times_2\bv^{(1)}_i\right\|_2\leq &\sum_{i=1}^{r_x}|(\bD^{(1)}_x)_{i,i}|\big\|\mathcal{E}_{x}\times_1\bv^{(1)}_i\times_2\bv^{(1)}_i\big\|_2\\
        \leq&\|\bD^{(1)}_x\|_F\sqrt{\sum_{i=1}^{r_x}\big\|\mathcal{E}_{x}\times_1\bv^{(1)}_i\times_2\bv^{(1)}_i\big\|_2^2}\\
        \leq&\sqrt{r_x}\|\bD^{(1)}_x\|_F\|\cE_x\|_{\op}\\
        \leq&\sqrt{r_x}(\langle \bu^*,\bu^{(0)}\rangle\|\bD^{*}_x\|_F + 2\sqrt{2r_x}\|\cE_x\|_{\op})\|\cE_x\|_{\op}\\
        \leq&\frac{3}{2}\sqrt{r_x}\langle \bu^*,\bu^{(0)}\rangle\|\bD^{*}_x\|_F\|\cE_x\|_{\op}.
    \end{split}
    \end{equation*}
    where the fourth line is due to \eqref{eq:VDV_err}, and the last line is due to \eqref{eq:update_SNR_generalized}. Recall Condition \ref{cond:SNR1}, we know that $\left\|\sum_{i=1}^{r_x}(\bD^{(1)}_x)_{i,i}\mathcal{E}_{x}\times_1\bv^{(1)}_i\times_2\bv^{(1)}_i\right\|_2\leq \frac{1}{8}\sqrt{r_x}\langle \bu^*,\bu^{(0)}\rangle\|\bD_x^*\|_F^2\leq \frac{1}{4}\alpha_{\lambda}$. Similarly, we can apply the same argument on $\left\|\sum_{i=1}^{r_y}(\bD^{(1)}_y)_{i,i}\mathcal{E}_{y}\times_1\bw^{(1)}_i\times_2\bw^{(1)}_i\right\|_2$ and obtain that $$\left\|\sum_{i=1}^{r_y}(\bD^{(1)}_y)_{i,i}\mathcal{E}_{y}\times_1\bw^{(1)}_i\times_2\bw^{(1)}_i\right\|_2\leq \frac{3}{2}\sqrt{r_y}\langle\bu^*,\bu^{(0)}\rangle \|\bD_y^*\|_F\leq \frac{1}{4}\alpha_{\lambda}.$$ Therefore, plugging in these bounds into \eqref{eq:u_update_dist} gives us a deterministic upper bound for $\big|\sin\theta(\bu^*,\bu^{(1)})\big|$ under Conditions \ref{cond:init} and \ref{cond:SNR1}:
    \begin{equation}\label{eq:generalized_u_bnd_deterministc}
    \begin{split}
         \big|\sin\theta(\bu^*,\bu^{(1)})\big| \leq &\frac{4\|\be_{\lambda}\|_2}{3\alpha_{\lambda}}\\
         \leq&\frac{8\|\be_{\lambda}\|_2}{3\langle\bu^*,\bu^{(0)}\rangle(\lambda\|\bD_x^*\|_F^2 + (1-\lambda)\|\bD_y^*\|_F^2)}\\
         \leq&\frac{4\lambda\sqrt{r_x}\|\bD_x^*\|_F\|\cE_x\|_{\op} + 4(1-\lambda)\sqrt{r_y}\|\bD_y^*\|_F\|\cE_y\|_{\op}}{\lambda\|\bD_x^*\|_F^2 + (1-\lambda)\|\bD_y^*\|_F^2}.
    \end{split}
    \end{equation}
    Furthermore, \eqref{eq:generalized_u_bnd_deterministc} and Condition \ref{cond:SNR1} together imply
    \begin{equation*}
    \begin{split}
        \big|\sin\theta(\bu^*,\bu^{(1)})\big| \leq &\frac{4\sqrt{r_x}\|\cE_x\|_{\op}}{\|\bD_x^*\|_F} \vee \frac{4\sqrt{r_y}\|\cE_y\|_{\op}}{\|\bD_y^*\|_F}\\
        \leq &\frac{1}{3}\\
        \leq&\sqrt{1-\frac{32r_x\|\cE_x\|_{\op}^2}{\|\bD_y^*\|_F^2}}\vee\sqrt{1-\frac{32r_y\|\cE_y\|_{\op}^2}{\|\bD_y^*\|_F^2}},
    \end{split}
    \end{equation*} 
    and thus the initialization condition (Condition \ref{cond:init}) is also satisfied by $\bu^{(1)}$. Therefore, we can apply the same arguments with initialization $\bu^{(k)}$ for $k\geq 0$ to show that \eqref{eq:generalized_u_bnd_deterministc} holds not only for $\bu^{(1)}$, but for $\bu^{(k+1)}$ with any $k\geq 0$.
    \subsubsection{Deterministic Bounds for Network Factors}
    Now we turn to the estimation error of $\bV^{(k+1)}$ and $\bW^{(k+1)}$ for $k\geq 1$. To achieve this, we require the following SNR condition.
    \begin{cond}[Deterministic SNR condition for network factors in generalized Jisst PCA]\label{cond:SNR2}
        $$
        \sigma_{r_x}(\bD_x^*)\geq \frac{3\sqrt{2}}{2}\|\cE_x\|_{\op},\quad \sigma_{r_y}(\bD_y^*)\geq \frac{3\sqrt{2}}{2}\|\cE_y\|_{\op}.
        $$
    \end{cond}
    Recall our updating rules for $\bV^{(k+1)}$ and $\bW^{(k+1)}$ in Algorithm \ref{single_diag}: $\bV^{(k+1)}$ and $\bW^{(k+1)}$ are the leading $r_x$ and $r_y$ singular vectors of $\cX\times_3\bu^{(k)}$ and $\cY\times_3\bu^{(k)}$, respectively. We can also write
    \begin{equation*}
    \begin{split}
        \cX\times_3\bu^{(k)} = &\langle \bu^*, \bu^{(k)}\rangle \bV^*\bD_x^*\bV^{*\prime} + \cE_x\times_3\bu^{(k)},\\
        \cY\times_3\bu^{(k)} = &\langle \bu^*, \bu^{(k)}\rangle \bW^*\bD_y^*\bW^{*\prime} + \cE_y\times_3\bu^{(k)},
    \end{split}
    \end{equation*}
    with signal matrices $\langle \bu^*, \bu^{(k)}\rangle \bV^*\bD_x^*\bV^{*\prime}$ and $\langle \bu^*, \bu^{(k)}\rangle \bW^*\bD_y^*\bW^{*\prime}$ and noise matrices $\cE_x\times_3\bu^{(k)}$ and $\cE_y\times_3\bu^{(k)}$. As we have shown earlier, the joint factor $\bu^{(k)}$ satisfies $|\sin\theta(\bu^*,\bu^{(k)})|\leq \frac{1}{3}$ for $k\geq 1$, which implies $\langle \bu^*,\bu^{(k)}\rangle\geq \frac{2\sqrt{2}}{3}$. Furthermore, similar to the proof of Theorem \ref{thm:stat_converg_main}, it is not hard to show that $\|\cE_x\times_3\bu^{(k)}\|\leq \|\cE_x\|_{\op}$ and $\|\cE_y\times_3\bu^{(k)}\|\leq \|\cE_y\|_{\op}$. Combining these two results with Condition \ref{cond:SNR2}, we have 
    $\|\cE_x\times_3\bu^{(k)}\|\leq \frac{1}{2}\sigma_{r_x}(\langle \bu^*, \bu^{(k)}\rangle \bV^*\bD_x^*\bV^{*\prime})$, $\|\cE_y\times_3\bu^{(k)}\|\leq \frac{1}{2}\sigma_{r_y}(\langle \bu^*, \bu^{(k)}\rangle \bW^*\bD_y^*\bW^{*\prime})$. Now we can invoke the Davis-Kahan's theorem \citep[see, e.g., Theorem 2.7 in][]{chen2021spectral} to obtain the following: 
    \begin{equation}\label{eq:generalized_VW_bnd_deterministic}
        \begin{split}
            \|\sin\Theta(\bV^*,\bV^{(k+1)})\|\leq \frac{3\sqrt{2}\|\cE_x\|_{\op}}{2\sigma_{r_x}(\bD_x^*)\rangle},\quad \|\sin\Theta(\bW^*,\bW^{(k+1)})\|\leq \frac{3\sqrt{2}\|\cE_y\|_{\op}}{2\sigma_{r_y}(\bD_y^*)\rangle}
        \end{split}
    \end{equation}
    hold for any $k\geq 2$ as long as Conditions \ref{cond:init}-\ref{cond:SNR2} are satisfied.
    \subsubsection{Probabilistic Bounds with Sub-Gaussian Noise}
    Now we would like to show that under Assumption \ref{assump:dJisst_SNR1}, Conditions \ref{cond:init} and \ref{cond:SNR1} are satisfied with high probability; If Assumption \ref{assump:dJisst_SNR2} also holds, then Conditions \ref{cond:init}-\ref{cond:SNR2} are all satisfied. As has been shown in the proof of Theorem \ref{thm:main_subgauss}, with probability at least $1-4\exp\{-N\}$, $\|\cE_x\|_{\op}\leq C\sigma(\sqrt{N+p})$, $\|\cE_y\|_{\op}\leq C\sigma(\sqrt{N+q})$. As long as the constant $C>0$ in Assumption \ref{assump:dJisst_SNR1} is chosen sufficiently large, Condition \ref{cond:SNR1} holds. In addition, Condition \ref{cond:SNR1} also implies $\sqrt{1-\frac{32r_x\|\cE_x\|_{\op}^2}{\|\bD^*_x\|_F^2}}, \sqrt{1-\frac{32r_y\|\cE_y\|_{\op}^2}{\|\bD^*_y\|_F^2}}\geq \sqrt{1-\frac{2}{9}} = \frac{\sqrt{7}}{3}$. To show that Condition \ref{cond:init} holds, we can apply the same arguments as in the proof of Proposition \ref{thm:init}. Note that similar to the vanilla setting with scalar $d_x^*$ and $d_y^*$, we can also write $\bu^{(0)}$ as the top left singular vector of $d_{\lambda}\bu^*\bz' + \bE_{\lambda}$, where $d_{\lambda} = \sqrt{\lambda^2\|\bD^*_x\|_F^2 + (1-\lambda)^2\|\bD_y^*\|_F^2}$, $\bz = \mathrm{Norm}(\left[\lambda  \mathrm{Vec}\left(\boldsymbol{V}\bD_x^*\boldsymbol{V}^{\prime}\right)^\prime, (1-\lambda) \mathrm{Vec}\left(\boldsymbol{W}\bD_y^*\boldsymbol{W}^{\prime}\right)^\prime \right]$, and $\bE_{\lambda} = [\lambda\mathcal{M}_3(\cE_x), (1-\lambda)\mathcal{M}_3(\cE_y)]$. By Assumption \ref{assump:dJisst_SNR1} which is analogous to Assumption \ref{assump:SNR1}, we have 
    $$\|d_{\lambda}\bu^*(\bE_{\lambda}\bz)' + d_{\lambda}\bE_{\lambda}\bz\bu^{*\prime}+ \bE_{\lambda}\bE_{\lambda}' - \bbE(\bE_{\lambda}\bE_{\lambda}')\|\leq\frac{1}{2}d_{\lambda}^2,
     $$ 
    and hence applying the Davis-Kahan's theorem gives us the following bound with probability at least $1-CN^{-c}$:
    \begin{equation*}
    \begin{split}
        \sin\theta(\bu^*,\bu^{(0)})\leq &\frac{C\sigma\left(\sqrt{N} + (N(p^2+q^2))^{\frac{1}{4}}\right)\sqrt{\log N}}{d_{\lambda}}\\
        \leq &\frac{\sqrt{7}}{3}\\
        \leq &\min\left\{\sqrt{1-\frac{32r_x\|\cE_x\|_{\op}^2}{\|\bD_x^*\|_F^2}}, \sqrt{1-\frac{32r_y\|\cE_y\|_{\op}^2}{\|\bD_y^*\|_F^2}}\right\},
    \end{split}
    \end{equation*}
    as long as the constant $C>0$ in Assumption \ref{assump:dJisst_SNR1} is sufficiently large. Therefore, under Assumption \ref{assump:dJisst_SNR1}, with probability at least $1-CN^{-c}$,  the following holds for $k\geq 1$: 
    \begin{equation*}
        \begin{split}
            |\sin\theta(\bu^*,\bu^{(k)})|\leq &\frac{C\sigma\left(\lambda\sqrt{r_x(p+N)}\|\bD_x^*\|_F + (1-\lambda)\sqrt{r_y(q+N)}\|\bD_y^*\|_F\right)}{\lambda\|\bD_x^*\|_F^2 + (1-\lambda)\|\bD_y^*\|_F^2}\\
            \leq &\frac{C\sigma\sqrt{r_x(p+N)}}{\|\bD_x^*\|_F}\vee\frac{C\sigma\sqrt{r_y(q+N)}}{\|\bD_y^*\|_F}.
        \end{split}
    \end{equation*}

    Finally, when Assumption \ref{assump:dJisst_SNR2} also holds, we can apply the spectral norm bounds of $\cE_x$, $\cE_y$ again to show that Condition \ref{cond:SNR2} holds with probability at least $1-4\exp\{-N\}$. Then with probability at least $1-CN^{-c}$, the following holds for $k\geq 1$: $$\|\sin\Theta(\bV^*,\bV^{(k+1)})\|\leq \frac{C\sigma\sqrt{p+N}}{\sigma_{r_x}(\bD_x^*)},\quad \|\sin\Theta(\bW^*,\bW^{(k+1)})\|\leq \frac{C\sigma\sqrt{q+N}}{\sigma_{r_y}(\bD_y^*)}.$$
    
    The proof of Theorem \ref{thm:converg_D} is now complete. In the end of this section, we present the proof of the technical lemma used earlier.
    \begin{proof}[Proof of Lemma \ref{lem:perturb_frobenius_err}]
        Since $\bX_r-\bX^*$ is of rank at most $2r$, we have
        \begin{equation*}
            \|\bX_r-\bX^*\|_F\leq \sqrt{2r}\|\bX_r-\bX^*\|\leq \sqrt{2r}(\|\bX_r-\bX\| + \|\bX^*-\bX\|).
        \end{equation*}
        Meanwhile, since $\bX_r$ is also the best rank-$r$ approximation of $\bX$ in terms of spectral norm error, we can write
        \begin{equation*}
            \|\bX_r-\bX^*\|_F\leq \sqrt{2r}(\|\bX_r-\bX\| + \|\bX^*-\bX\|)\leq 2\sqrt{2r}\|\bX^*-\bX\|=2\sqrt{2r}\|\bE\|.
        \end{equation*}
    \end{proof}

\section{Additional Results for Real Data Analysis}
\label{supp:realdata}

Figure \ref{fig:mean_sc_fc} shows the mean SC and FC from the 1058 subjects. 

\begin{figure}
    \centering
    \includegraphics[width=0.7\textwidth]{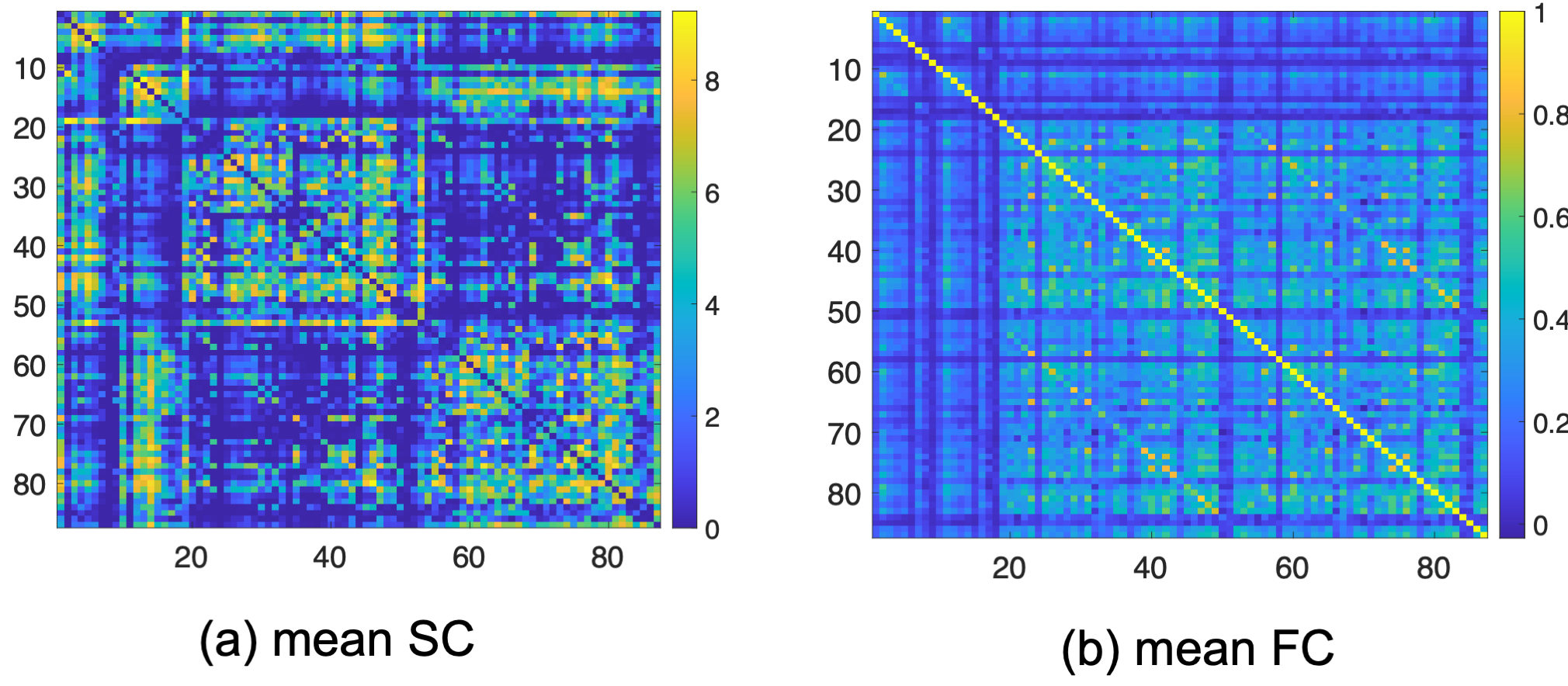}
    \caption{Mean SC (left) and FC (right) from the 1058 HCP subjects.}
    \label{fig:mean_sc_fc}
\end{figure}

In Figure~\ref{fig:gjisstpca_correlation} panels (a) and (b), we correlate \( \boldsymbol{u}^{*}_1 \) and \( \boldsymbol{u}^{*}_2 \) with the 45 cognitive traits. In these plots, traits encoded in the opposite direction (e.g., higher values indicate worse cognitive ability) are colored in pink. We observe that 1) \( \boldsymbol{u}^{*}_1 \) correlates better with behavioral traits than does \( \boldsymbol{u}^{*}_2 \), and 2) most behavioral traits show a decent amount of correlation with the joint factors. Finally, we examine how well we can predict these behavioral traits using both \( \boldsymbol{u}^{*}_1 \) and \( \boldsymbol{u}^{*}_2 \). Panel (c) displays the correlations between the predicted and measured behavioral traits in test datasets, based on a simple linear regression model. These results are computed as the average over 50 runs of 80-20 random splitting of training and testing datasets.


\begin{figure}[!t]
    \centering
    \includegraphics[width = \textwidth]{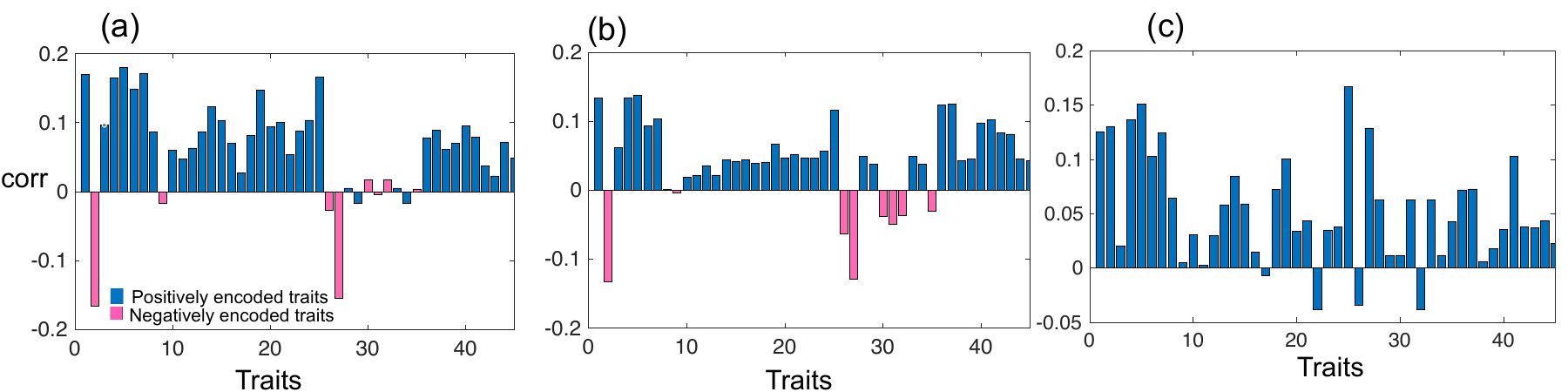}
    \caption{We correlate \( \boldsymbol{u}^{*}_1 \) (in panel a) and \( \boldsymbol{u}^{*}_2 \) (in panel b) with the 45 cognitive traits. Traits encoded in the opposite direction (e.g., higher values indicate worse cognitive ability) are colored in pink.  Panel (c) displays the correlations between the predicted and measured behavioral traits in test datasets, based on a simple linear regression model. These results are computed as the average over 50 runs of 80-20 random splitting of training and testing datasets.}
    \label{fig:gjisstpca_correlation}
\end{figure}

Figure \ref{fig:sc_fc_pred} compares the prediction power of PCA scores with G-JisstPCA scores. We first apply the principal
components analysis (PCA) to vectorized SC and FC matrices separately to extract the first two PC scores. We then treat the two PC scores as predictors and use simple linear
regression to predict the 45 behavior traits. Panel (a) shows results based on the PCA with the FC data, panel (b) shows results based on the PCA with the SC data, and panel (c) shows results based on JisstPCA with both the SC and FC data.

\begin{figure}[!t]
    \centering
    \includegraphics[width = \textwidth]{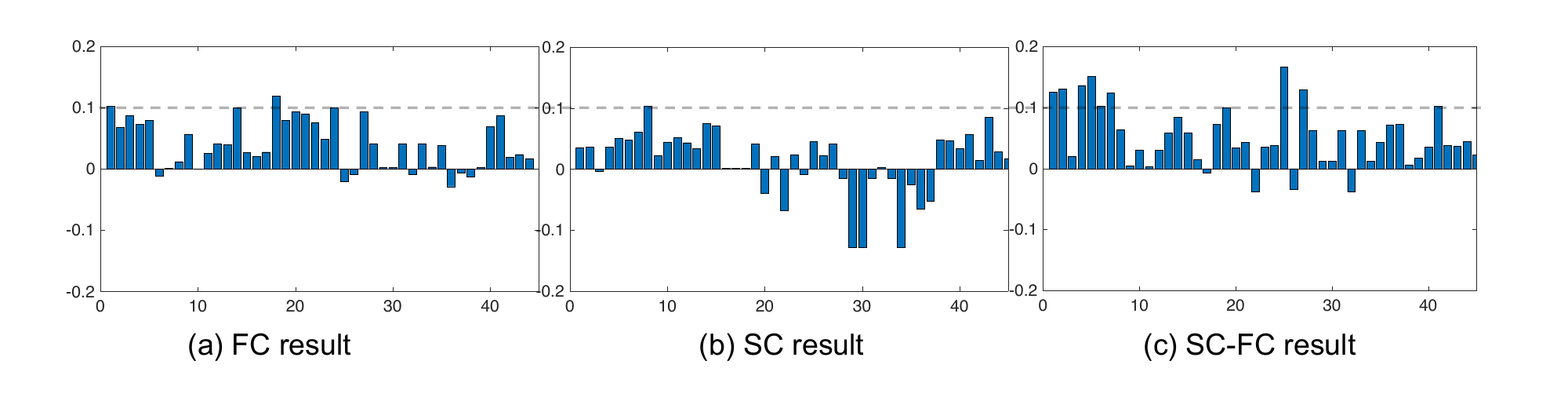}
    \caption{Prediction results using the first two PC scores extracted with different methods. Panel (a) shows results based on the PCA with the FC data, panel (b) shows results based on the PCA with the SC data, and panel (c) shows results based on JisstPCA with both the SC and FC data.}
    \label{fig:sc_fc_pred}
\end{figure}

Figure \ref{fig:hcp_result_set2} shows the adjacency matrices for the network loadings for SC (\{$\boldsymbol{V}^*_1 \boldsymbol{D}_{SC,1}^*\boldsymbol{V}_1^{*\prime}$,$\boldsymbol{V}^*_2 \boldsymbol{D}_{SC,2}^*\boldsymbol{V}_2^{*\prime}$ \}) and FC ($\{\boldsymbol{W}^*_1 \boldsymbol{D}_{FC,1}^*\boldsymbol{W}_1^{*\prime}$, $\boldsymbol{W}^*_2 \boldsymbol{D}_{FC,2}^*\boldsymbol{W}_2^{*\prime}\}$). The first 19 ROIs (rows of the adjacency matrix) are subcortical regions and the next 68 ROIs are cortical regions from the Desikan-Killiany atlas. Figure 5 in the main paper shows circular plots with top 200 connections. 

 From the correlation plots in panel (a) of Figure \ref{fig:gjisstpca_correlation}, it is evident that traits of fluid intelligence assessment, English reading and vocabulary comprehension, and line orientation test have high ($r>0.15$) and positive correlations with $\boldsymbol{u}_1^*$. The corresponding loadings in the circular plots highlight major subcortical to cortical SC pathways, such as those from the putamen to the frontal lobe and insula, as well as from the thalamus to the parietal lobe. Additionally, several significant cross-hemisphere FC pathways are noted, including those from the left parietal lobe to the right parietal lobe. The positive nature of these pathways in the loadings indicates that higher SC and FC connections correlate with enhanced cognitive abilities. Similarly, the second JisstPCA score shows high and positive correlations with cognitive traits. Although some negative connections are observed in the loadings, positive connections are predominant.

\begin{figure}
    \centering
    \includegraphics[width=\textwidth]{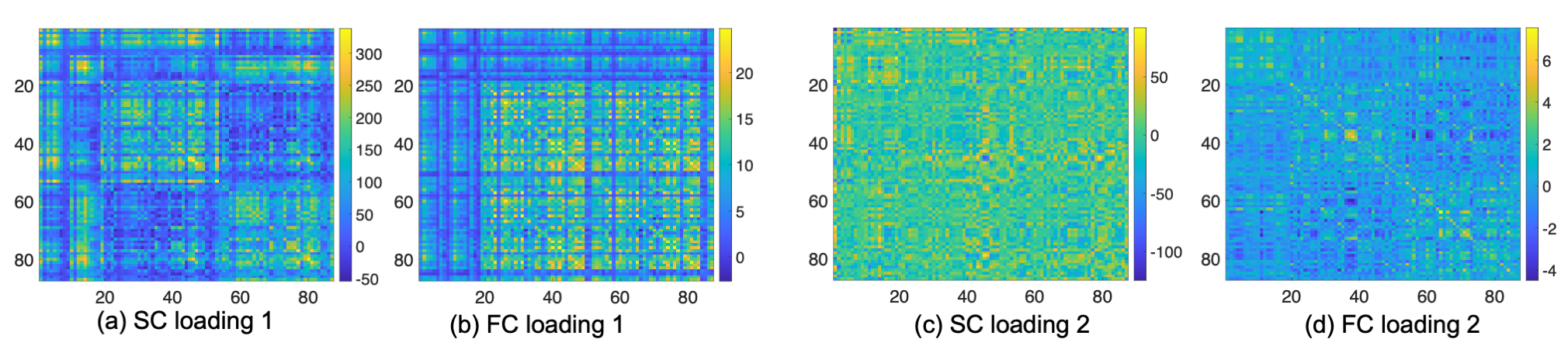}
    \caption{The first two JisstPCA loadings. Panel (a) shows $\boldsymbol{V}_1\boldsymbol{D}_{SC,1}\boldsymbol{V}_1^\prime$, panel (b) shows $\boldsymbol{W}_1\boldsymbol{D}_{FC,1}\boldsymbol{W}_1^\prime$, panel (c) shows $\boldsymbol{V}_2\boldsymbol{D}_{SC,2}\boldsymbol{V}_2^\prime$, panel (b) shows $\boldsymbol{W}_2\boldsymbol{D}_{FC,2}\boldsymbol{W}_2^\prime$  }
    \label{fig:hcp_result_set2}
\end{figure}
\bibliographystyle{apalike}
\bibliography{reference.bib}
\end{document}